\documentclass[a4paper,onecolumn,11pt,accepted=2026-06-23]{quantumarticle}
\pdfoutput=1
\usepackage{soul}
\usepackage[T1]{fontenc}
\usepackage{amssymb}
\usepackage{tikz}
\usetikzlibrary{tikzmark}
\usepackage{amsthm}
\usepackage{amsmath}
\usepackage{stmaryrd}
\usepackage{physics}
\usepackage{braket}
\usepackage{dsfont}
\usepackage{bbold}
\usepackage{graphicx}
\usepackage{relsize}
\usepackage{makecell}
\usepackage[colorlinks = true, linkcolor = blue, urlcolor  = blue, citecolor = red, backref=page]{hyperref}
\usepackage[margin=1in]{geometry}
\usepackage{enumitem}
\usepackage{mathtools}
\usepackage{graphbox}
\usepackage{comment}
\usepackage{float}
\usepackage{standalone}
\usepackage[capitalize]{cleveref}
\usepackage{booktabs}
\usepackage[skins]{tcolorbox}
\usepackage{nicematrix}
\usepackage{multirow}
\usepackage{tikz-cd}
\usepackage[table]{xcolor}
\usepackage{booktabs, tabularx}
\usepackage{tcolorbox}
\usepackage[dvipsnames]{xcolor}
\usetikzlibrary{arrows.meta, matrix, calc, fit, backgrounds, positioning} 
\usepackage{hhline}
\usepackage{caption} 
\usepackage{makecell}
\usepackage{colortbl}
\NiceMatrixOptions{cell-space-limits = 2pt}

\newcommand{\C}[1]{\mathbb{C}^{#1}}
\newcommand{\M}[1]{\mathcal{M}_{#1}}

\renewcommand{\epsilon}{\varepsilon}
\renewcommand{\phi}{\varphi}
\newcommand{\overbar}[1]{\mkern 1.5mu\overline{\mkern-1.5mu#1\mkern-1.5mu}\mkern 1.5mu}

\newcommand{\diagU}{\mathcal{DU}}

\newcommand{\upwardsubseteq}{\rotatebox{90}{\(\subseteq\)}}

\newcommand{\upwardeq}{\rotatebox{90}{\(=\)}}

\newcommand{\homo}[2]{\mathbb{H}[x]_{#1}^{#2}}
\newcommand{\knd}[2]{\mathbb{N}^{(#1)}_{#2}}

\newtheorem{theorem}{Theorem}[section]

\newtheorem{proposition}[theorem]{Proposition}
\newtheorem{corollary}[theorem]{Corollary}
\newtheorem{lemma}[theorem]{Lemma}

\newtheorem{definition}[theorem]{Definition}

\theoremstyle{definition}
\newtheorem{remark}[theorem]{Remark}
\newtheorem{example}[theorem]{Example}

\newcommand\vertarrowbox[3][6ex]{
  \begin{array}[t]{@{}c@{}} #2 \\
  \left\uparrow\vcenter{\hrule height #1}\right.\kern-\nulldelimiterspace\\
  \makebox[0pt]{\scriptsize#3}
  \end{array}
}

\def\mybigplus{\mathop{\mathchoice{
   \vcenter{\hbox to12bp{\vrule height15bp width0pt \pdfliteral{
   q 1 J .8 w 0 7.5 m 12 7.5 l S 6 1 m 6 14 l S Q
}\hss}}}{
   \vcenter{\hbox to12bp{\kern1bp\vrule height10bp width0pt \pdfliteral{
   q 1 J .65 w 0 5 m 10 5 l S 5 0 m 5 10 l S Q
}\hss}}}{+}{+}
}}

\newcommand{\idx}{\mathbb{i}}

\makeatletter
\def\smallunderbrace#1{\mathop{\vtop{\m@th\ialign{##\crcr
   $\hfil\displaystyle{#1}\hfil$\crcr
   \noalign{\kern3\p@\nointerlineskip}%
   \tiny\upbracefill\crcr\noalign{\kern3\p@}}}}\limits}
\makeatother

\title{Entanglement in the Dicke subspace}

\author{Aabhas Gulati}
\email{aabhas.gulati@math.univ-toulouse.fr}
\address{Institut de Mathématiques, Université de Toulouse, UPS, France.}

\author{Ion Nechita}
\email{nechita@irsamc.ups-tlse.fr}
\address{Laboratoire de Physique Th\'eorique, Universit\'e de Toulouse, CNRS, UPS, France}

\author{Clément Pellegrini}
\email{clement.pellegrini@math.univ-toulouse.fr}
\address{Institut de Mathématiques, Université de Toulouse, UPS, France.}

\begin{document}

\begin{abstract}
In this paper, we provide a complete mathematical theory for the entanglement of mixtures of Dicke states. These quantum states form an important subclass of bosonic states arising in the study of indistinguishable particles. We introduce a tensor-based parametrization where the diagonal entries of these states are encoded as a symmetric tensor, enabling a direct translation between entanglement properties and well-studied convex cones of tensors. Our results bridge multipartite entanglement theory with semialgebraic geometry and the theory of completely positive and copositive tensors.

This dictionary maps separability to completely positive tensors, the PPT property to moment tensors, entanglement witnesses to copositive tensors, and decomposable witnesses to sum of squares tensors. Using this framework, we construct explicit PPT entangled states in three or more qutrits, disproving a recent conjecture. 
We establish that PPT entanglement exists for all multipartite systems with local dimension $d \geq 3$ and $n \geq 3$ parties. We also show that, for mixtures of Dicke states, the PPT condition with respect to the most balanced bipartition implies all other PPT conditions. 

We further connect bosonic extendibility of mixtures of Dicke states to the duals of known hierarchies for non-negative polynomials, such as the ones by Reznick and Polya. We thus provide semidefinite programming relaxations for separability and entanglement testing in the Dicke subspace.
\end{abstract}

\maketitle

\vspace{-.5cm}
\tableofcontents


\section{Introduction}
The characterization of entanglement in mixed quantum states is an important challenge in quantum information theory. Entanglement is a central resource for various tasks in quantum communication \cite{bennett1993teleporting}, cryptography \cite{Ekert1991crypto}, and even quantum computing \cite{Raussendorf2001computer,Raussendorf2003computer}. Although being an important question in quantum information, the separability problem, i.e deciding if a bipartite mixed quantum state is entangled or separable has been shown to be NP-hard \cite{gurvits2003classical,gharibian2010strong}. This implies the efficient algorithmic solution of the problem is out of reach, unless P=NP. Understanding entanglement in multipartite quantum systems is even more difficult, as there exist different notions of entanglement. In this setting, \emph{genuine multipartite entanglement} underpins protocols such as secret sharing and multiparty teleportation \cite{hillery1999quantum, chen2006general}.

In general, different criteria for detecting entanglement have been constructed. The most famous of them is the positivity under partial transpose (PPT) criterion: any unentangled (or separable) bipartite state remains positive after transposition is applied to either subsystem \cite{peres1996separability, Horodecki1997PPTent}. This criterion is also sufficient in $2 \times 2$, or $2 \times 3$ systems \cite{horodecki1996separability}; entangled states satisfying the PPT condition exist in all other dimensions. These states display some exotic properties; they cannot be distilled into pure entangled states and are central to the irreversibility of entanglement theory \cite{Vidal2001irreversible}. Moreover, even finding new constructions of such states has been observed to be highly challenging. Various open problems and conjectures surrounding this class of states remain open, for example, the PPT-squared conjecture \cite{PPTsq, Christandl2018} and the maximum Schmidt number of PPT entangled states \cite{huber2018high,pal2019class,krebs2024high}.

The intrinsic indistinguishability of particles is a distinctly quantum mechanical phenomenon. In systems composed of such particles, the principles of quantum mechanics require that the particles be described not by individual quantum states, but rather by a single collective state encompassing the entire system. Moreover, the structure of this collective state depends fundamentally on whether the particles are \emph{bosons} or \emph{fermions}. Specifically, bosonic wavefunctions are symmetric under particle exchange, whereas fermionic wavefunctions are antisymmetric. Importantly, the quantum description of indistinguishable particle systems often involves a substantially reduced number of independent parameters, thereby simplifying the theoretical treatment of many problems in quantum information science.

There has been considerable progress in the understanding of multipartite entanglement in bosonic quantum systems. The bosonic states  display the extreme features of multipartite entanglement; as they are either completely separable or genuinely entangled across all systems \cite{ichikawa2008exchange}. Moreover, even in three qubits, there are no PPT entangled states \cite{eckert2002quantum}. For all $4 \leq n \leq 23$, $n$-qubit examples of quantum states satisfying the PPT condition for every bi-partition, and still being entangled, were obtained \cite{augusiak2012entangled}. Also, the existence of 2-qutrit PPT entangled bosonic states was shown \cite{PhysRevLett.102.170503}. Surprisingly, in the bosonic subspace, the bipartite entanglement criterion, such as covariance and realignment, coincide with applying the PPT test \cite{PhysRevLett.102.170503}.

In this article, we focus on a particular subclass of bosonic states called \emph{mixtures of Dicke states}. Experimentally, Dicke states have been demonstrated in photonic
systems~\cite{Kiesel2007,Wieczorek2009,Prevedel2009,Krischek2011},
trapped ions~\cite{Hume2009}, and cold atomic
ensembles~\cite{Luecke2011,Luecke2014}, in correspondence with the theoretical analysis of their entanglement properties \cite{toth2007detection,TothNJP2009,Duan2011,bergmann2013entanglement}. It was shown that the theory of bipartite entanglement in mixtures of Dicke states is connected to well-studied cones in polynomial optimization, the completely positive and copositive matrices \cite{yu2016separability, tura2018separability}. The NP-hardness of deciding whether a matrix is completely positive implies a corresponding complexity of the separability problem. Moreover, the PPT condition for these states is also in exact correspondence with a well-studied relaxation for completely positive matrices, known as \emph{doubly non-negative matrices} \cite{yu2016separability, tura2018separability}. The doubly non-negative and completely positive cones coincide if and only if $d \leq 4$, and therefore, no PPT entangled states exist for this regime in the bipartite setting.  Construction of PPT entangled states in multipartite states for $d \geq 5$ can be achieved by considering bipartite marginals that are PPT entangled \cite{romero2025multipartite}. At the same time, it was shown that for any $n$-qubit state in multipartite \emph{mixtures of Dicke states}, there is no PPT entanglement \cite{yu2016separability, gulatiposext}. The class of mixture of Dicke states also enjoys the local \emph{diagonal unitary symmetry}, recently studied in \cite{singh2021diagonal}. Recently, a strong conjecture was proposed in a recent work \cite{romero2025multipartite} that for a multi-partite mixture of Dicke states in $d=3$ and $d=4$, all separable states are PPT. Overall, the situation seemed far from understood, and a rigorous and well-understood mathematical framework to explicitly describe PPT and indecomposable maps was lacking. This is a central point for the questions in this paper :  

\begin{enumerate}
\renewcommand{\labelenumi}{\boxed{{\mathrm{Q}}_\arabic{enumi}.}}
\setlength{\itemsep}{.5em}
    \item What is the appropriate parametrization and tools for the \emph{mixtures of Dicke states} to properly utilise the reduced number of parameters? 
    \item Can we characterize separability of the class of diagonally symmetric bosonic states using some generalization of completely positive matrices? Can we do something similar for the entanglement witnesses with the same symmetries? 
    \item Can we characterise the PPT condition for such a given state? Is the PPT criterion also sufficient for separability in $d=3,4$ as conjectured in \cite{romero2025multipartite}? Can we construct new examples of PPT entangled states and indecomposable witnesses?
\end{enumerate}

This article provides a complete answer to all these questions by translating the entanglement properties of these states into well-known objects in tensor analysis and semi-algebraic geometry. The \emph{dramatis personae} of this article are sum of square polynomials and tensors. Deciding whether a polynomial is globally non-negative is a central and computationally difficult problem in optimization theory. A tractable sufficient certificate is \emph{sum-of-squares} (SOS): if $p = \sum^K_{k=1} f_k^2$, then $p(x) \ge 0$ for all $x \in \mathbb{R}^d$. Hilbert famously characterized the exceptional cases where non-negativity coincides with SOS \cite{hilbert1888darstellung}; outside these regimes, there exist positive polynomials that are not SOS (e.g., the Motzkin/Robinson examples \cite{Motzkin1967,robinson1973some}). A canonical example in three variables is
$$p(x,y,z)= x^2y^4 + x^4y^2 + z^6 - 3x^2 y^2 z^2,$$
which satisfies $p\ge 0$ but is not SOS. 
To systematically tighten SOS certificates, one can use Positivstellens\"atze: for instance, test whether $p(x)\lVert x\rVert^{2r}$ is SOS for some $r$. This yields a semidefinite hierarchy for polynomial positivity; by a theorem of Reznick, every positive definite form is detected at some finite level \cite{reznick1995uniform}. These ideas of positive polynomials and certificates of positivity have been recently used to detect non-classicality of light and spin systems \cite{ohst2025revealing}.

\subsection{Contributions}

We present below, in an informal way, the main ideas and results in this article.

\begin{itemize}[leftmargin=2em]\setlength\itemsep{.5em}
    \item \underline{Tensor-based parametrization of diagonally symmetric subspace}: We introduce a novel way to parameterize the diagonally symmetric subspace in terms of symmetric tensors in \cref{thm:tensor-param}, and \cref{def:parametrize-witness}. This tensor can be understood as an array consisting of the diagonal entries of the state. 
    
    \item \underline{Complete dictionary for entanglement properties}: By using the new tensor-based para\-metrisation, we translate the entanglement properties of diagonally symmetric states into the well-studied properties of tensors. We establish a correspondence between separability and completely positive tensors (\cref{thm:sep-and-cp}), and on the dual side, entanglement witnesses and copositive tensors (\cref{prop:ew-cop}). We also show the characterization of the PPT property in terms of the flattenings of the even slices of the tensors (\cref{thm:ppt-mom}). Moreover, the dual property of decomposability corresponds to the sum of squares relaxation of copositive tensors (\cref{thm:dec-iff-sos}). Our findings are summarized in \cref{fig:entanglement-DS-tensors}.

    \item \underline{The balanced bipartition gives the strongest PPT criterion}: For a mixture of Dicke states on $n$ particles, we show that the PPT property across the most balanced bipartition (the one that transposes the first $\lfloor n/2 \rfloor$ particles, leaving the other $\lceil n/2 \rceil$ unchanged) is the most stringent one, in the sense that it implies the PPT property for all the other bipartitions (\cref{cor:PPT-balanced-strongest}); this settles in the affirmative \cite[Conjecture III.1]{romero2025multipartite}.
    
    \item \underline{PPT entangled states and indecomposable witnesses}: We use some highly celebrated results in polynomial SOS to construct indecomposable witnesses, and to detect examples of PPT entangled states. We show that, except for qubits and bipartite systems with $d \leq 4$, all the other cases contain examples of PPT entangled states (\cref{thm:qutrits-and-beyond}). Moreover, these states are PPT across all bi-partitions. This also disproves the conjecture in \cite{romero2025multipartite}, i.e there exist PPT entangled states in $3$ qutrits, and beyond. 
    
    \item \underline{Bosonic Extendibility}: We introduce the notion of the hypergraph (PPT) extendibility of bosonic states. Due to symmetry properties, it follows that all diagonally symmetric states are extendable if and only if there is a diagonally symmetric extension (\cref{lem:extension-also-symmetric}). Moreover, the marginals of DS states correspond to the marginals of the  corresponding symmetric tensor (\cref{thm:marginals-dicke-states}). This shows that the extendibility of such states is equivalent to the classical extendibility problem for exchangeable probability distributions (\cref{thm:extendibility-dicke-states}). On the dual side, this implies that the non-negative and sum of squares hierarchy corresponds to the set of extendibility witnesses (\cref{cor:extendibility-witnesses}). 
    
    \item \underline{Dicke subsystems are entangled}: We use the tensor parametrization to compute reduced states of pure Dicke states. We provide a much simpler proof of the NPT property of all marginals (across all bipartitions). We show that all the $2$-body marginals of pure entangled Dicke states are entangled (\cref{prop:2-body-marginals}). This also reproduces the recently obtained results in \cite{szalay2025dicke}.
\end{itemize}

We recommend the reader to see the picture \cref{fig:full-diagram} for the full picture of correspondences, inclusions, and dualities presented in this article. The notations for the acronyms can be found in \cref{tbl:bosonic-matrices} and \cref{tbl:symmetric-tensors}.

\begin{figure}[htbp]
\makebox[\textwidth][c]{
\scalebox{0.8}{
\begin{tikzpicture}[
    >=Latex,
    mathnode/.style={inner sep=2pt, minimum height=0.8cm, font=\Large},
    hsubset/.style={font=\Large, inner sep=2pt},
    vsubset/.style={rotate=90, font=\Large, inner sep=2pt},
    arrow/.style={<->, line width=0.8pt},
    slanted/.style={inner sep=1pt, font=\Large},
    hierarchyBox/.style={
        draw=black!30, 
        dashed, 
        rounded corners, 
        inner sep=4pt
    },
    majorBox/.style={
        draw=black!50, 
        dotted, 
        line width=1.5pt,
        rounded corners=10pt, 
        inner sep=15pt,
        fill=gray!5,
        fill opacity=0.2
    },
    horizBox/.style={
        draw=#1!60, 
        dashed, 
        line width=1.8pt,
        rounded corners=15pt, 
        inner sep=10pt,
        fill=#1!15,
        fill opacity=0.4
    },
    boxLabel/.style={
        font=\small\sffamily\bfseries, 
        text=black!60, 
        inner sep=2pt
    },
    majorLabel/.style={
        font=\Large\sffamily\bfseries, 
        text=black!80, 
        inner sep=5pt
    }
]

    \matrix (leftM) [matrix of nodes, column sep=0.2cm, row sep=1.5cm, nodes={mathnode}, ampersand replacement=\&] {
        \& |[name=CP]| $\mathsf{CP}_d^{(n)}$ \& \\
        |[name=EWPExt]| $\mathsf{NNExt}_d^{(n,r)}$ \& \& |[name=MomExt]| $\mathsf{MomExt}_d^{(n,r)}$ \\
        |[name=dots1L]| $\vdots$ \& \& |[name=dots1R]| $\vdots$ \\
        |[name=EWP_L]| $\mathsf{NNExt}_d^{(n)}$ \& \& |[name=Mom]| $\mathsf{MomExt}_d^{(n,0)}$ \\
        \& |[name=EWP_C]| $\mathsf{NN}_d^{(n)}$ \& \\
        |[name=EWP_R]| $\mathsf{PNN}_d^{(n,0)}$ \& \& |[name=SOS]| $\mathsf{RSOS}_d^{(n,0)}$ \\
        |[name=dots2L]| $\vdots$ \& \& |[name=dots2R]| $\vdots$ \\
        |[name=PNN]| $\mathsf{PNN}_d^{(n,r)}$ \& \& |[name=RSOS]| $\mathsf{RSOS}_d^{(n,r)}$ \\
        \& |[name=Cop]| $\mathsf{Cop}_d^{(n)}$ \& \\
    };

    \matrix (rightM) [right=4cm of leftM, matrix of nodes, column sep=0.2cm, row sep=1.5cm, nodes={mathnode}, ampersand replacement=\&] {
        \& |[name=Sep]| $\mathsf{Sep}_d^{(n)}$ \& \\
        |[name=BExt]| $\mathsf{BExt}_d^{(n,r)}$ \& \& |[name=PPTBExt]| $\mathsf{PPTBExt}_d^{(n,r)}$ \\
        |[name=dots3L]| $\vdots$ \& \& |[name=dots3R]| $\vdots$ \\
        |[name=PSD_L]| $\mathsf{PSD}_d^{(n)}$ \& \& |[name=PPT]| $\mathsf{PPT}_d^{(n)}$ \\
        \& |[name=PSD_C]| $\mathsf{PSD}_d^{(n)}$ \& \\
        |[name=PSD_R]| $\mathsf{ExtW}_d^{(n,0)}$ \& \& |[name=Dec]| $\mathsf{DecExtW}_d^{(n,0)}$ \\
        |[name=dots4L]| $\vdots$ \& \& |[name=dots4R]| $\vdots$ \\
        |[name=PSDExt]| $\mathsf{ExtW}_d^{(n,r)}$ \& \& |[name=DecExt]| $\mathsf{DecExtW}_d^{(n,r)}$ \\
        \& |[name=EW]| $\mathsf{EW}_d^{(n)}$ \& \\
    };

    \begin{scope}[on background layer]
        
        \node[horizBox=blue, fit=(CP) (Sep) (EWPExt) (MomExt) (BExt) (PPTBExt) (EWP_L) (Mom) (PSD_L) (PPT)] (topHoriz) {};
        \node[majorLabel, left=30pt, rotate=90, anchor=south, name=Pri] at (topHoriz.west) {Primal};

        \node[horizBox=orange, fit=(Cop) (EW) (EWP_R) (SOS) (PSD_R) (Dec) (PNN) (RSOS) (PSDExt) (DecExt)] (botHoriz) {};
        \node[majorLabel, left=30pt, rotate=90, anchor=south, name=Dua] at (botHoriz.west) {Dual};

        \node[majorBox, fit=(CP) (Cop) (EWPExt) (PNN) (MomExt) (Mom)] (bigBoxLeft) {};
        \node[majorLabel, above=15pt, name=Sym] at (bigBoxLeft.north) {Symmetric Tensors};

        \node[majorBox, fit=(Sep) (EW) (BExt) (DecExt)] (bigBoxRight) {};
        \node[majorLabel, above=15pt, name=Ent] at (bigBoxRight.north) {Entanglement Theory};

        \node[hierarchyBox, fit=(EWPExt) (dots1L) (EWP_L)] (box1) {};
        \node[boxLabel,above] at (box1) {EWP Hierarchy};
        \node[hierarchyBox, fit=(MomExt) (dots1R) (Mom)] (box2) {};
        \node[boxLabel, above] at (box2) {Mom Hierarchy};
        \node[hierarchyBox, fit=(EWP_R) (dots2L) (PNN)] (box3) {};
        \node[boxLabel, above] at (box3) {Polya's hierarchy};
        \node[hierarchyBox, fit=(SOS) (dots2R) (RSOS)] (box4) {};
        \node[boxLabel, above] at (box4) {Reznick's hierarchy};
        \node[hierarchyBox, fit=(BExt) (dots3L) (PSD_L)] (box5) {};
        \node[boxLabel, above] at (box5) {B-Ext};
        \node[hierarchyBox, fit=(PPTBExt) (dots3R) (PPT)] (box6) {};
        \node[boxLabel, above] at (box6) {PPT-Ext};
        \node[hierarchyBox, fit=(PSD_R) (dots4L) (PSDExt)] (box7) {};
        \node[boxLabel, above] at (box7) {Dual Ext};
        \node[hierarchyBox, fit=(Dec) (dots4R) (DecExt)] (box8) {};
        \node[boxLabel, above] at (box8) {Dual PPT-Ext};
    \end{scope}

    \node[vsubset] at ($(EWPExt)!0.5!(dots1L)$) {$\supseteq$};
    \node[vsubset] at ($(dots1L)!0.5!(EWP_L)$) {$\supseteq$};
    \node[vsubset] at ($(MomExt)!0.5!(dots1R)$) {$\supseteq$};
    \node[vsubset] at ($(dots1R)!0.5!(Mom)$) {$\supseteq$};
    \node[vsubset] at ($(EWP_R)!0.5!(dots2L)$) {$\supseteq$};
    \node[vsubset] at ($(dots2L)!0.5!(PNN)$) {$\supseteq$};
    \node[vsubset] at ($(SOS)!0.5!(dots2R)$) {$\supseteq$};
    \node[vsubset] at ($(dots2R)!0.5!(RSOS)$) {$\supseteq$};
    \node[hsubset] at ($(MomExt)!0.5!(EWPExt)$) {$\supseteq$};
    \node[hsubset] at ($(Mom)!0.5!(EWP_L)$) {$\supseteq$};
    \node[hsubset] at ($(SOS)!0.5!(EWP_R)$) {$\subseteq$};
    \node[hsubset] at ($(RSOS)!0.5!(PNN)$) {$\subseteq$};
    \node[slanted, rotate=45] at ($(CP)!0.5!(EWPExt)$) {$\supseteq$};
    \node[slanted, rotate=-45] at ($(CP)!0.5!(MomExt)$) {$\subseteq$};
    \node[slanted, rotate=-45] at ($(EWP_L)!0.5!(EWP_C)$) {$=$};
    \node[slanted, rotate=45] at ($(Mom)!0.5!(EWP_C)$) {$\supseteq$};
    \node[slanted, rotate=45] at ($(EWP_C)!0.5!(EWP_R)$) {$=$};
    \node[slanted, rotate=-45] at ($(EWP_C)!0.5!(SOS)$) {$\subseteq$};
    \node[slanted, rotate=-45] at ($(PNN)!0.5!(Cop)$) {$\subseteq$};
    \node[slanted, rotate=45] at ($(RSOS)!0.5!(Cop)$) {$\supseteq$};

    \node[vsubset] at ($(BExt)!0.5!(dots3L)$) {$\supseteq$};
    \node[vsubset] at ($(dots3L)!0.5!(PSD_L)$) {$\supseteq$};
    \node[vsubset] at ($(PPTBExt)!0.5!(dots3R)$) {$\supseteq$};
    \node[vsubset] at ($(dots3R)!0.5!(PPT)$) {$\supseteq$};
    \node[vsubset] at ($(PSD_R)!0.5!(dots4L)$) {$\supseteq$};
    \node[vsubset] at ($(dots4L)!0.5!(PSDExt)$) {$\supseteq$};
    \node[vsubset] at ($(Dec)!0.5!(dots4R)$) {$\supseteq$};
    \node[vsubset] at ($(dots4R)!0.5!(DecExt)$) {$\supseteq$};
    \node[hsubset] at ($(PPTBExt)!0.5!(BExt)$) {$\supseteq$};
    \node[hsubset] at ($(PPT)!0.5!(PSD_L)$) {$\supseteq$};
    \node[hsubset] at ($(Dec)!0.5!(PSD_R)$) {$\subseteq$};
    \node[hsubset] at ($(DecExt)!0.5!(PSDExt)$) {$\subseteq$};
    \node[slanted, rotate=45] at ($(Sep)!0.5!(BExt)$) {$\supseteq$};
    \node[slanted, rotate=-45] at ($(Sep)!0.5!(PPTBExt)$) {$\subseteq$};
    \node[slanted, rotate=-45] at ($(PSD_L)!0.5!(PSD_C)$) {$=$};
    \node[slanted, rotate=45] at ($(PPT)!0.5!(PSD_C)$) {$\supseteq$};
    \node[slanted, rotate=45] at ($(PSD_C)!0.5!(PSD_R)$) {$=$};
    \node[slanted, rotate=-45] at ($(PSD_C)!0.5!(Dec)$) {$\subseteq$};
    \node[slanted, rotate=-45] at ($(PSDExt)!0.5!(EW)$) {$\subseteq$};
    \node[slanted, rotate=45] at ($(DecExt)!0.5!(EW)$) {$\supseteq$};

    \draw[arrow] (Pri) -- node[midway, above, rotate=90] {dual coneity} (Dua);
    \draw[arrow] (Sym) -- node[midway, above] {1-to-1 correspondence} (Ent);

    \draw[arrow] (CP) -- node[midway, below] {\cref{thm:sep-and-cp}} (Sep);
    \draw[arrow] (EWP_C) -- node[midway, below] {\cref{thm:tensor-param}} (PSD_C);
    \draw[arrow] (Cop) -- node[midway, below] {\cref{prop:ew-cop}}(EW);

    \draw[arrow] (EWPExt) to [bend right=-12] node[midway, above] {\cref{thm:extendibility-dicke-states}} (BExt);
    \draw[arrow] (MomExt) to [bend right=-12] node[midway, above] {\cref{thm:extendibility-dicke-states}} (PPTBExt);
    \draw[arrow] (PNN) to [bend right=12] node[midway, below] {\cref{cor:extendibility-witnesses}} (PSDExt);
    \draw[arrow] (RSOS) to [bend right=12] node[midway, below] {\cref{cor:extendibility-witnesses}} (DecExt);
    \draw[arrow] (Mom) to [bend right=15] node[midway, below] {\cref{thm:ppt-mom}} (PPT);
     \draw[arrow] (Dec) to [bend right=15] node[midway, above] {\cref{thm:dec-iff-sos}} (SOS);
\end{tikzpicture}
}
}
    \caption{The different sets explored in this work, with their inclusion and duality structures. On the left-hand side, we have sets of symmetric tensors from $\vee^n \mathbb R^d$, and on the right-hand side, entanglement properties of the mixtures of Dicke states $\mathsf{DS}_d^{(n)}$. Every set in the diagram is in convex duality with the set that is symmetric with respect to the middle horizontal line.}
    \label{fig:full-diagram}
\end{figure}

\definecolor{tablegray}{gray}{0.95}

\begin{table}[htbp] 
\centering 
\renewcommand{\arraystretch}{1.8} 
\setlength{\tabcolsep}{8pt} 
\begin{tabularx}{\textwidth}{l l >{\raggedright\arraybackslash}X l} 
\noalign{\global\arrayrulewidth=1pt} 
\arrayrulecolor{black}
\hline 
\rowcolor{gray!40}  
\textbf{Notation} & \textbf{Name} & \textbf{Description} & \textbf{Reference} \\ 
\noalign{\global\arrayrulewidth=0.4pt} 
\arrayrulecolor{black}
  $\mathsf{Sep}^{(n)}_d$    & Bosonic separable  & $\operatorname{cone}\{\ketbra{v}{v}^{\otimes n} \mid v \in \mathbb{C}^d\}$  & \cref{eq:separability} \\  
\rowcolor{gray!15} 
$\mathsf{PPT}^{(n)}_d$   & Bosonic PPT  & $\{X \in \operatorname{Herm}[\vee^n \mathbb{C}^d] \mid \forall k \in [\lfloor n/2 \rfloor], \,  X^{\Gamma_{[k]}} \succeq 0\}$  & \cref{def:ppt-and-psd} \\  
$\mathsf{PSD}^{(n)}_d$   & Bosonic PSD  & $\{X \in \operatorname{Herm}[\vee^n \mathbb{C}^d] \mid X \succeq 0\}$  & \cref{def:ppt-and-psd} \\  
\rowcolor{gray!15} 
$\mathsf{Dec}^{(n)}_d$   & Decomposable  & $\{X \in \operatorname{Herm}[\vee^n \mathbb{C}^d] \mid \exists \{P_0, \ldots, P_{\lfloor n/2 \rfloor} \succeq 0 \mid X = \Pi^{(n)}_d \sum^{\lfloor n/2 \rfloor}_{k=0} P^{\Gamma_{[k]}}_k \Pi^{(n)}_d\}$  & \cref{cor:duality-ppt-dec} \\  
$\mathsf{EW}^{(n)}_d$   & Entanglement Witness & $\{X \in \operatorname{Herm}[\vee^n \mathbb{C}^d] \mid \forall v \in \mathbb{C}^d, \,  \braket{v^{\otimes n} \vert X \vert v^{\otimes n}} \geq 0\}$  & \cref{prop:ew-dual-sep} \\  
\rowcolor{gray!15} 
$\mathsf{BExt}^{(n,r)}_d$   & Bosonic Extendible  & $\{\operatorname{Tr}_{[r]}(X) \mid  X \in \mathsf{PSD}^{(n+r)}_d\}$  & \cref{eq:bext-sets} \\  
$\mathsf{PPTBExt}^{(n,r)}_d$   &  PPT Bosonic Extendible & $\{\operatorname{Tr}_{[r]}(X) \mid  X \in \mathsf{PPT}^{(n+r)}_d\}$  & \cref{eq:bext-sets} \\  
\rowcolor{gray!15} 
$\mathsf{ExtW}^{(n,r)}_d$   & Extendibility Witness & $\{X \in \operatorname{Herm}[\vee^{n} \mathbb{C}^d] \mid  \Pi^{(n+r)}_d (X \otimes I^{\otimes r}) \Pi^{(n+r)}_d \succeq 0\}$  & \cref{cor:ext-witness} \\  
$\mathsf{DecExtW}^{(n,r)}_d$   & Dec. Extendibility Witness & $\{X \in \operatorname{Herm}[\vee^{n} \mathbb{C}^d] \mid  \Pi^{(n+r)}_d (X \otimes I^{\otimes r}) \Pi^{(n+r)}_d \in \mathsf{Dec}^{(n+r)}_d\}$  & \cref{thm:ppt-ext-witness} \\  
\hline 
\end{tabularx} 
\caption{Notation, definition, and main properties for bosonic matrices.} 
\label{tbl:bosonic-matrices} 
\end{table}

\begin{table}[H]
\centering
\renewcommand{\arraystretch}{1.8}
\setlength{\tabcolsep}{10pt}
\begin{tabularx}{\textwidth}{l l >{\raggedright\arraybackslash}X l}
\noalign{\global\arrayrulewidth=1pt}
\arrayrulecolor{black}\hline
\rowcolor{gray!20} 
\textbf{Notation} & \textbf{Name} & \textbf{Description} & \textbf{Reference} \\
\noalign{\global\arrayrulewidth=0.4pt}
\arrayrulecolor{black}

$\mathsf{CP}^{(n)}_d$   
& Completely Positive
& $\operatorname{cone}\{v^{\otimes n} \mid v \in \mathbb{R}^d_+\}$ 
& \cref{def:cp-tensors} \\

\rowcolor{tablegray}
$\mathsf{Mom}^{(n)}_d$  
& Moment 
& $\{T \in \vee^n \mathbb{R}^d \mid \forall X \in \mathrm{SF}[T], \, X \succeq 0\}$ 
& \cref{eq:mom-definition} \\

$\mathsf{NN}^{(n)}_d$  
& Nonnegative
& $\{T \in \vee^n \mathbb{R}^d \mid \forall i \in [d]^n, \, T_i \geq 0\}$ 
& \cref{eq:definition-ewp} \\

\rowcolor{tablegray}
$\mathsf{SOS}^{(n)}_d$  
& Sum-of-Squares
& $\{T \in \vee^n \mathbb{R}^d \mid p_T (x \odot x) \in \Sigma^{(2n)}_d \}$ 
& \cref{def:sos-tensors} \\

$\mathsf{Cop}^{(n)}_d$  
& Copositive
& $\{T \in \vee^n \mathbb{R}^d \mid \forall x \in \mathbb{R}^d, \, p_T (x \odot x) \geq 0 \}$ 
& \cref{def:copositive-tensors} \\

\rowcolor{tablegray}
$\mathsf{RSOS}^{(n,r)}_d$  
& Reznick-SOS 
& $\{T \in \vee^n \mathbb{R}^d \mid p_T (x \odot x)\|x\|^{2r} \in \Sigma^{(2n+2r)}_d \}$ 
& \cref{eq:reznick-sos} \\

$\mathsf{PNN}^{(n,r)}_d$  
& Polya-NN
& $\{T \in \vee^n \mathbb{R}^d \mid p_T (x \odot x)\|x\|^{2r} \in \chi^{(2n+2r)}_d \}$ 
& \cref{eq:polyna-nn} \\

\rowcolor{tablegray}
$\mathsf{MomExt}^{(n,r)}_d$  
& Moment Extendible
& $\{\operatorname{tr}_{[r]}(X) \mid X \in \mathsf{Mom}^{(n+r)}_d\}$ 
& \cref{eq:mom-ext} \\

$\mathsf{NNExt}^{(n,r)}_d$  
& Non-Negative Extendible
& $\{\operatorname{tr}_{[r]}(X) \mid X \in \mathsf{NN}^{(n+r)}_d\}$ 
& \cref{eq:nn-ext} \\

\hline
\end{tabularx}
\caption{Notation, definition, and main properties for convex cones of symmetric tensors.}
\label{tbl:symmetric-tensors}
\end{table}

\subsection{Organization of the article}
The paper is organized as follows: In \cref{sec:prelims}, we provide some preliminary material to understand the article.  In \cref{sec:ds-subspace}, we discuss the central object of this article, the diagonally symmetric subspace, and focus on the quantum states in this subspace called the \emph{mixtures of Dicke states}. We discuss the parametrization of the diagonally symmetric bosonic subspace in terms of a symmetric tensor. In \cref{sec:entanglement-mods}, we discuss the entanglement properties of \emph{mixtures of Dicke states} by showing exact correspondence to properties of tensors. In \cref{sec:ppt-ent-ds-subspace}, we construct indecomposable entanglement witnesses using positive polynomials that do not have a sum of squares decomposition, like the Motzkin polynomial. We also provide some examples of PPT entangled states in $3$ qutrits, and beyond. In the last \cref{sec:bosonic-extendibility}, we discuss bosonic extendibility for states in the diagonally symmetric subspace, and show connections to non-negative and the sum of squares hierarchy for copositivity.

\section{Preliminaries and Background}
\label{sec:prelims}
\subsection{Notation}
We denote by $\mathbb{R}$ and $\mathbb{C}$ as the fields of real numbers and complex numbers, respectively. We denote by $\mathcal{M}_{d \times d'} (\mathbb{F})$, the set of $d \times d'$ matrices with entries in the field $\mathbb{F}$ which will be either $\mathbb{R}$ or $\mathbb{C}$. If $d = d'$, for ease of notation, we denote the set of matrices as just $\mathcal{M}_{d} (\mathbb{F}) := \mathcal{M}_{d \times d} (\mathbb{F})$. We use the Dirac bra-ket notation to denote vectors as $\ket{v} \in \mathbb{C}^{d}$, the dual vectors as $\bra{v} \in (\C{d})^*$, and the complex inner product $\braket{v,w} := \braket{v \mid w}$. We use the indexed vectors $\{\ket{i}\}^d_{i=1}$ for the computational basis. The operator $\ketbra{v}{w} : \C{d} \rightarrow \C{d}$ is a linear operator with the matrix representation $vw^*$ (or $vw^\dagger$ for physicists). By $S_n$, we denote the group of bijections from $[n] \rightarrow [n]$, called the symmetric group. For any finite set $X$, we use $|X|$ to denote the cardinality of the set. We use $\odot$ to denote the Hadamard product, i.e the entrywise product between vectors, $(v \odot w)_i = v_i w_i.$ We use $\operatorname{ran}(X)$ to denote the range of a matrix $X.$ We use $\operatorname{cl}$ to denote the closure of the set.  

\subsection{Symmetric subspace and Bosonic states}
\label{sec:symmetric-subspace}

\subsubsection{Symmetric subspace}
In this section, we provide a brief introduction to the symmetric subspace, which, while serving as the description of bosonic particles, also has a fundamental role in various other problems in quantum information. For a much more thorough quantum information-oriented introduction to this topic, the reader should refer to the excellent review \cite{harrow2013church}. 

For any element $\pi$ in the symmetric group $S_n$, we define the following permutation operator $P_\pi :(\mathbb{C}^d)^{\otimes n} \rightarrow (\mathbb{C}^d)^{\otimes n}$, 
$$P_\pi := \sum_{i_1, i_2 \ldots i_n \in [d]} \ketbra{i_{\pi^{-1}(1)}i_{\pi^{-1}(2)}  \ldots i_{\pi^{-1}(n)}}{i_1 i_2 \ldots i_n}.$$ Note that this is also a unitary representation of the group $S_n$ on the space $(\mathbb{C}^d)^{\otimes n}$; in particular, we have $P_{\text{id}} = I^{\otimes n}$ and $(P_{\pi})^* = P_{\pi^{-1}}.$ We define the symmetric subspace of order $n$ and local dimension $d$ as the linear space of vectors $\ket{\psi} \in (\mathbb{C}^d)^{\otimes n}$ that are invariant under all permutations 
$$\vee^n \C{d} := \{\ket{\psi} \in (\C{d})^{\otimes n} \mid \forall \pi \in S_n, \, \,  P_\pi \ket{\psi} =\ket{\psi}\}.$$ 
The symmetric subspace has an equivalent description as the following span \cite[Theorem 3]{harrow2013church}:
\begin{equation}
\label{eq:span-symmetric-space}
    \vee^n  \C{d} = \operatorname{span}_{\mathbb C}\{\ket{\psi}^{\otimes n} \mid \ket{\psi} \in \C{d}\}.
\end{equation} Moreover, the orthogonal projection into the symmetric subspace $\Pi^{(n)}_d : (\C{d})^{\otimes n} \rightarrow \vee^n \C{d}$ can be explicitly constructed by the sum of the permutation operators,
$$\Pi^{(n)}_d :=\frac{1}{n!}\sum_{\pi \in S_n} P_\pi.$$  
Let $\alpha \in \mathbb{N}^d$  be a vector of natural numbers, such that $|\alpha| := \sum_{i \in [d]}\alpha_i = n$. We will denote the set of all such vectors as $\knd{n}{d}$. We define the counting function, $\gamma : [d]^n \rightarrow \knd{n}{d}$ with the components $\gamma(i)_j:= |\{k \in [n] \mid i_k = j\}|$ that count the number of elements of the multi-index  that get mapped to the value $j$. 
For all occupation vectors $\alpha \in \knd{n}{d}$, there is a \emph{unique} canonical multi-index $\idx(\alpha)$ in $\gamma^{-1}(\alpha)$ that satisfies the ordering, $\idx(\alpha)_1 \leq \idx(\alpha)_2 \ldots \leq \idx(\alpha)_n$.

We define the following pure quantum states labeled by $\alpha \in \knd{d}{n}$, 
\[
\ket{D_\alpha}  := \frac{1}{\sqrt{\binom{n}{\alpha}}} \quad  \sum_{i \in [d]^n \, : \, \gamma(i) = \alpha} \ket{i} \in (\C{d})^{\otimes n},
\] where $\binom{n}{\alpha} = \frac{n!}{\alpha_1! \alpha_2! \ldots \alpha_d!}$ is the multinomial coefficient. These states are known as \emph{qudit Dicke states}, and vector $\alpha$ is called the \emph{occupation vector}.

It is easy to verify that $\Pi^{(n)}_d \ket{D_\alpha} = \ket{D_\alpha}$ for all $\alpha \in \knd{n}{d}$, hence the Dicke states are symmetric.  These states are an orthonormal basis for the symmetric subspace $\vee^n \mathbb{C}^d$ \cite[Theorem 3]{harrow2013church}. Therefore, the cardinality of the set $\knd{n}{d}$ is equal to the dimension of the symmetric subspace. Hence we have, $$d_{\mathsf{sym}} = \operatorname{dim}(\vee^n \mathbb{C}^d) = |\knd{n}{d}| = \binom {d + n-1}{d-1} = \binom {d + n-1}{n}.$$ For $n=1$, the Dicke states are just equal to the computational basis:
$$\ket{D_{(0, \ldots, 1_i, \ldots, 0)}} = \ket i.$$
For the first non-trivial case, $n = 2$, we have the following possible occupation vectors
\begin{align*}
    \forall i \neq j \quad &(0, \ldots 1_i, \ldots 1_j, \ldots 0),  \\
    \forall i \quad &(0, \ldots 2_i, \ldots 0).
\end{align*} 
These correspond to the following states in the computational basis:
\begin{align*}
    \ket{D_{(0, \ldots 1_i, \ldots 1_j, \ldots 0)}} &= \frac{\ket{ij} + \ket{ji}}{\sqrt{2}}\\
    \ket{D_{(0, \ldots 2_i, \ldots 0)}} &= \ket{ii}.
\end{align*}  
\subsubsection{Multipartite bosonic states}
We will now discuss the extension of the bosonic symmetry to linear operators. We have the following definition. 

\begin{definition}
    A matrix $X \in \mathcal{M}^{\otimes n}_d$ is called bosonic if any of the following equivalent statements are true. 
    \begin{itemize}
        \item[(1)] $X = \Pi^{(n)}_d X \Pi^{(n)}_d.$
        \item[(2)] $\operatorname{ran}(X), \operatorname{ran}(X^*) \subseteq \vee^n \mathbb{C}^d.$
        \item[(3)] $\forall \sigma \in S_n, \quad X = P_\sigma X = X P_\sigma.$
    \end{itemize}
\end{definition}

\begin{proof}
    We show that all the above definitions are equivalent. 
    \begin{itemize}
        \item[] (1) $\implies$ (2). We observe that, $\operatorname{ran}(\Pi^{(n)}_d X \Pi^{(n)}_d) \subseteq \operatorname{ran}(\Pi^{(n)}_d) = \vee^n \mathbb{C}^d, $ and the same for the adjoint.
        \item[] (2) $\implies$ (3). Let $\ket{\psi} \in \mathbb{C}^d.$ We have, $X \ket{\psi} \in \vee^n \mathbb{C}^d$ as $\operatorname{ran}(X) \in \vee^n \mathbb{C}^d.$ This shows, $$\forall \sigma, \ket{\psi} \in \mathbb{C}^d, \quad P_{\sigma} X \ket{\psi} = X \ket{\psi}.$$ Also, by a similar argument, we have, 
        $$\forall \sigma, \ket{\psi} \in \mathbb{C}^d, \quad P_{\sigma} X^* \ket{\psi} = X^* \ket{\psi} \implies (X P_{\sigma})^* = X^* \implies X P_{\sigma} = X.$$
        \item[] (3) $\implies$ (1). $X = P_\sigma X \implies X = \frac{1}{n!} \sum_{\sigma \in S_n} P_\sigma X = \Pi^{(n)}_d X.$ Similarly, we have,  $X = X \Pi^{(n)}_d.$ This completes the proof. 
    \end{itemize}
\end{proof}

\noindent We denote the linear space of bosonic matrices as $$\mathcal{M}[\vee^n \mathbb{C}^d]:= \{X \in \mathcal{M}^{\otimes n}_d \mid X = \Pi^{(n)}_d X \Pi^{(n)}_d\}$$ and the real linear space of hermitian bosonic matrices as $$\operatorname{Herm}[\vee^n \mathbb{C}^d]:= \{X \in \operatorname{Herm}[\mathcal{M}^{\otimes n}_d] \mid X = \Pi^{(n)}_d X \Pi^{(n)}_d\}.$$ From \cite[Section 1.1]{harrow2013church}, we also have the following descriptions of these spaces.

\begin{align}
\label{eq:span-operators}
\mathcal{M}[\vee^n \mathbb{C}^d]
&= \operatorname{span}_{\mathbb{C}}
\left\{ \ketbra{\psi}{\psi}^{\otimes n} \;\middle|\; \ket{\psi} \in \mathbb{C}^d \right\}, \\[0.6em]
\operatorname{Herm}[\vee^n \mathbb{C}^d]
&= \operatorname{span}_{\mathbb{R}}
\left\{ \ketbra{\psi}{\psi}^{\otimes n} \;\middle|\; \ket{\psi} \in \mathbb{C}^d \right\}.
\end{align}

We finally present a useful proposition about linear operators, which connects the restriction of the operator to the symmetric subspace to product forms. 
\begin{proposition}
\label{thm:symmetric-subspace-proj}
The following statements are equivalent for $X \in \mathcal{M}^{\otimes n}_d, $

\begin{enumerate}
    \item[(1)] $\forall \ket{v} \in \mathbb{C}^d, \quad \braket{v^{\otimes n} \vert X \vert v^{\otimes n}} = 0.$
    \item[(2)] $\Pi^{(n)}_d X \Pi^{(n)}_d = 0.$
\end{enumerate}
\end{proposition}

\begin{proof}
    The reverse direction ($2 \implies 1$) is easy, by noting that for all $\ket{v} \in \mathbb{C}^d$, $\Pi^{(n)}_d \ket{v^{\otimes n}} = \ket{v^{\otimes n}}.$  For the forward direction ($1 \implies 2$), from \cref{eq:span-operators}, we have the following equivalence:
    $$\forall \ket{v} \in \mathbb{C}^d, \quad \braket{v^{\otimes n} \vert X \vert v^{\otimes n}} = 0 \iff \forall W \in \mathcal{M}[\vee^n \mathbb{C}^d], \quad \operatorname{Tr}(XW) = 0.$$
    Moreover, this is equivalent to the statement
    $$\forall \tilde W \in \mathcal{M}^{\otimes n}_d, \quad \operatorname{Tr}(X \Pi^{(n)}_d \tilde W \Pi_d^{(n)}) =  \operatorname{Tr}(\Pi_d^{(n)} X \Pi^{(n)}_d \tilde W) = 0.$$ This shows the claim.
\end{proof}

\begin{remark}
    Note that the bosonic matrices are a strict subset of the class of permutation symmetric matrices, 
    $\{X \in \mathcal{M}^{\otimes n}_d \mid \forall \sigma \in S_n, P_\sigma X = X P_\sigma\}.$ To see this, note that the matrix $\mathbb{I}_{d^ n} = \mathbb{I}^{\otimes n}_d$ is permutation symmetric, but not bosonic, as $\operatorname{ran}(\mathbb{I}_{d^n}) = (\C{d})^{\otimes n} \supsetneq \vee^n \C{d}.$ 
\end{remark}

Since the Dicke states form a basis of the symmetric subspace of $n$ qudits, they can also be used to describe a basis for the vector space of operators acting on the symmetric subspace. 

\begin{proposition}
\label{prop:dicke-expansion}
    Any bosonic matrix $X \in \mathcal{M}[\vee^n \mathbb{C}^d]$ can be expanded in terms of Dicke states. 
    $$X = \sum_{\alpha, \beta \in \knd{n}{d}} \braket{D_\alpha |X|D_\beta}\ketbra{D_\alpha}{D_\beta}.$$
\end{proposition}

All the sets of bosonic matrices that we shall discuss in this work are presented in \cref{tbl:bosonic-matrices}, which can be used as a reference for the acronyms and the main results about each set.

\subsection{Tensors and Polynomials}
\label{sec:tensor-cones}

One of the main conceptual  contributions of this work is a correspondence between Dicke states and real symmetric tensors. In this section, we shall describe the main properties of tensors that are relevant for us. We refer the reader to \cref{tbl:symmetric-tensors} for a quick reference on the main sets of tensors and their properties used in this paper. 

\subsubsection{Notations and diagrams for tensors} For any multi-index $i \in [d]^n$, we denote the components of the tensor $T \in (\mathbb{R}^d)^{\otimes n}$ as $T_{i_1 i_2 \ldots i_n}$ in the canonical basis $e_i := e_{i_1} \otimes \, e_{i_2}  \ldots e_{i_n}$. This will be shortened to $T_i$ whenever appropriate. We use the notation $i \sqcup j \in [d]^{n+m}$ to show the join of the multi-indices $i \in [d]^n$ and $j \in [d]^m$. For two tensors $Q \in (\mathbb{R}^d)^{\otimes n}$ and $T \in (\mathbb{R}^d)^{\otimes m}$, we have the tensor $(Q \otimes T) \in (\mathbb{R}^d)^{\otimes (n+m)}$ defined component-wise as, $$\forall i \in [d]^n, j \in [d]^m, \quad (Q \otimes T)_{i \sqcup j} := Q_{i} T_{j}.$$ For any multi-index $i \in [d]^n$ and $\sigma \in S_n$,  we denote by $\sigma \cdot i := (i_{\sigma(1)}, i_{\sigma(2)} \ldots i_{\sigma(n)}).$ 

We also use the diagrammatic notation for tensors wherever it is useful. A tensor with $n$ indices is represented with a spider with $n$ legs as in \cref{fig:tensors}. The tensor product $Q \otimes T$ will be denoted by placing two tensors together. A contraction of a tensor is represented by joining two wires as \cref{fig:tensor-contraction}, and represents the following operation, 

$$(T')_{i_1 i_2 \ldots i_{n-2}} = \sum^d_{j=1} T_{jj i_1 i_2 i_3 \ldots i_{n-2}}.$$

A tensor with one index is a vector. A matrix can be interpreted as a tensor with $2$ indices. The matrix multiplication of two matrices $M$ and $K$ can be represented as a tensor contraction. $$\sum^d_{i_2=1} (M \otimes K)_{i_1 i_2 i_2 j_3} = \sum^d_{i_2=1} M_{i_1 i_2} K_{i_2 i_3}.$$ Moreover, the Euclidean inner product of two tensors can also be represented as a tensor contraction. 
$$\braket{Q,T} := \sum_{i \in [d]^n} T_{i} Q_{i} = \sum_{i \in [d]^n} (T \otimes Q)_{i \sqcup i}.$$

\begin{figure}[htb]
    \centering
    \includegraphics[width=0.5\linewidth]{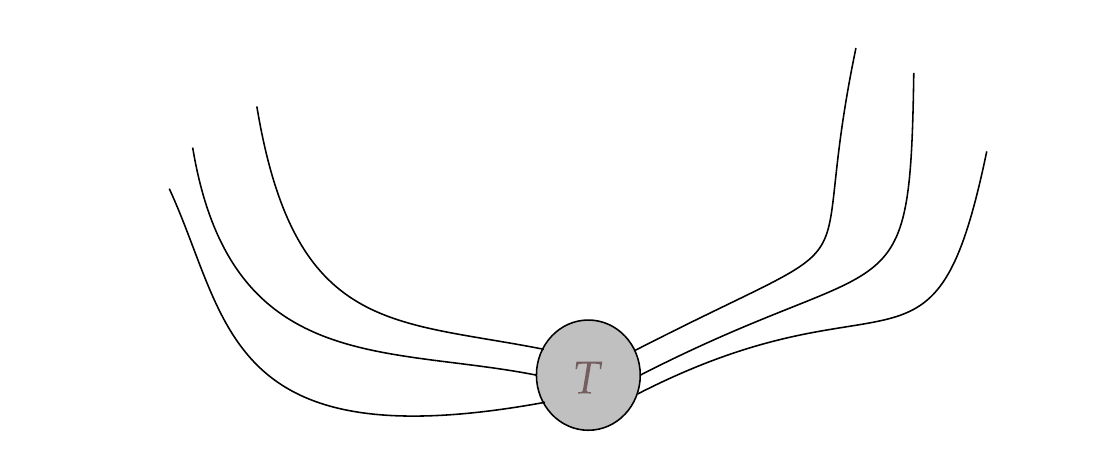}
    \caption{The tensor spider $T$ with $6$ legs.}
    \label{fig:tensors}
\end{figure}

\begin{figure}[htb]
    \centering
    \includegraphics[width=0.5\linewidth]{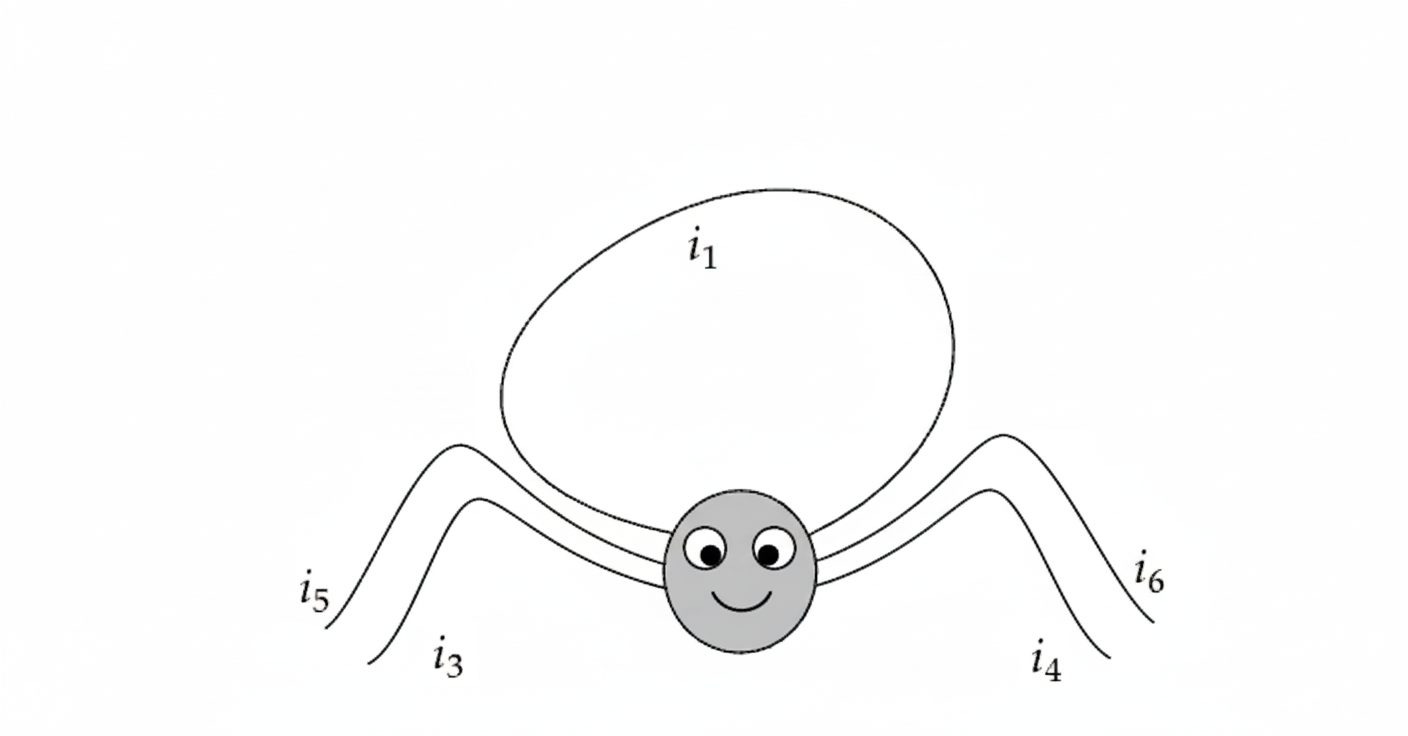}

    \caption{A tensor contraction: the contracted tensor has now $4$ legs $i_3, i_4, i_5, i_6$. The other two legs are equal $i_1 = i_2$ and summed over. }
    \label{fig:tensor-contraction}
\end{figure}

\subsubsection{Symmetric tensors and polynomials}
In this article, we will work primarily with \emph{real} tensors with permutation symmetry. Formally, any tensor $T \in (\mathbb{R}^d)^{\otimes n}$ is \emph{symmetric tensor} if it stays invariant under the permutation of its indices, i.e, $$\forall \sigma \in S_n, \forall i \in [d]^n, \quad T_{i} = T_{\sigma \cdot i}.$$ Note that in the case of the $2$ indices, by interpreting the indices to be row and column indices of a matrix, these tensors correspond exactly to real symmetric matrices. The linear space of such tensors will be denoted by $\vee^n\mathbb{R}^d$. We can define the following projection map, $\Pi^{(n, \mathbb{R})}_d : (\mathbb{R}^d)^{\otimes n} \to \vee^n\mathbb{R}^d$ such that, 

$$\forall i \in [d]^n, \quad \Pi^{(n, \mathbb{R})}_d [T]_{i} := \frac{1}{n!} \sum_{\sigma \in S_n} T_{\sigma \cdot i}.$$ 

For $\alpha \in \mathbb{N}^d$, $x \in \mathbb{R}^d$, we define the expression $x^\alpha$ to denote the monomial $x^{\alpha_1}_1 x^{\alpha_2}_2 \ldots x_d^{\alpha_d}$. Note that the degree of such a monomial is $\sum_i \alpha_i:= |\alpha|$. Let us denote by $[x]^n_d$ the vector of homogeneous monomials of degree $n$ and in $d$ variables, by ordering the monomials using the lexicographic ordering. This is done by imposing the ordering $x^\alpha > x^\beta$ if in $\alpha - \beta$, the first non-zero entry from the left is positive. An example of such an ordering in the case of monomials of degree $2$ and $2$ variables is, $[x]^2_2:= (x^2_1, x_1x_2, x_2^2)$. 

It is well-known that the symmetric tensors have a linear bijection to the set of homogeneous polynomials of degree $n$ and $d$ variables, which we denote as $\homo{d}{n}$. This bijection is the following map $T \mapsto p_T$:
\begin{equation}\label{eq:Ti-vs-Talpha}
    p_T(x):= \braket{T, x^{\otimes n}} = \sum_{i \in [d]^n} T_{i} x_{i_1} x_{i_2} \ldots x_{i_n} = \sum_{\alpha \in \knd{n}{d}} \binom{n}{\alpha} T_{\idx(\alpha)} x^\alpha.
\end{equation}

Moreover, note the following characterization of \emph{real} symmetric tensors \cite{comon2008symmetric}, which is analogous to that of complex tensors in \cref{eq:span-symmetric-space}:
$$\vee^n \mathbb{R}^d = \operatorname{span}_{\mathbb{R}}\{v^{\otimes n} \mid v \in \mathbb{R}^d\}.$$

\subsubsection{Convex duality}
Let $V$ be a real vector space, $V^*$ be the dual vector space  and $\mathcal{C}$ be a convex cone. The \emph{dual cone} of a convex cone $\mathcal{C} \subseteq V$ is defined by the following set.  $$\mathcal{C}^\circ := \{X \in V^* \mid \forall \, Y \in \mathcal{C}, \, X (Y) \geq 0\}.$$
If for two convex cones, $\mathcal{C}, \mathcal{K}$, we have $\mathcal{C} \subseteq \mathcal{K}$, it implies that $\mathcal{K}^\circ \subseteq \mathcal{C}^\circ$. A cone is called self-dual if $\mathcal{C}^\circ \cong \mathcal{C}$. If the vector space $V$ has an inner product $\braket{. \, ,  .}$, then all the linear functionals are of the form, 
$$X(Y) = \braket{L_X, Y},$$ for a unique $L_X \in V$. This provides a natural isomorphism between $V$ and $V^*.$ In this article, we consider tensor cones $C \subseteq \vee^n \mathbb{R}^d$ with the Euclidean inner product. Then, the entrywise non-negative is an example of a self-dual convex cone. 
\begin{equation}
\label{eq:definition-ewp}
    \mathsf{NN}^{(n)}_d := \{Q \in \vee^n \mathbb{R}^d \mid \forall i \in [d]^n, \, Q_{i} \geq 0\}.
\end{equation}

\subsubsection{Completely Positive (CP) and Copositive (Cop) tensors}
We first introduce the notion of completely positive tensors. This notion generalizes that of completely positive matrices, which have been investigated since the 1960s \cite{hall1963copositive}. These are symmetric matrices that can be written as a sum of rank-1 symmetric matrices, i.e, $$X = \ketbra{v_1}{v_1} + \ketbra{v_2}{v_2} \ldots + \ketbra{v_K}{v_K}$$ where each $\ket{v_q} \in \mathbb{R}^d_+$ is a vector with non-negative entries. We recommend the excellent book by A. Berman and N. Shaked-Monderer \cite{shaked2021copositive} as a comprehensive reference on this topic. This cone of completely positive matrices is quite well-studied due to its importance in the theory of optimization, and also recently found a connection to entanglement theory for bosonic states \cite{yu2016separability, tura2018separability}. Let us now define the general notion for tensors,

\begin{definition}
\label{def:cp-tensors}
A symmetric tensor $Q \in \vee^n (\mathbb{R}^d)$ is called \emph{completely positive} if there exists a finite set of entry-wise non-negative vectors $\{v_q\}^K_{q=1} \in \mathbb{R}^d_+$ such that 
$$Q = v_1^{\otimes n} + v_2^{\otimes n} + \ldots + v^{\otimes n}_K.$$
\end{definition} 
The set of completely positive tensors is a convex, pointed, and closed cone with a non-empty interior \cite[Proposition 1]{pena2015completely} in the real vector space of symmetric tensors $\vee^n\mathbb{R}^d$. We denote this convex cone as $\mathsf{CP}^{(n)}_d$ where $n$ denotes the number of legs (indices) and $d$ is the local dimension. Clearly, any completely positive tensor is entrywise non-negative, as all the components of $v_q$ are non-negative. We refer the reader to \cite{qi2014nonnegative,luo2015doubly,nie2022dehomogenization} for the recent work on this class of tensors. The membership problem for even completely positive matrices is known to be NP-hard \cite{murty1987some}. Despite the hardness of general characterization, some results have been obtained in this direction by different approaches; looking at structured tensors \cite{luo2016completely}, in low dimensions \cite{nie2018complete}, or solving a growing sequence of semi-definite programs \cite{fan2014semidefinite,nie2022dehomogenization}. Similarly to the matrix case, some classes of polynomial optimization problems can be reformulated in terms of the completely positive tensors \cite{pena2015completely}. Surprisingly, it has been shown that for $d=2$, deciding complete positivity is equivalent to checking positivity of two Hankel matrices \cite{fan2019completely}.

We introduce the \emph{dual cone} of the completely positive tensors, called \emph{copositive tensors} \cite{qi2013symmetric}: 
$$\mathsf{Cop}^{(n)}_d:= (\mathsf{CP}^{(n)}_d)^{\circ}.$$ 
Recall that for all symmetric tensors $Q \in \vee^n\mathbb{R}^d$, there is an associated homogeneous polynomial $p_Q(x) = \braket{Q, x^{\otimes n}}$; $p_Q(x)$ is a homogeneous polynomial of degree $n$ in $d$ variables with real coefficients. The following proposition provides another characterization of copositivity.
\begin{proposition}
\label{def:copositive-tensors}
A symmetric tensor $Q \in \vee^n\mathbb{R}^d$ is copositive if and only if
$$\forall x \in \mathbb{R}^d_+ \qquad p_Q (x) \geq 0.$$
\end{proposition}
Note that in the previous definition, we only need the polynomial to be positive for all non-negative vectors $x$, but in general, the polynomial $p_Q (x)$ can take negative values. This generalizes the notion of copositive matrices, which correspond to the second-order copositive tensors \cite{Diananda_1962, hall1963copositive, dur2010copositive, BomzeEtAl2000}. The membership problem for copositive matrices is known to be co-NP-hard \cite{murty1987some}. Like the completely positive cone, this convex cone is well-studied due to its importance in the theory of optimization as copositivity offers a unified convex way to reformulate non-convex mixed quadratic programs into convex programs \cite{dur2010copositive}.  For a comprehensive review of copositive matrices, we again refer the reader to the detailed book by Berman and Shaked-Monderer \cite{shaked2021copositive}. 
\subsubsection{Sum-of-squares}

We review in this section some basic but important facts about multivariate polynomials that can be written as sums of squares. Such polynomials play a central role in semi-algebraic geometry, see \cite{blekherman2012semidefinite}.

\begin{definition}
    Let $p \in \homo{d}{2n}$ be a homogeneous polynomial in $d$ variables $(x_1 \ldots x_d) =: x$ and degree $2n$. Such a polynomial $p$ is said to be a \emph{sum of squares} polynomial (or in short a SOS polynomial) if there exist a finite number of homogeneous polynomials ${f_k({x})}^K_{k=1} \in \homo{d}{n}$ having degree $n$, such that, 
    $$p({x}) = \sum^K_{k=1} f_k^2({x}).$$
\end{definition} 
We denote the set of sum of squares polynomials in even degree $2n$ and $d$ variables as $\Sigma^{2n}_d[x] \subseteq \homo{d}{2n}$. The polynomials that can be written as a sum of squares are clearly globally nonnegative (as $\forall f \in \mathbb{R}[x], f(x)^2 \geq 0$). But the converse of the statement is false, which is the content of Hilbert's celebrated results \cite{hilbert1888darstellung}.  This means that there exist globally positive homogeneous polynomials that are not SOS. The following \cref{tab:homo-sos} summarizes Hilbert's theorem for the case of homogeneous polynomials. Note that the same result holds true also for \emph{symmetric} homogeneous polynomials \cite{goel2017analogue}, i.e polynomials that satisfy, $$\forall \sigma \in S_d \quad p(x_1, x_2 \ldots x_d) = p(x_{\sigma(1)},x_{\sigma(2)} \ldots x_{\sigma(d)}).$$

\begin{table}[htb] 
\centering 
\renewcommand{\arraystretch}{1.3}
\begin{tabular}{|c|c|c|c|}
\hline
\rowcolor{yellow!10}
 & $2n = 2$ & $2n = 4$ & $2n \ge 6$ \\
\hline
\rowcolor{yellow!10}
$d = 2$ & \cellcolor{green!20}Yes & \cellcolor{green!20}Yes & \cellcolor{green!20}Yes \\
\hline
\rowcolor{yellow!10}
$d = 3$ & \cellcolor{green!20}Yes & \cellcolor{green!20}Yes & \cellcolor{red!20}No \\
\hline
\rowcolor{yellow!10}
$d \ge 4$ & \cellcolor{green!20}Yes & \cellcolor{red!20}No & \cellcolor{red!20}No \\
\hline
\end{tabular}
\caption{Can any homogeneous or symmetric homogeneous polynomial of degree $2n$ in $d$ variables be represented as a sum of squares? \label{tab:homo-sos}}
\end{table}

Next, we show a reformulation of the sum of squares condition of the polynomial in terms of positive semidefinite matrices. Similar results are well-known in the literature, see \cite{powers1998algorithm, laurent2008sums} under the name of the \emph{Gram-Matrix Method}. This formulation also provides an efficient way to check the sum of squares property of polynomials, by recasting it as a semi-definite program (SDP) \cite{boyd2004convex,lasserre2009moments,blekherman2012semidefinite}. 
\begin{theorem}
\label{thm:gram-matrix}
    A homogeneous polynomial $p \in \homo{d}{2n}$ of degree $2n$ is a sum of squares (SOS) if and only if there exists a PSD matrix, $P \in \operatorname{Herm}[\mathcal{M}_d(\mathbb{R})^{\otimes n}]$ such that 
    $$p(x) =\braket{x^{\otimes n} |P| x^{\otimes n}}.$$
\end{theorem}

\begin{proof}
    ($\impliedby$) Assume that there exists a PSD matrix $P$ such that $p(x) =\braket{x^{\otimes n} |P| x^{\otimes n}}$. Since it is PSD, $P$ can be written as $X^\top X$, where $X \in \mathcal{M}_d(\mathbb{R})^{\otimes n}$. Therefore, 
    $$p(x) = \braket{x^{\otimes n} |P |x^{\otimes n}} =\braket{x^{\otimes n} |X^\top X |x^{\otimes n}}= \sum_{i \in [d]^n} f_{i}^2(x)$$ 
    where $\forall i \in [d]^n, \, f_{i}(x) = \sum_{j \in [d]^n} X_{i,j} x_{j_1} x_{j_2} \ldots x_{j_n}$. This implies that $p(x)$ is a sum of squares polynomial.

    \smallskip

    \noindent($\implies$) For the reverse implication, assume that the polynomial $p(x)$ is a sum of squares, i.e, of the form
    $$p(x) = \sum^K_{k=1} f^2_k(x).$$
    For all $k$, expand $f_k(x) = \sum_\alpha \mu_{k, \alpha}  x^\alpha$ as a sum of monomials. Note that, we can also expand $f_k(x) = \braket{\tilde \mu_k, x^{\otimes n}}$ where $\forall i \in [d]^n,  (\tilde \mu_k)_{i} = \binom{n}{\gamma(i)}^{-1}\mu_{k, \gamma(i)}$, see \cref{eq:Ti-vs-Talpha}. Then, we have: 
    $$p(x) = \sum^K_{k=1} \braket{\tilde \mu_k, x^{\otimes n}}^2  = \sum^K_{k=1} \braket{x^{\otimes n}, \tilde \mu_k} \braket{\tilde \mu_k, x^{\otimes n}}  =  \bra{x^{\otimes n}}\sum^K_{k=1} \vert \tilde \mu_k \rangle \langle \tilde  \mu_k \vert \ket{x^{\otimes n}}.$$ By choosing $P := \sum^K_{k=1} \vert \tilde \mu_k \rangle \langle \tilde  \mu_k \vert$, we conclude the proof. 
\end{proof}
\begin{remark}
    It is easy to show with a slight modification of the above proof that the polynomial is SOS if and only if there is a PSD matrix $P$ such that 
    $p(x) = ([x]^n_d)^\top P [x]^n_d$ where $[x]^n_d$ is the vector of monomials of degree $n$ in $d$ variables. The matrix $P$ in this case is known as the \emph{Gram}-Matrix \cite{powers1998algorithm, laurent2008sums}.
\end{remark} Recall that there is a homogenous polynomial of $d$ variables and degree $n$ associated to the tensor $Q \in \vee^n \mathbb{R}^d$, 
$$p_Q(x) := \sum_{i \in [d]^n} Q_{i_1 i_2 \ldots i_n} x_{i_1} x_{i_2} \ldots x_{i_n}.$$
From \cref{def:copositive-tensors}, it is clear that $Q$ is a copositive tensor if and only if $p_Q (x \odot x) \geq 0$ for all $x \in \mathbb{R}^d$. This motivates the definition of sum-of-squares (SOS) tensors.  
\begin{definition}[SOS Tensors]
\label{def:sos-tensors}
A symmetric tensor $T$ will be called a sum-of-squares tensor, or in short, an SOS Tensor, if the polynomial, $p_T(x \odot x)$ is a sum of squares. We denote the cone of these tensors using the notation, 
$$\mathsf{SOS}^{(n)}_d := \{T \in \vee^n \mathbb{R}^d \mid p_T(x \odot x) \in \Sigma^{(2n)}_d\}.$$
\end{definition} 

It follows from the definition that all SOS tensors are copositive (as all SOS polynomials are positive). The SOS tensors form a pointed convex cone. Note that in \cref{def:sos-tensors}, we do not demand that the tensor polynomial $p_T$ is a sum of squares polynomial, but a weaker condition that $p_T(x \odot x)$ is an SOS polynomial. For example, consider the two-variable polynomial $p_1(x_1,x_2) = x_1 x_2.$ This polynomial is not positive for all $x$, hence also not a sum of squares. But the polynomial $p_1(x^2_1,x^2_2) = x^2_1 x^2_2$ is an SOS polynomial. The homogenous polynomials of this form were studied in \cite{ZVP06}, and the form of these polynomials was completely characterized. We rephrase the result in the next theorem.

\begin{theorem}[{{\cite[Proposition 9]{ZVP06}}}] \label{prop:SOS}
Let $p(x) \in \homo{d}{n}$ be a real homogeneous polynomial in $x \in \mathbb{R}^d$ of degree $n$ such that $p(x \odot x)$ is a sum of squares: $p(x \odot x) \in \Sigma^{(2n)}_d$. Then the polynomial $p$ admits a decomposition of the form 
\begin{equation} \label{eq:SOS}
    p(x)=\sum_{j=0}^{\lfloor n/2\rfloor}\sum_{\alpha \in \knd{n-2j}{d}
    }  x^{\alpha}\psi_{j,\alpha}(x)
\end{equation}
where $|\alpha|:=\alpha_1+\cdots+\alpha_d$, $ x^{\alpha}:=x_1^{\alpha_1}\cdots x_d^{\alpha_d}$, and each $\psi_{j,\alpha}$ is a homogeneous SOS polynomial of degree $2j$: $\psi_{j,\alpha}(x) \in \Sigma^{(2j)}_d$.
\end{theorem}

Let us consider the SOS tensors for $n=2$. In this case, by definition, $Q \in \mathsf{SOS}^{(n)}_d$, if the polynomial $p_Q(x \odot x) = \sum_{i,j} Q_{ij}x^2_i x^2_j \in \Sigma_d^{(4)}$. We apply \cref{prop:SOS} to obtain the following decomposition:
\begin{align*}
    p_Q(x) = \sum_{i,j \in [d]} Q_{ij}x_i x_j &= \sum_{|\alpha| = 2}  x^\alpha \psi_{0,\alpha}(x)  + \sum_{|\alpha| = 0}  x^\alpha \psi_{1,\alpha}(x)\\
    &= \sum_{i,j \in [d]} x_ix_j \lambda_{ij} + 1 \cdot \braket{x \vert A \vert x}.
\end{align*}
We obtain for the first term with $|\alpha|=2$, all the monomials $x_i x_j$ and the degree $0$ SOS polynomial $\lambda_{ij}:=\psi_{0,\alpha}(x) \in \mathbb{R}_+$. For the second term, we have $ x^\alpha = 1$ and the degree $2$ SOS polynomial, $\psi_{1,\alpha}(x) = \braket{x \vert A \vert x}$ for some $A \in \mathsf{PSD}_d$. 
Therefore this shows that $Q_{ij} = \lambda_{ij} + A_{ij}$. If we interpret the indices of $Q$ as a  matrix, we get, $Q = \{\lambda_{ij} + A_{ij} \}^d_{i,j=1} \in \mathsf{NN}_d +\mathsf{PSD}_d$.  The reverse is also true. If $Q  \in \mathsf{NN}_d +\mathsf{PSD}_d$, this implies that $p_Q(x \odot x) \in \Sigma_d^{(4)}.$ This fully characterizes the SOS cone in $n=2.$

\begin{remark}

The previous characterization also shows that the convex cone of SOS tensors is closed (since sum of square polynomials of degree $n$ and $d$ variables is closed). Moreover, there exists an $\epsilon > 0$ such that the open ball $B(\epsilon, \ket{e}^{\otimes n}) \subseteq \mathsf{SOS}^{(n)}_d$ implying that the cone has a non-empty interior where $\ket{e} := \sum_{i} e_i$
\end{remark}

\begin{remark}
\label{rem:sos-tensors-n-2}
    The convex cone of $\mathsf{NN} +\mathsf{PSD}$ is well-studied, usually denoted as $\mathsf{SPN}$ in the literature. It is the inner approximation to the set of copositive matrices, and is dual to the cone of \emph{doubly non-negative matrices}, that are the positive-semidefinite matrices with non-negative entries, $\mathsf{DNN} := \mathsf{NN} \cap \mathsf{PSD}$. For results about these matrix cones, we refer the reader to \cite{shaked2021copositive}. The sum of squares tensors can be considered a generalisation of $\mathsf{SPN}$ matrix cone. 
\end{remark}

We also present here the characterization of the duals of sum of square polynomials. For polynomials, we define the duals without any reference to any inner product, and instead, we look at the linear functionals that are positive on all SOS polynomials.

$$(\Sigma^{(2n)}_d)^\circ := \{f : \homo{d}{2n} \to \mathbb{R} \mid \forall \, p \in \Sigma^{(2n)}_d, \, f(p) \geq 0 \}.$$

Recall that the monomials of a degree $n$ in $d$ variables, $[x]^{n}_d$ are labeled by the vectors $\alpha \in \mathbb{N}^d$ such that $|\alpha|=n.$ The set of these monomials is of size $\binom{n+d-1}{d-1}.$ Then, we have the following definition.
\begin{definition}[Moment Matrix \cite{laurent2008sums}]
    Let $(y_{\alpha})_{|\alpha| = 2n}$ be a sequence labelled by the monomial vectors, $\alpha \in \mathbb{N}^d$ such that $|\alpha|=2n$. Then, define the following $\binom{n+d-1}{d-1} \times \binom{n+d-1}{d-1}$ matrix with rows and columns indexed using the monomial vectors, $$\mathcal{M}_n(y)_{\alpha \beta}:= y_{\alpha+\beta}.$$ This matrix is called the moment matrix associated with the sequence $y_{\alpha}$. 
\end{definition} 

Let $f$ be a linear functional on homogeneous polynomials. Then, we associate a moment matrix with the linear functional $f$ by choosing the sequence $y_\alpha:= f(x^\alpha)$ where $\alpha \in \mathbb{N}^d, |\alpha| = 2n$.  We then have the following result about the dual of SOS polynomials.

\begin{theorem}
\label{thm:dual-sos-polynomials}
The dual cone of SOS polynomials, 
    $$(\Sigma^{(2n)}_d)^\circ = \{f : \homo{d}{2n} \rightarrow \mathbb{R} \mid \mathcal{M}_n(y) \succeq 0  \text{ where } y_{\alpha} := f (x^\alpha)\}.$$

\end{theorem}
\begin{proof}
Recall that a polynomial is an SOS polynomial, i.e $p \in \Sigma^{(2n)}_d$ if and only if there exists a positive semidefinite matrix $P$ such that $p(x) = ([x]^n_d)^\top P [x]^n_d$. Expanding this, we get,
\[
p(x) = ([x]^n_d)^\top P [x]^n_d
      = \sum_{\alpha, \beta \in \knd{n}{d}}
        P_{\alpha\beta}\, x^{\alpha+\beta},
\]
where we recall that $\knd{n}{d}$ denotes the set of $d$-tuples of non-negative integers that sum up to $n$. Applying the linear functional $f$ to this expression while having $f(x^\alpha) =y_\alpha$, we obtain,
\[
f(p) 
= \sum_{\alpha, \beta \in \knd{n}{d}}
    P_{\alpha\beta}\, y_{\alpha+\beta}
= \operatorname{Tr}\!\big(P\,\mathcal M_n(y)\big).
\] Moreover, this is positive for all PSD matrices $P$, if and only if $\mathcal{M}_n(y) \succeq 0$. 
\end{proof} 

The last result, \cref{thm:dual-sos-polynomials}, completely characterizes the dual cone of the SOS polynomials; a linear functional is in the dual cone if and only if the corresponding moment matrix is positive-semidefinite. We refer the reader to the excellent notes by Laurent \cite[Section 4.4]{laurent2008sums} for more details on moment matrices, particularly for the connection to the moment problem.

\subsubsection{SOS and Non-Negative Hierarchy}
In \cite{parrilo2000structured}, an inner approximation of copositive matrices was introduced based on the SOS hierarchy that approximates positive polynomials. It consists of the following convex cones. 

\begin{equation}
\label{eq:reznick-sos}
\mathsf{RSOS}^{(r)}_d :=  \{Q \in \vee^2 \mathbb{R}^d \, | \, p_Q(x \odot x)||x||^{2r} \in \Sigma_d^{(2r)}\}
\end{equation}

Each of the sets $\mathsf{RSOS}^{(r)}_d$ involves checking the sum of squares property of the polynomial $p_Q(x \odot x) = \sum_{i,j \in [d]} Q_{ij} x^2_i x^2_j$, after multiplying it with another polynomial, $||x||^{2r},$ that is SOS by definition. We then have the following inclusions, 
$$\mathsf{PSD}_d+ \mathsf{NN}_d=\mathsf{RSOS}^{(0)}_d \subseteq \mathsf{RSOS}^{(1)}_d \subseteq \mathsf{RSOS}^{(2)}_d \subseteq \mathsf{RSOS}^{(3)}_d \ldots \subseteq \mathsf{RSOS}^{(r)}_d \subseteq \mathsf{Cop}_d.$$ 
The following are well-known results about this hierarchy.
\begin{itemize}
    \item $\mathsf{RSOS}_n^{(0)}= \mathsf{SPN}_d = \mathsf{NN}_d+\mathsf{PSD}_d$. In particular we have, $\mathsf{Cop}_d=\mathsf{RSOS}^{(0)}_d$ for $d\leq 4$.

    \item \cite{parrilo2000structured,DDGH13} If $d\geq 5$, $\mathsf{RSOS}^{(0)}_d\subsetneq \mathsf{RSOS}_d^{(1)}\subseteq\cdots\subseteq \mathsf{RSOS}^{(r)}_d \subsetneq \mathsf{Cop}_d$. 

    \item \cite{laurent2023exactness} $\bigcup_{r\geq 0}\mathsf{RSOS}_5^{(r)}=\mathsf{Cop}_5$ and $\operatorname{int}(\mathsf{Cop}_n)\subsetneq \bigcup_{r\geq 0}\mathsf{RSOS}^{(r)}_d\subsetneq \operatorname{cl}\big({\bigcup_{r\geq 0}\mathsf{RSOS}^{(r)}_d}\big)=\mathsf{Cop}_d$ for $d\geq 6$.
\end{itemize}

A similar approach of constructing the sum of squares (SOS) hierarchy is also possible for tensors.
$$\mathsf{RSOS}^{(m,r)}_d :=  \{Q \in \vee^m \mathbb{R}^d \, | \, p_Q(x \odot x) ||x||^{2r} \in \Sigma_d^{2(m+r)}\}.$$ These sets satisfy similar inclusions for each $m$.
$$\mathsf{RSOS}^{(m,0)}_d \subseteq \mathsf{RSOS}^{(m,1)}_d \subseteq \mathsf{RSOS}^{(m,2)}_d \subseteq \mathsf{RSOS}^{(m,3)}_d \ldots \subseteq \mathsf{RSOS}^{(m,r)}_d \subseteq \mathsf{Cop}^{(m)}_d.$$

\begin{itemize}
    \item  For $d=2$, $\mathsf{RSOS}^{(m,0)}_2 = \mathsf{Cop}^{(m)}_2.$ This follows from the fact that all homogenous polynomials in $2$ variables are SOS \cite{hilbert1888darstellung}.
    \item $\operatorname{int}(\mathsf{Cop}^{(m)}_d)\subseteq \bigcup_{r\geq 0}\mathsf{RSOS}_d^{(m,r)}\subseteq \operatorname{cl}\big({\bigcup_{r\geq 0}\mathsf{RSOS}_d^{(m,r)}} \big)=\mathsf{Cop}^{(m)}_d$ by Reznick's Positivstellens\"atze \cite{reznick1995uniform}. 
\end{itemize}

Let $e :=\sum^d_{i=1} e_i$ be the all-ones vector. Then, we can rewrite, $$p_Q(x \odot x)||x||^{2r} = \braket{Q \otimes e^{\otimes r}, (x \odot x)^{\otimes (m+r)}} = \braket{\Pi^{(m+r, \mathbb{R})}_d[Q \otimes e^{\otimes r}], (x \odot x)^{\otimes (m+r)}}.$$

This implies that the hierarchy can be rephrased purely in terms of tensors.

\begin{equation}
    \mathsf{RSOS}^{(m,r)}_d = \{Q \in \vee^m \mathbb{R}^d \mid \Pi^{(m+r, \mathbb{R})}_d[Q \otimes e^{\otimes r}] \in \mathsf{SOS}^{(m+r)}_d\}.
\end{equation}

Let $\chi^{(n)}_d \subseteq \homo{d}{n}$ be the cone of polynomials with non-negative coefficients. It is easy to see that if a tensor polynomial $p_T \in \chi^{(n)}_d$ it implies that $p_T(x \odot x) \in  \Sigma^{(2n)}_d.$ We consider the following hierarchy of convex sets considered in \cite{de2002approximation}, which also provides an inner approximation of copositive matrices.
\begin{equation}
\label{eq:polyna-nn}
    \mathsf{PNN}^{(r)}_{d}
:= 
\Bigl\{
  Q \in \vee^{2}\mathbb{R}^{d}
  \;\Big|\;
  \Bigl( \sum_{i,j \in [d]} Q_{ij}\, x_i x_j \Bigr)
  \Bigl( \sum_{l=1}^{d} x_l \Bigr)^r
  \in \chi^{(r+2)}_{d}
\Bigr\}.
\end{equation}
For all $r$, we have the inclusion $\mathsf{PNN}^{(r)}_{d} \subseteq \mathsf{RSOS}^{(r)}_{d},$ i.e, this hierarchy acts as a weaker approximation than the SOS one. This hierarchy satisfies the following properties.

\begin{itemize}
    \item $\mathsf{PNN}_d=\mathsf{C}^{(0)}_d$
    \item $\forall d \geq 2, \quad \mathsf{PNN}^{(0)}_d \subseteq \mathsf{PNN}^{(1)}_d  \subseteq \mathsf{PNN}^{(2)}_d \subseteq \mathsf{PNN}^{(3)}_d \ldots \subseteq \mathsf{PNN}^{(r)}_d \subseteq \mathsf{Cop}_d.$ 
    \item  $\operatorname{int}(\mathsf{Cop}_d) \subseteq \bigcup_{r \geq 0} \mathsf{PNN}^{(r)}_d \subsetneq \operatorname{cl}\big({\bigcup_{r \geq 0} \mathsf{PNN}^{(r)}_d}\big) = \mathsf{Cop}_d$ follows from Polya's Positivstellens\"atze \cite{polya1911positive, de2002approximation}.
\end{itemize}
A similar approach of constructing the non-negative (NN) hierarchy is also possible for tensors.
\begin{align*}
    \mathsf{PNN}^{(m,r)}_d &:=  \{Q \in \vee^m \mathbb{R}^d \, | \, p_Q(x) (\sum^d_{l=1} x_l)^r \in \chi_d^{(m+r)}\} \\
    &= \{Q \in \vee^m \mathbb{R}^d \, | \, p_Q(x\odot x) \|x\|^{2r} \in \chi_d^{2(m+r)}\}.
\end{align*}
These sets satisfy similar inclusions for each $m$. 
$$\mathsf{PNN}^{(m,0)}_d \subseteq \mathsf{PNN}^{(m,1)}_d \subseteq \mathsf{PNN}^{(m,2)}_d \subseteq \mathsf{PNN}^{(m,3)}_d \ldots \subseteq \mathsf{PNN}^{(m,r)}_d \subseteq \mathsf{Cop}^{(m)}_d.$$ 
This hierarchy can also be rephrased purely in terms of tensors. \begin{equation}
    \mathsf{PNN}^{(m,r)}_d = \{Q \in \vee^m \mathbb{R}^d \mid \Pi^{(m+r, \mathbb{R})}_d[Q \otimes e^{\otimes r}] \in \mathsf{NN}^{(m+r)}_d\}.
\end{equation}

\section{Diagonally symmetric subspace}
\label{sec:ds-subspace}
In this section, we introduce the class of multipartite bosonic quantum states called \emph{mixtures of Dicke states} or \emph{Diagonally Symmetric States}. These states have been studied in the case of bipartite systems \cite{tura2018separability, singh2021diagonal, marconi2021entangled, gulati2025entanglement} and, recently, the multipartite extension has been considered in \cite{romero2025multipartite}. We begin by discussing the basic structural properties of these states, along with their symmetries. Recall that any $n$-partite bosonic matrix $X \in (\M{d})^{\otimes n}$ has an expansion into the orthonormal Dicke basis 
$$X = \sum_{\alpha,\beta \in \knd{n}{d}} \lambda(X)_{\alpha,\beta}\vert D_\alpha \rangle \langle D_{\beta} \vert,$$ where $\lambda(X)_{\alpha,\beta} := \braket{D_\alpha \vert X \vert D_{\beta}}.$

Let us begin by introducing the diagonally symmetric subspace. We denote by $\diagU_d$ the abelian group of diagonal unitary matrices. These are exactly the diagonal matrices with complex phases on the diagonal
$$\diagU_d := \{U_{\theta} \mid U_{\theta} := \operatorname{diag}(e^{i \theta_1},e^{i \theta_2},\ldots,e^{i \theta_d}), \, \,  \theta \in [0,2\pi)^d\}.$$

\begin{definition}[Diagonally Symmetric Subspace]
\label{def:DS-subspace}
A matrix $X \in (\mathcal{M}_{d})^{\otimes n},$ will be called diagonally symmetric if it satisfies the following two symmetry conditions:
\begin{enumerate}
    \item Bosonic symmetry, $$\Pi^{(n)}_d X \Pi^{(n)}_d = X$$
    \item Diagonal unitary invariance, $$\forall U \in \diagU_{d} \quad \; [U^{\otimes n}, X] = 0.$$
\end{enumerate}
\end{definition} We will denote the vector space of diagonally symmetric (DS) matrices in $(\mathcal{M}_{d})^{\otimes n}$ as $\mathsf{DS}^{(n)}_d, $ and  
the \emph{real} vector space of self-adjoint diagonally symmetric (DS) matrices as $\operatorname{Herm}[\mathsf{DS}^{(n)}_d].$ In the next \cref{prop:symmetry-mods}, we show that this space consists of matrices that are diagonal in the Dicke basis.  
\begin{proposition}
\label{prop:symmetry-mods} Let \(d,n \in \mathbb{N}\). We have: 
\begin{itemize}
    \item[(1)] $\mathsf{DS}_d^{(n)} = \operatorname{span}_{\mathbb{C}} \bigl\{\ketbra{D_\alpha}{D_\alpha} \,\big|\, \alpha \in \knd{n}{d} \bigr\}, $
    \item[(2)] $\operatorname{Herm}[\mathsf{DS}^{(n)}_d] = \operatorname{span}_{\mathbb{R}} \bigl\{\ketbra{D_\alpha}{D_\alpha} \,\big|\, \alpha \in \knd{n}{d} \bigr\}.$
\end{itemize}

Therefore, the matrices $\{\ketbra{D_\alpha}{D_\alpha}\}_{\alpha \in \knd{n}{d}}$ form an ONB (orthonormal basis) of the DS subspace.  
\end{proposition}
\begin{proof}

Consider an operator $X \in \mathsf{DS}_d^{(n)}.$ By \cref{def:DS-subspace}, it satisfies the two symmetries. \[\Pi^{(n)}_d X \Pi_d^{(n)} = X, \]
\[ [U^{\otimes n}, X] = 0 
\text{ for all } U \in \diagU_d.
\]
Since $\Pi^{(n)}_d X \Pi_d^{(n)} = X$, then by \cref{prop:dicke-expansion}, \(X \in \mathcal{M}_d^{\otimes n}\) admits an expansion in the Dicke basis:
\[
X = \sum_{\alpha,\beta \in \knd{n}{d}} \lambda(X)_{\alpha,\beta}\, \ket{D_\alpha}\!\bra{D_\beta}
\]
for some complex coefficients, \(\lambda(X)_{\alpha,\beta}\). Let $U_{\theta} := \mathrm{diag}(e^{i\theta_1}, \dots, e^{i\theta_{d}}) \in \diagU_d$ be an arbitrary diagonal unitary matrix. 
The action of the unitary matrix $U_{\theta}^{\otimes n}$ on Dicke states is:
\[
U_{\theta}^{\otimes n} \ket{D_\alpha}  = \frac{1}{\binom{n}{\alpha}}\sum_{j \, : \, \gamma(j) = \alpha} U_{\theta}^{\otimes n} \ket{j} = \frac{1}{\binom{n}{\alpha}}\sum_{j \, : \, \gamma(j) = \alpha} e^{i \alpha \cdot \theta}\ket{j} = e^{i \alpha \cdot \theta} \ket{D_\alpha},
\]
where $\alpha \cdot \theta$ is the Euclidean dot product of vectors. Hence, we have
\[
\forall \alpha, \beta \in \knd{n}{d},\quad U_{\theta}^{\otimes n} \ket{D_\alpha}\!\bra{D_\beta}\, (U_{\theta}^{\otimes n})^\ast
= e^{i (\alpha - \beta) \cdot \theta} \ket{D_\alpha}\!\bra{D_\beta}.
\]
By using the diagonal unitary invariance of \(X\),
\[
X
= U_{\theta}^{\otimes n} X (U_{\theta}^{\otimes n})^\ast
= \sum_{\alpha,\beta \in \knd{n}{d}} \lambda(X)_{\alpha,\beta}\,
e^{i (\alpha - \beta) \cdot \theta} \,
\ket{D_\alpha}\!\bra{D_\beta}
\qquad
\text{for all } \theta \in [0,2 \pi)^d.
\]
Since the Dicke basis is orthonormal, we must have, for all $\alpha,\beta \in \knd{n}{d}$ and for all $\theta$ in $[0, 2 \pi)^d$ that, 
\[\lambda(X)_{\alpha,\beta}\, =  e^{i (\alpha - \beta) \cdot \theta} \lambda(X)_{\alpha,\beta}. \]
This implies either $\lambda(X)_{\alpha,\beta} = 0$ or \( \forall \theta \in [0, 2 \pi], \, \, e^{i (\alpha-\beta) \cdot \theta } = 1. \) The latter condition holds if and only if \(\alpha = \beta\) implying \(\lambda(X)_{\alpha,\beta} = 0\) whenever \(\alpha \ne \beta\).  
Therefore,
\[
X = \sum_{\alpha \in \knd{n}{d}} \lambda(X)_{\alpha, \alpha} \,\ket{D_\alpha}\!\bra{D_\alpha}
\in \operatorname{span}_{\mathbb{C}}\{\ketbra{D_\alpha}{D_\alpha}\}_{\alpha \in \knd{n}{d}}.
\] 
We can thus conclude:
\begin{equation}
\label{eq:symmetry-in-span}
\mathsf{DS}^{(n)}_d \subseteq \operatorname{span}_{\mathbb{C}}\{\ketbra{D_\alpha}{D_\alpha}\}.
\end{equation} Conversely, assume any operator of the form,
\[
X = \sum_{\alpha \in \knd{n}{d}} \lambda_\alpha\, \ket{D_\alpha}\!\bra{D_\alpha}
\]
is clearly supported on the symmetric subspace, so \(\Pi^{(n)}_d X \Pi^{(n)}_d = X\), and is invariant under all \(U_{\theta}^{\otimes n}\), since $U_{\theta}^{\otimes n} \ket{D_\alpha}\!\bra{D_\alpha}\, (U_{\theta}^{\otimes n})^* = \ket{D_\alpha}\!\bra{D_\alpha}$. Thus, we have the inclusion, 
\begin{equation}
\label{eq:span-in-symmetry}
\operatorname{span}_{\mathbb{C}}\{\ketbra{D_\alpha}{D_\alpha}\}
\subseteq \mathsf{DS}^{(n)}_d.
\end{equation}
By combining \cref{eq:symmetry-in-span} and \cref{eq:span-in-symmetry}, we complete the proof of the first assertion. For the second assertion, note that any state of the form, 
$$\sum_{\alpha \in \knd{n}{d}} \lambda_\alpha\, \ket{D_\alpha}\!\bra{D_\alpha}$$ is self-adjoint if and only if $\lambda_\alpha \in \mathbb{R}$. The rest of the proof is analogous to the complex case. 
\end{proof}

Note that the diagonal symmetry property of Dicke states was first noted in the papers \cite{yu2016separability, wolfe2014certifying, quesada2017entanglement}. By a slight abuse of notation, we write $\lambda(X)_\alpha:= \lambda(X)_{\alpha,\alpha}.$ In this paper, we are interested in studying the quantum states in this subspace, i.e, PSD matrices (with trace $1$). Then, these properties can be captured by imposing positivity conditions and normalization constraints on the coefficients $\lambda(X)_\alpha$.

\begin{proposition}
\label{prop:prop-mods}
    Let $X \in \mathsf{DS}^{(n)}_d$, such that $$X = \sum_{\alpha \in \knd{n}{d}} \lambda(X)_\alpha \ket{D_\alpha} \bra{D_\alpha}.$$ 
    Then, the following equivalences hold,
    \begin{enumerate}
        \item $X \in \mathsf{Herm}^{(n)}_d  \iff \forall \alpha \in \knd{n}{d}, \quad  \lambda(X)_\alpha \in \mathbb{R}$,
        \item $X \in \mathsf{PSD}^{(n)}_d \iff \forall \alpha \in \knd{n}{d}, \quad \lambda(X)_\alpha \in \mathbb{R}_+$.
    \end{enumerate}   
    Moreover, we have for the trace of $X$
    $$\operatorname{Tr}(X) = \sum_{\alpha \in \knd{n}{d}}\lambda(X)_\alpha,$$ 
    and the rank of $X$ 
    $$\operatorname{rk}(X) = |\{\alpha \in \knd{n}{d} \mid \lambda(X)_\alpha \neq 0 \}|.$$

\end{proposition}

 \begin{proof}
The proof follows from the orthogonality of the Dicke states; the spectrum of the normal matrix $X$ is the set $\{\lambda(X)_\alpha\}_{\alpha \in \knd{n}{d}}$.
\end{proof}

It follows from \cref{prop:prop-mods} that $X$ is a quantum state if and only if $\lambda(X)$ is a finite probability distribution over the set $\knd{n}{d}$, and these quantum states are called \emph{mixtures of Dicke states}. These are the states introduced and studied for qubits in \cite{wolfe2014certifying, yu2016separability, quesada2017entanglement}, for bipartite qudit systems \cite{yu2016separability,tura2018separability}, and for multi-partite qudit systems in the recent paper \cite{romero2025multipartite}. Note that for the parametrization in terms of the vectors $\lambda(X)_\alpha$, we get the inner product between two states $X, Y$ in the diagonally symmetric subspace as 

$$\operatorname{Tr}(X^* Y) = \sum_{\alpha \in \knd{n}{d}} \overbar{\lambda(X)}_\alpha \lambda(Y)_\alpha.$$

\begin{remark}
We point out that a class of similar states that have been analysed in the literature for qudits, i.e, the mixtures of D-symmetric states.

\noindent \textbf{D-symmetric states}. A symmetric tensor $Q \in \vee^n \mathrm{R}^d$ is called a Hankel tensor if $\forall i,j \in [d]^n$ such that $i_1 + i_2 \ldots i_n = j_1 + j_2 \ldots j_n$, $Q_i = Q_j.$ In the case of $n = 2$, this reduces to the case of Hankel matrices. Let $i \vdash k$ to denote indices $i$ such that $\sum^n_{k=1} i_k = k.$ Then, we define $$\ket{R_k} = \sum_{i \vdash k} \ket{i}.$$ These class of states are called restricted Dicke states. Then, the class of mixtures of restricted Dicke states (or D-symmetric states) was discussed in 

$$X^{\text{D}} := \sum^{dn}_{k=n} p_k \ketbra{R_k}{R_k}.$$ For $d=2$, this coincides with the mixtures of Dicke states, parametrized by $n+1$ parameters. For higher dimensions, this class is incomparable to the mixtures of Dicke states. For example, even in the bipartite case and $d=3$ consider the projector corresponding to the state $\ket{R_4} = \ket{13} + \ket{31} + \ket{22}.$ This cannot be written in the form of mixtures of Dicke states. On the other hand $\ket{13} + \ket{31}$ is not D-symmetric. Inspite of this, we can consider the projection of $D$-symmetric states into the diagonally symmetric subspace. To this end, note that we can project it into the symmetric subspace by reading the diagonal elements of this state, i.e, $\int_{\diagU} U^{\otimes 2} X \overline{U}^{\otimes 2} dU$ is a diagonally symmetric state with the tensor 

$$Q[\sum^{dn}_{k=n} p_k \ketbra{R_k}{R_k}]_l = \sum^{dn}_{k=n} p_k \sum_{i \vdash k} \delta_{il}$$ which is a Hankel tensor. 

\noindent \textbf{}
\end{remark}

\subsection{Permutation symmetry on local bases}
\label{sec:sd-symmetry}
Besides the bosonic symmetry (the action of $S_n$ on tensor legs) and diagonal unitary invariance, it is sometimes natural to impose an additional invariance under permutations of the \emph{local basis labels} $[d]$.

\begin{definition}[$S_d$-invariant DS states]
\label{def:sd-invariant-ds}
For $\tau \in S_d$, let $V_\tau \in \mathcal{M}_{d}$ be the permutation unitary defined by
$$V_\tau \ket{j} := \ket{\tau(j)} \qquad \qquad \forall j \in [d].$$
We say that a DS matrix $X \in \mathsf{DS}^{(n)}_d$ is \emph{$S_d$-invariant} if
$$\forall \tau \in S_d, \qquad V_\tau^{\otimes n}\, X \, (V_\tau^{\otimes n})^* = X.$$
We denote by $\mathsf{DS}^{(n),S_d}_d \subseteq \mathsf{DS}^{(n)}_d$ the subspace of $S_d$-invariant DS matrices.
\end{definition}

The action of $S_d$ permutes the components of occupation vectors. For $\alpha \in \knd{n}{d}$ and $\tau \in S_d$, define $(\tau \cdot \alpha) \in \knd{n}{d}$ by
$$ (\tau \cdot \alpha)_j := \alpha_{\tau^{-1}(j)} \qquad  \forall j\in[d]. $$
Then, for a DS matrix
$$X=\sum_{\alpha \in \knd{n}{d}} \lambda(X)_\alpha \ket{D_\alpha}\bra{D_\alpha},$$
the condition in \cref{def:sd-invariant-ds} is equivalent to
\begin{equation}
\label{eq:sd-lambda-constraint}
\forall \tau \in S_d,\ \forall \alpha \in \knd{n}{d}, \qquad \lambda(X)_\alpha = \lambda(X)_{\tau \cdot \alpha}.
\end{equation}
In other words, the vector $\lambda(X)$ is constant on $S_d$-orbits, so the resulting family of matrices has fewer parameters than the whole $\mathsf{DS}$ space.

\medskip

Let
$$\mathsf{Par}_d(n) := \{\mu \in \knd{n}{d} \mid \mu_1 \ge \mu_2 \ge \cdots \ge \mu_d \}$$
be the set of integer partitions of $n$ with (at most) $d$ parts. For $\alpha \in \knd{n}{d}$, denote by $\operatorname{sort}(\alpha) \in \mathsf{Par}_d(n)$ the non-increasing rearrangement of the entries of $\alpha$, and by
$$\operatorname{Orb}(\alpha) := \{\tau \cdot \alpha \mid \tau \in S_d\}$$
its orbit. Then $\operatorname{sort}(\alpha)$ is an orbit representative, so the orbits in $\knd{n}{d}$ are in bijection with $\mathsf{Par}_d(n)$.

\begin{proposition}[Extremal points of the $S_d$-invariant DS simplex]
\label{prop:sd-extremal}
Let
$$\mathcal{S}^{(n),S_d}_d := \{\rho \in \operatorname{Herm}[\mathsf{DS}^{(n),S_d}_d] \mid \rho \in \mathsf{PSD}^{(n)}_d,\ \operatorname{Tr}(\rho)=1\}$$
be the convex set of $S_d$-invariant mixtures of Dicke states. Then $\mathcal{S}^{(n),S_d}_d$ is a simplex whose extreme points are indexed by $\mathsf{Par}_d(n)$. Concretely, for $\mu \in \mathsf{Par}_d(n)$ define
$$\rho_\mu := \frac{1}{|\operatorname{Orb}(\mu)|} \sum_{\alpha \in \operatorname{Orb}(\mu)} \ket{D_\alpha}\bra{D_\alpha}.$$
Then $\{\rho_\mu\}_{\mu \in \mathsf{Par}_d(n)}$ are exactly the extreme points of $\mathcal{S}^{(n),S_d}_d$, and every $\rho \in \mathcal{S}^{(n),S_d}_d$ admits a unique convex decomposition
$$\rho = \sum_{\mu \in \mathsf{Par}_d(n)} p_\mu\, \rho_\mu \qquad \text{with } p_\mu \ge 0,\ \sum_{\mu} p_\mu = 1.$$
In particular, $\mathcal{S}^{(n),S_d}_d$ is parameterized by $|\mathsf{Par}_d(n)|-1$ real parameters.
\end{proposition}

\begin{proof}
Any DS quantum state is determined by a probability distribution $\lambda$ on $\knd{n}{d}$ via $\rho=\sum_{\alpha\in\knd{n}{d}}\lambda_\alpha \ket{D_\alpha}\bra{D_\alpha}$. By \cref{eq:sd-lambda-constraint}, $\rho$ is $S_d$-invariant if and only if $\lambda$ is constant on each orbit $\operatorname{Orb}(\alpha)$. Thus specifying $\rho \in \mathcal{S}^{(n),S_d}_d$ is equivalent to specifying a probability distribution $(p_\mu)_{\mu\in\mathsf{Par}_d(n)}$ on the set of orbits, by setting $\lambda_\alpha := p_{\operatorname{sort}(\alpha)}/|\operatorname{Orb}(\operatorname{sort}(\alpha))|$. This identifies $\mathcal{S}^{(n),S_d}_d$ affinely with the standard simplex over $\mathsf{Par}_d(n)$, whose vertices correspond to placing all mass on a single orbit. The corresponding DS states are precisely the orbit-averages $\rho_\mu$.
\end{proof}

\begin{example}[Small-$n$ examples]
\label{ex:sd-small-n}
We give explicit descriptions of $\mathsf{Par}_d(n)$ and the corresponding extreme points $\rho_\mu$ from \cref{prop:sd-extremal}, noting that some partitions only exist when $d$ is large enough.
\begin{enumerate}[label=(\roman*)]
    \item \emph{$n=1$.} For any $d\ge 1$, we have $\mathsf{Par}_d(1)=\{(1,0,\ldots,0)\}$, so $\mathcal{S}^{(1),S_d}_d$ consists of a single state:
    $$\rho_{(1,0,\ldots,0)}=\frac{1}{d}\sum_{j=1}^{d}\ket{j}\bra{j}=\frac{I_d}{d}.$$
    
    \item \emph{$n=2$.} If $d=1$, then $\mathsf{Par}_1(2)=\{(2)\}$. If $d\ge 2$, then
    $$\mathsf{Par}_d(2)=\{(2,0,\ldots,0),\ (1,1,0,\ldots,0)\},$$
    and $\mathcal{S}^{(2),S_d}_d$ is a $1$-simplex with vertices
    \begin{align*}
        \rho_{(2,0,\ldots,0)} &= \frac{1}{d}\sum_{j=1}^{d} \ket{jj}\bra{jj},\\
        \rho_{(1,1,0,\ldots,0)} &= \frac{1}{\binom{d}{2}}\sum_{1\le i<j\le d} \ket{D_{ij}}\bra{D_{ij}},
        \qquad \ket{D_{ij}}:=\frac{\ket{ij}+\ket{ji}}{\sqrt{2}}.
    \end{align*}
    
    \item \emph{$n=3$.} If $d=1$, then $\mathsf{Par}_1(3)=\{(3)\}$. If $d=2$, then $\mathsf{Par}_2(3)=\{(3,0),\ (2,1)\}$. If $d\ge 3$, then
    $$\mathsf{Par}_d(3)=\{(3,0,\ldots,0),\ (2,1,0,\ldots,0),\ (1,1,1,0,\ldots,0)\},$$
    so $\mathcal{S}^{(3),S_d}_d$ is a $2$-simplex with vertices
    \begin{align*}
        \rho_{(3,0,\ldots,0)} &= \frac{1}{d}\sum_{j=1}^{d} \ket{jjj}\bra{jjj},\\
        \rho_{(2,1,0,\ldots,0)} &= \frac{1}{d(d-1)}\sum_{\substack{i,j\in[d]\\ i\ne j}} \ket{D_{iij}}\bra{D_{iij}},
        \qquad \ket{D_{iij}}:=\frac{\ket{iij}+\ket{iji}+\ket{jii}}{\sqrt{3}},\\
        \rho_{(1,1,1,0,\ldots,0)} &= \frac{1}{\binom{d}{3}}\sum_{1\le i<j<k\le d} \ket{D_{ijk}}\bra{D_{ijk}},
        \qquad \ket{D_{ijk}}:=\frac{1}{\sqrt{6}}\sum_{\pi\in S_3}\ket{\pi\cdot(ijk)}.
    \end{align*}
    In particular, the vertex $\rho_{(1,1,1,0,\ldots,0)}$ only exists when $d\ge 3$.
\end{enumerate}
\end{example}

\subsection{Tensor-based parametrization}
As we saw in the last section, the mixtures of Dicke states can be completely parameterized by the $d_{\mathsf{sym}}$ positive real numbers ($d_{\mathsf{sym}}-1$ if we include normalization), the parameters $\lambda(X)_\alpha$. In this section, we complement this with another useful parametrization of the diagonally symmetric subspace in terms of a symmetric tensor. Recall that we have the occupation counting function, $\gamma : [d]^n \rightarrow \knd{n}{d}$ defined component-wise as, $$\big(\gamma(i)\big)_j := |\{k \in [n] \mid i_k = j\}|.$$ 

Note that $\gamma$ is not an invertible map: for example, it is symmetric under permutations 
$$\forall \sigma \in S_n, \, i \in [d]^n, \, \, \gamma(\sigma \cdot i) = \gamma(i).$$
This is reflected in the fact that the input and output spaces have different cardinalities, $d_{\mathsf{sym}} = \binom{n+d-1}{d-1} \neq d^n.$ Nonetheless, for all occupation vectors $\alpha \in \knd{n}{d}$, there is a \emph{unique} canonical multi-index $\idx(\alpha)$ in $\gamma^{-1}(\alpha)$ that satisfies the ordering, $\idx(\alpha)_1 \leq \idx(\alpha)_2 \cdots \leq \idx(\alpha)_n$. Recall that for a multi-index $i \in [d]^n$, the stabilizer subgroup of $S_n$, $\operatorname{stab}(i) := \{\pi \mid P_{\pi} \ket{i} = \ket{i}\}$ is the set of permutations that leave the multi-index $i$ invariant. In the next theorem, we introduce a new parametrization of the $\mathsf{DS}^{(n)}_{d}$ subspace, and prove some of its useful properties.
\begin{theorem}
\label{thm:tensor-param}
Let $X$ be a $n$-partite matrix in $\mathsf{DS}^{(n)}_{d}.$ Then, we define the tensor $Q[X] \in (\mathbb{C}^d)^{\otimes n}$  component-wise as, 

$$\forall i \in [d]^n, \quad Q[X]_{i} := \braket{i \vert X \vert i}.$$ Then, we have the following results, 

\begin{enumerate}
    \item[(1)] $Q[X]$ is a symmetric tensor, i.e, $Q[X] \in \vee^n\mathbb{C}^d$. Moreover, $X \in \operatorname{Herm}[\mathsf{DS}^{(n)}_{d}]$, if and only if $Q[X] \in \vee^n\mathbb{R}^d$, i.e $Q[X]$ is a real symmetric tensor. Moreover, $X \in \mathsf{PSD}^{(n)}_d \iff Q[X] \in \mathsf{NN}^{(n)}_d.$

    \item[(2)] Since $X \in \mathsf{DS}^{(n)}_d$, we have an expansion in the Dicke basis, $$X = \sum_{\alpha \in \knd{n}{d}} \lambda(X)_\alpha \ket{D_\alpha} \bra{D_\alpha}.$$ Then, we have the following connection between the vectors $\lambda(X)$ and the tensor $Q[X]$, 
    
    $$\forall i \in [d]^n \quad \binom{n}{\gamma(i)}Q[X]_i = \lambda(X)_{\gamma(i)}.$$
    Equivalently, for all $\alpha \in \knd{n}{d}$ and all $i \in \gamma^{-1}(\alpha)$, we have
    \begin{equation}\label{eq:p-to-Q}
        \binom{n}{\alpha}\,Q[X]_i = \lambda(X)_{\alpha}.
    \end{equation}

    \item[(3)] $X$ has the following expansion, $$X = \sum_{i \in [d]^n} \sum_{\pi \in S_n}  \frac{\binom{n}{\gamma(i)}}{n!} Q[X]_{i} P_{\pi}\ket{i}\bra{i}.$$
    
    \item[(4)] The trace of the matrix $X$, $$\operatorname{Tr}(X) = \sum_{i \in [d]^n} Q[X]_{i} = \sum_{\alpha \in \knd{n}{d}} \lambda(X)_\alpha.$$
\end{enumerate}
\end{theorem}
\begin{proof}
Let \(X\in\mathsf{DS}^{(n)}_{d}\), and \(Q[X]\in\mathbb{C}^{\otimes n}_d\) is defined component-wise by,
\[
\forall i\in[d]^n,\qquad Q[X]_i:=\langle i\vert X\vert i\rangle.
\]

\begin{enumerate}
\item[(1)]
Let \(\pi\in S_n\) be any permutation and let \(P_\pi\) denote the corresponding permutation operator on \((\mathbb{C}^d)^{\otimes n}\). Since \(X\in\mathsf{DS}^{(n)}_{d}\), by \cref{def:DS-subspace}, we have \(P_\pi X = X P_\pi = X\) for every \(\pi\). For any multi-index \(i\in[d]^n\) and some $\pi \in S_n$,
\[
Q[X]_{\pi\cdot i}
= \langle \pi\cdot i\vert X\vert \pi\cdot i\rangle
= \langle i\vert P_{\pi^{-1}} X P_\pi\vert i\rangle
= \langle i\vert X\vert i\rangle
= Q[X]_i.
\]
Therefore, \(Q[X]\) is invariant under the action of \(S_n\) and hence is a symmetric tensor, \(Q[X]\in\vee^n\mathbb{C}^d\). If additionally, \(X\) is Hermitian then each diagonal matrix element is real \(\langle i\vert X\vert i\rangle \in \mathbb{R} \), hence \(Q[X]\in\vee^n\mathbb{R}^d\).

\item[(2)] Let \(X\in\mathsf{DS}^{(n)}_d\) have the decomposition in the Dicke basis:
\[
X=\sum_{\alpha\in \knd{n}{d}} \lambda(X)_\alpha\;|D_\alpha\rangle\langle D_\alpha|,
\] For any computational basis vector \(|i\rangle\), 
\[
\langle i\vert D_\alpha\rangle = \frac{1}{\sqrt{\binom{n}{\alpha}}} \mathbb{1}_{\gamma(i) = \alpha}
\] Hence, \begin{align*}
Q[X]_i = \langle i\vert X\vert i\rangle
= \sum_{\alpha \in\knd{n}{d}} \lambda(X)_\alpha\,|\langle i\vert D_{\alpha}\rangle|^2 = \lambda(X)_{\gamma(i)}\,\frac{1}{\binom{n}{\gamma(i)}}.
\end{align*}
Rearranging gives the claimed identity

\begin{equation}
\label{eq:Q-to-p}
\binom{n}{\gamma(i)}\,Q[X]_i = \lambda(X)_{\gamma(i)},
\end{equation}
which shows that the Dicke-basis coefficients \(\lambda(X)_\alpha\) are constant on each orbit \(\gamma^{-1}(\alpha)\). 
\item[(3)]  
Since $X \in \mathsf{DS}^{(n)}_d,$ it has the expansion in the Dicke basis. 
\[
X=\sum_{\alpha \in \knd{n}{d}} \lambda(X)_{\alpha}\;|D_\alpha\rangle\langle D_\alpha|.
\] Then, by using \cref{eq:p-to-Q}, and by expanding the Dicke state, $\ket{D_\alpha} = \frac{1}{\sqrt{\binom{n}{\alpha}}}\sum_{\substack{i \in [d]^n  \\\gamma(i) = \alpha}} \ket{i},$ we get
\begin{equation*}
X = 
\sum_{\alpha \in \knd{n}{d}}
    Q[X]_{\idx(\alpha)}
    \sum_{\substack{i,j \in [d]^n \\ \gamma(i)=\alpha,\;\gamma(j)=\alpha}}
        \ket{i}\bra{j}
\end{equation*} By re-labeling the indices $l := \idx(\alpha)$, we have the following equation,   
\begin{equation*}
    X = \sum_{\substack{l \in [d]^n \\ l_1 \le \cdots \le l_n}}
    Q[X]_{l}
    \sum_{\substack{i,j \in [d]^n \\ \gamma(i)=\gamma(l),\;\gamma(j)=\gamma(l)}}
        \ket{i}\bra{j}
\end{equation*} Now $\gamma(i)=\gamma(l) \iff \exists \pi \in S_n, \, i = \pi \cdot l$. Therefore, our expression reduces to,
\begin{equation*}
\sum_{\substack{l \in [d]^n \\ l_1 \le \cdots \le l_n}}
    Q[X]_{l}\;
    \frac{\binom{n}{\gamma(l)}^2}{n!^2}
    \sum_{\pi,\sigma \in S_n}
        \ket{\pi \cdot l}\bra{\sigma \cdot l}
\end{equation*} where we divide by $|\operatorname{stab}(l)|^2 = \frac{{n!^2}}{\binom{n}{\gamma(l)}^2}$ to avoid overcounting. By replacing ordered indices, $l_1 \le \cdots \le l_n$ by all indices $l \in [d]^n,$ we divide by ${\binom{n}{\gamma(l)}},$ \begin{equation*}
\sum_{l \in [d]^n}
    Q[X]_{l}\;
    \frac{\binom{n}{\gamma(l)}}{n!^2}
    \sum_{\pi,\sigma \in S_n}
        \ket{\pi \cdot l}\bra{\sigma \cdot l}
\end{equation*} Finally, by relabeling permutations, we have, \begin{equation*}
\sum_{l \in [d]^n}
    Q[X]_{l}\;
    \frac{\binom{n}{\gamma(l)}}{n!}
    \sum_{\pi' \in S_n}
        \ket{\pi' \cdot l}\bra{l}.  
\end{equation*} which is the claimed expression.

\item[(4)] 
By definition of \(Q[X]\),
\[
\operatorname{Tr}(X)=\sum_{j\in[d]^n}\langle j\vert X\vert j\rangle=\sum_{j\in[d]^n} Q[X]_j.
\]
Recall that from \cref{prop:prop-mods}, $\operatorname{Tr}(X) = \sum_{\alpha \in \knd{n}{d}} \lambda(X)_\alpha$. This concludes the proof.
\end{enumerate}

\end{proof}

Note that in the previous \cref{thm:tensor-param}, we establish a new tensor parametrization of the diagonally symmetric subspace. In particular, it shows that the relevant parameters for a diagonally symmetric matrix $X$ are just entries on the diagonal of the state, which can be arranged as a symmetric tensor $Q[X]$.  This parametrization will be quite useful for studying the entanglement properties in the remaining sections. We end this section by  characterizing the Hilbert-Schmidt (HS) inner product between two diagonally symmetric matrices in terms of the tensor parametrization.

\begin{remark}[$S_d$-symmetry in the tensor parametrization]
\label{rem:sd-tensor}
Let $X \in \mathsf{DS}^{(n),S_d}_d$ and define $Q[X]$ as in \cref{thm:tensor-param}. For $\tau \in S_d$ and $i=(i_1,\ldots,i_n)\in[d]^n$, write $\tau(i):=(\tau(i_1),\ldots,\tau(i_n))$. Then
$$\forall \tau \in S_d,\ \forall i \in [d]^n, \qquad Q[X]_{\tau(i)} = Q[X]_i.$$
Equivalently, $Q[X]_i$ depends only on $\operatorname{sort}(\gamma(i))\in\mathsf{Par}_d(n)$, so an $S_d$-invariant DS tensor has $|\mathsf{Par}_d(n)|$ independent diagonal values (and $|\mathsf{Par}_d(n)|-1$ after normalization to a state). Note that the action of $S_d$ on indices $i$ defined above is different that the action of $S_n$ on the same set, which ensures that the tensor $Q[X]$ is symmetric.
\end{remark}

\begin{proposition}
\label{prop:inconsistency-inner-prod}
The Hilbert-Schmidt Inner product between the two DS matrices $X$ and $Y$ satisfies,
\begin{equation}
\operatorname{Tr}(X^* Y) = \sum_{\substack{i \in [d]^n \\ i_1 \leq i_2 \ldots \leq i_n}} \overbar{Q[X]_{i}} Q[Y]_{ i} \binom{n}{\gamma({i})}^2 = \sum_{i \in [d]^n} \overbar{Q[X]_{i}} Q[Y]_{i} \binom{n}{\gamma({i})}.
\end{equation}
\end{proposition}

\begin{proof}
    Recall that for two DS matrices, $$\operatorname{Tr}(X^* Y) = \sum_{\alpha \in \knd{n}{d}} \overbar{\lambda (X)}_\alpha \lambda (Y)_\alpha.$$
    Using \cref{thm:tensor-param}, it follows that $\lambda (X)_\alpha = Q_{\idx(\alpha)} \binom{n}{\alpha}$ hence the previous expression can be reduced to
    $$\operatorname{Tr}(X^* Y) = \sum_{\alpha \in \knd{n}{d}} \overbar{Q[X]}_{\idx(\alpha)} Q[Y]_{\idx(\alpha)} \binom{n}{\alpha}^2 = \sum_{\substack{i \in [d]^n \\ i_1 \leq i_2 \ldots \leq i_n}} \overbar{Q[X]}_{i} Q[Y]_{i} \binom{n}{\gamma({i})}^2.$$ By replacing the ordered indices with unordered indices, we get
    $$\operatorname{Tr}(X^* Y) = \sum_{i \in [d]^n} \overbar{Q[X]_{i}} Q[Y]_{ i} \binom{n}{\gamma({i})}.$$
    \end{proof}

Note that \cref{prop:inconsistency-inner-prod} shows that the inner product of two matrices parametrized using the diagonal tensor $X \leftrightarrow Q[X]$ is \emph{not} equal to the Euclidean inner product between two tensors corresponding to the matrices, i.e,  $ \operatorname{Tr}(X Y) \neq \braket{Q[X], Q[Y]}$. This leads to an inconsistency while looking at duals, as in entanglement theory, the duals are usually considered with respect to the Hilbert-Schmidt inner product, while for tensors, the Euclidean product is standard. To solve this, we introduce another related parametrization of the DS subspace.

\begin{definition}
\label{def:parametrize-witness}
    Let $X \in \mathsf{DS}^{(n)}_d.$ Then, define the symmetric tensor, 
    $$W[X]_{i} := {Q[X]_i}{\binom{n}{\gamma(i)}}  = \braket{i \vert X \vert i}{\binom{n}{\gamma(i)}}.$$ 
\end{definition}
Let us consider the following map, $\operatorname{diag} :\vee^n \mathbb{R}^d \mapsto \mathcal{M}(\vee^n \mathbb{C}^d)$, defined as,  

$$\operatorname{diag}[Q] := \sum_{i \in [d]^n} Q_{i} \ket{i}\bra{i}.$$ This maps a symmetric real tensor to a diagonal matrix in $\M{d}^{\otimes n}.$ Note that this matrix is not bosonic, i.e, it is not supported only on the symmetric subspace (if all the entries of the tensor $Q$ are non-zero, then $\operatorname{diag}[Q]$ is supported on the whole tensor space).

\begin{proposition}
\label{prop:X-expansion-W}
Let $X \in \operatorname{Herm}[\mathsf{DS}^{(n)}_d].$ Then, we have the following equality: 
    $$X = \Pi^{(n)}_{d} \operatorname{diag}[W[X]] \Pi^{(n)}_{d}$$
\end{proposition}
\begin{proof}
    By using \cref{thm:tensor-param}, we have  $$X = \sum_{i \in [d]^n} \sum_{\pi \in S_n} P_{\pi} \frac{\binom{n}{\gamma(i)}}{n!} Q[X]_{i} \ket{i}\bra{i}.$$
    Hence, we have 
     \begin{align*}X = \sum_{i \in [d]^n} \sum_{\pi \in S_n} \frac{1}{n!} P_{\pi} W[X]_i \ket{i}\bra{i} = \Pi^{(n)}_{d} \sum_{i \in [d]^n} W[X]_i \ket{i}\bra{i} &= \Pi^{(n)}_{d} \operatorname{diag}[W[X]] \\ &=  \Pi^{(n)}_{d} \operatorname{diag}[W[X]] \Pi^{(n)}_{d}
     \end{align*}
     where the last equality follows from the self-adjointness of $X$.
\end{proof}

In the remaining sections, we use these parametrizations for the set of diagonally symmetric states, as well as diagonally symmetric entanglement witnesses. Note that by \cref{prop:inconsistency-inner-prod}, we now have 
\begin{equation}\label{eq:HS-Q-W}
    \braket{X, Y}_{\mathcal{B}(\vee^n \mathbb{C}^d)} = \braket{Q[X], W[Y]}_{\vee^n\C{d}}.
\end{equation}
Hence, the Hilbert-Schmidt product now corresponds exactly to the Euclidean product of tensors in the $Q$ and $W$ parametrizations. The fact that the Hilbert-Schmidt product is the same as the Euclidean product in the $Q$ and $W$ parametrizations will be used in later section to relate the duality of bosonic matrices with diagonal symmetry to that of the corresponding $Q$ and $W$ parametrizations.

\subsection{Reduced states and marginals}
From the previous section, we know that the diagonally symmetric matrices $X$ are parameterized with a symmetric tensor $Q[X]$. The next theorem shows that the quantum marginals (reduced density matrices) of these states have a direct correspondence with the ``classical'' marginals of this tensor. First, recall that the marginal of a symmetric tensor $T \in \vee^n \mathbb R^d$, with respect to a set of indices $A \subseteq [n]$ is given by
    $$\forall j \in [d]^{n-|A|} \qquad [\operatorname{tr}_A(T)]_j := \sum_{i \in [d]^{|A|}} T_{i \sqcup j}.$$
    Note that, for symmetric tensors, the marginal depends only on the cardinality of the set $A$. 

\begin{theorem}
\label{thm:marginals-dicke-states}
    Let $X \in \mathsf{DS}^{(n)}_d$ be a diagonally symmetric matrix. For any $A \subseteq [n]$, let $\operatorname{Tr}_{A} (X)$ denote the marginal of the state over the tensor factors in $A$. We have then
    \begin{itemize}
        \item $\operatorname{Tr}_{A} (X) \in \mathsf{DS}^{(n - |A|)}_d$;
        \item $Q[\operatorname{Tr}_{A} (X)] = \operatorname{tr}_{A}(Q[X])$.
    \end{itemize}
    
\end{theorem}
\begin{proof}
To prove the first assertion, we show that the reduced states also lie in the diagonally symmetric subspace. Firstly, we can show that any reduced bosonic state satisfies
    $$\forall \sigma \in S_m, \quad P_\sigma \operatorname{Tr}_{A}(X)  = \operatorname{Tr}_{A} \big((P_\sigma \otimes \mathbb{I}_{|A|})(X) \big) = \operatorname{Tr}_{A}(X),$$ 
    where the last equality follows from the fact that $X$ is bosonic. Moreover, we have, 
    $$\forall U \in \mathcal{DU}_d \quad  U^{\otimes (n-|A|)}\operatorname{Tr}_{A}(X) (U^{\otimes (n-|A|)})^* = \operatorname{Tr}_{A}(U^{\otimes n} X (U^{\otimes n})^*) = \operatorname{Tr}_{A}(X),$$ 
    where the first equality follows from the cyclicity of the trace and the second follows from the diagonal invariance of the matrix $X$. The second assertion follows from the direct computation:
    \begin{align*}
        \forall i \in [d]^{n-|A|}, \quad Q[\operatorname{Tr}_A(X)]_{i} &= \braket{i | \operatorname{Tr}_A(X)| i} =  \operatorname{Tr}(X \ketbra{i}{i} \otimes \mathbb{I}_A) \\
        &= \sum_{j \in [d]^A} \operatorname{Tr}(X \ketbra{i, j}{i,j})\\
        &= \sum_{j \in [d]^A} Q[X]_{i \sqcup j} = (\operatorname{tr}_A Q[X])_i.
    \end{align*}
    Above, we use the notation $i \sqcup j$ to denote the vector
    $$(i_1, \ldots, i_{n-|A|},j_1, \ldots, j_{|A|}) \in [d]^n.$$
\end{proof}

We use the previous result to compute the marginals of the pure Dicke states $\ket{D_\alpha}$. By \cref{thm:marginals-dicke-states}, the marginals lie in the DS subspace. We compute the marginals for two reasons: firstly, it shows that the approach of tensors in terms of the tensor parametrization is highly convenient, and secondly, the entanglement properties of these marginals have been of independent interest \cite{szalay2025dicke}. 
    Let $\ket{D_\alpha}$ be a pure Dicke state, defined as 
    $$\ket{D_\alpha} = \frac{1}{\binom{n}{\alpha}^{1/2}}\sum_{i \, : \, \gamma(i) = k} \ket{i}.$$ The tensor parametrization of this state is as follows:
    $$\forall i \in [d]^n \quad Q[\ket{D_\alpha}\bra{D_\alpha}]_i= |\braket{i | D_\alpha}|^2 = \frac{1}{\binom{n}{\gamma(i)}} \, \mathbb{1}_{\gamma(i)=\alpha}.$$ 
    We compute next the marginal of the pure Dicke state using the marginal of the tensor, by using \cref{thm:marginals-dicke-states}.

\begin{proposition}\label{prop:marginal-dicke-pure}
    For an arbitrary Dicke state $\ket{D_\alpha}$ and a marginal set $A \subseteq [n]$, we have:
    $$\forall j \in [d]^{n-|A|}, \qquad Q[\operatorname{Tr}_A(\ket{D_\alpha}\bra{D_\alpha})]_j = \operatorname{tr}_{A} [Q(\ket{D_\alpha}\bra{D_\alpha})]_j = \frac{\binom{|A|}{\alpha-\gamma(j)}}{\binom{n}{\alpha}} \mathbb{1}_{\gamma(j) \le \alpha}.$$ 
\end{proposition}
 \begin{proof}   
For the marginal of the tensor, we have:
\begin{align*}
[\operatorname{tr}_{A} (Q[\ket{D_\alpha}\bra{D_\alpha}])]_j
&= \sum_{i \in [d]^{|A|}} (Q[\ket{D_\alpha}\bra{D_\alpha}])_{i \sqcup j}
= \sum_{i \in [d]^{|A|}}
\frac{1}{\binom{n}{\alpha}} \mathbb{1}_{\gamma(i \sqcup j)=\alpha}\\
&= \frac{1}{\binom{n}{\alpha}} \sum_{i \in d^{|A|}} \mathbb{1}_{\gamma(i)+\gamma(j)=\alpha}\\
&=
\frac{\binom{|A|}{\alpha-\gamma(j)}}{\binom{n}{\alpha}}
\,\mathbb{1}_{\gamma(j)\le \alpha}.
\end{align*} 

\end{proof}

\section{Entanglement and PPT property}
\label{sec:entanglement-mods}
In this section, we show an exact correspondence between the entanglement properties of mixtures of Dicke states to some well-known positivity property of symmetric tensors; our findings are summarized in \cref{fig:entanglement-DS-tensors}. We then completely characterize for which local Hilbert space dimension and number of particles there exist PPT entangled DS states.  Finally, we show in the last subsection that all marginals of entangled pure Dicke states have negative partial transpose across all bipartitions, recovering a recent result of \cite{szalay2025dicke} in a simpler, more conceptual way. By establishing these results, we have a comprehensive framework for entanglement in mixtures of Dicke states, which we use to study PPT entanglement and bosonic extendibility in the later sections using well-known ideas and results about positive polynomials and 
tensors. 

\subsection{Bose-symmetry and entanglement}
We begin this section with a brief review of multipartite entanglement in mixed states, with the focus on the particular results about bosonic states. Note the following two examples of entangled states in $3$ qubits.
$$(\ket{00} + \ket{11}) \otimes \ket{1} \qquad \ket{000} + \ket{111}$$
Although both states are entangled, there is a difference between the structure of entanglement. The first state is entangled in a bipartite fashion, while for the GHZ state, no such decomposition into the bipartite entanglement exists. This motivates the following definition.

\begin{definition}[\cite{ichikawa2008exchange,marconi2025symmetric}]
Let $\mathcal{H} = \bigotimes_{i=1}^n \mathcal{H}_i$ be the Hilbert space of an $n$-partite quantum system.
A pure state $\lvert \psi \rangle \in \mathcal{H}$ is called \emph{$k$-producible}
if there exists a partition of $[n]$ into disjoint subsets $S_1, \ldots, S_m$ such that $\lvert S_j \rvert \le k$ for all $j$ and
\[
    \lvert \psi \rangle = \bigotimes_{j=1}^m \lvert \phi_j \rangle ,
\]
where $\lvert \phi_j \rangle$ is a pure state on the subsystems indexed by $S_j$. A mixed state $X \in \mathcal{M}^{\otimes n}_d$ is called \emph{$k$-producible} if,  
\[
    X \in \operatorname{cone} \{ 
    \lvert \psi \rangle \langle \psi \rvert \mid \lvert \psi \rangle \text{ is $k$-producible}\}.
\]
The entanglement depth of a state $\rho$ is the smallest $k$ such that $\rho$ is
$k$-producible. A state is \emph{genuinely $(k+1)$-partite entangled} if it is
$(k+1)$-producible but not $k$-producible.
\end{definition}

It is clear from the previous definition that bipartite quantum states are either separable (i.e. have entanglement depth $1$) or $2$-partite entangled. In comparison to bipartite systems, entanglement theory in the multipartite case is much more complicated.  Surprisingly, for \emph{multipartite} bosonic states, a stronger result holds.

\begin{theorem}[{{\cite[Theorem 1]{ichikawa2008exchange}}}]
\label{thm:bosonic-sep-or-ent}
    Let $X \in \operatorname{Herm}[\vee^n \mathbb{C}^d]$. Then, the state has entanglement depth $1$ or $n$. Particularly, if it has entanglement depth $1$, there exist $\{\ket{v_k}\}^K_{k=1}$ such that, 
    $$X = \sum^K_{k=1} \ketbra{v_k}{v_k}^{\otimes n}.$$
\end{theorem} We denote all separable bosonic states (entanglement depth = 1) as \begin{equation}\label{eq:separability}\mathsf{Sep}^{(n)}_d := \operatorname{cone}\{\ketbra{v_k}{v_k}^{\otimes n} \mid \ket{v_k} \in \mathbb{C}^d\},\end{equation}which forms a proper convex cone. We define the set of bosonic entanglement witnesses as the dual cone of the separable cone

$$\mathsf{EW}^{(n)}_d := \big(\mathsf{Sep}^{(n)}_d \big)^\circ.$$

From the form of the extremal rays of the cone $\mathsf{Sep}^{(n)}_d$, we obtain the following characterization of entanglement witnesses in the DS subspace.

\begin{proposition}
\label{prop:ew-dual-sep}
 Let $O \in \operatorname{Herm}[\vee^n \mathbb{C}^d].$ Then, $O \in \mathsf{EW}^{(n)}_d$ if and only if $$\forall \ket{v} \in \mathbb{C}^d \quad \braket{v^{\otimes n}|O| v^{\otimes n}} \geq 0.$$
\end{proposition}

Recall that all \emph{bipartite} separable states satisfy the PPT (Positivity under Partial Transposition) property, i.e, PPT is a necessary condition for separability. The separable bosonic states (ref. \cref{thm:bosonic-sep-or-ent}) are separable across any bipartition; hence, they also satisfy the PPT property across \emph{each} such partition. 
In general, a bosonic state might only satisfy the PPT property across some bipartitions, and not all. Let $I$ be a subset of $[n].$ We use $\Gamma_I$ to denote the partial transpose of the systems indexed by $I.$ We can use bosonic symmetry to show that the PPT property only depends on the size of the bipartition we consider while applying the partial transpose, and not on the particular systems that are transposed.

\begin{proposition}
\label{prop:ppt-symmetric}
Let $\rho \in \operatorname{Herm}[\vee^n \mathbb{C}^d]$ be a bosonic matrix. Let $I$ and $J$ be arbitrary subsets of $[n]$ such that $|I| = |J|.$ Then, 
$$\rho^{\Gamma_{I}} \geq 0 \iff \rho^{\Gamma_{J}} \geq 0.$$ 
\end{proposition}

\begin{proof}
    For a bosonic state, we have, $$\forall \pi \in S_n \, \,  P_{\pi} \rho P_{\pi^{-1}} = \rho.$$ Let $\pi$ be a permutation that maps the elements of the set $I$ to $J$. Then, 
    $$ \rho^{\Gamma_I} = (P_{\pi} \rho P_{\pi^{-1}})^{\Gamma_I} = (P_{\pi} \rho^{\Gamma_J} P_{\pi^{-1}}).$$ 
\end{proof}
Therefore, following the previous \cref{prop:ppt-symmetric}, consider the following convex cones.
\begin{definition}
\label{def:ppt-and-psd}
Let $0 \leq k \leq \lfloor n/2 \rfloor.$ We define the following convex cones,
    $$\mathsf{PPT}_d^{(n,0)} := \{\rho \in \operatorname{Herm}[\vee^n \mathbb{C}^d] \mid \rho \geq 0\}$$
    $$\mathsf{PPT}_d^{(n,k)} := \{\rho \in \operatorname{Herm}[\vee^n \mathbb{C}^d] \mid \rho^{\Gamma_{[k]}} \geq 0\}.$$
\end{definition}

Note that $\mathsf{PPT}^{(n,0)}$ is the bosonic PSD matrices, and $\mathsf{PPT}^{(n,k)}$ is the set of bosonic matrices having the PPT property across the bipartition $[k]:[n] \backslash [k]$. Clearly, all the separable states (see \cref{eq:separability}, \cref{thm:bosonic-sep-or-ent}) satisfy the following inclusion. 

\begin{equation}
\label{eq:sep-in-ppt}
\mathsf{Sep}^{(n)}_d  \subseteq \bigcap^{\lfloor n/2 \rfloor}_{k=0} \mathsf{PPT}_d^{(n,k)} =: \mathsf{PPT}_d^{(n)}.
\end{equation}This means that the separable states satisfy the PPT property across all the bipartitions. 

\begin{remark}
    The set inclusion in \cref{eq:sep-in-ppt} is strict except for the case $d=2, n=2,3$. This has been shown by explicit construction of PPT entangled states for $n \geq 4$ qubits. \cite{tura2012four,augusiak2012entangled}.
\end{remark}

We now introduce the dual cone of the PPT sets. Note that we will consider the dual cone in the ambient space of  $\operatorname{Herm}[\vee^n \mathbb{C}^d] \subseteq \operatorname{Herm}[\mathcal{M}^{\otimes n}_d].$ We define the following convex cone of decomposable bosonic matrices,
$$\mathsf{Dec}_d^{(n,k)} := \big(\mathsf{PPT}_d^{(n,k)})^\circ.$$ We also have the following characterization of the $\mathsf{Dec}_d^{(n,k)}.$

\begin{theorem}[]
\label{thm:dec-characterization}
$O \in \mathsf{Dec}^{(n,k)}_d$ if and only if there exist positive semi-definite operators, 
\(P \in \operatorname{Herm}[\mathcal{M}^{\otimes n}_d]\) such that,
\begin{equation}
    O \;=\; \Pi^{(n)}_d (P^{\Gamma_{[k]}})  \Pi^{(n)}_d.
\end{equation}
\end{theorem}
\begin{proof}
    We will show the proof in two steps. 
    ($\impliedby$) Consider that there exist positive-semidefinite matrix $P \in \operatorname{Herm}[\mathcal{M}^{\otimes n}_d]$ such that, 
    $$O \;=\; \Pi^{(n)}_d (P^{\Gamma_{[k]}}) \Pi^{(n)}_d.$$
    Then for all $\rho \in \mathsf{PPT}_d^{(n,k)}$, we have,  
    \begin{align*}
    \operatorname{Tr}(O\rho)
    = \operatorname{Tr}\!\left(\Pi^{(n)}\big(P^{\Gamma_{[k]}}\big)\Pi^{(n)}_d \rho \right) &= \operatorname{Tr}\!\left(\big(P^{\Gamma_{[k]}}\big)\Pi^{(n)}_d \rho \Pi^{(n)}_d\right) \\
    &= \operatorname{Tr}\!\left(P^{\Gamma_{[k]}} \rho\right) \\
    &= \operatorname{Tr}\!\left(P\, \rho^{\Gamma_{[k]}}\right) \ge 0 \, .
\end{align*} This shows that $O \in (\mathsf{PPT}_d^{(n,k)})^{\circ} = \mathsf{Dec}_d^{(n,k)}.$ 

$(\implies)$ To show the reverse direction, assume that $O \in \mathsf{Dec}_d^{(n,k)}.$ Then, note that, \begin{equation}
\label{eq:ppt-intersection}\mathsf{PPT}_d^{(n,k)} = \{\rho \in \operatorname{Herm}[\mathcal{M}^{\otimes n}_d] \mid \rho^{\Gamma_{[k]}} \geq 0\} \cap \{\rho \in \operatorname{Herm}[\mathcal{M}^{\otimes n}_d] \mid \Pi^{(n)}_d \rho \Pi^{(n)}_d = \rho\}.\end{equation} Consider the convex cone $(\mathsf{PPT}_d^{(n,k)})^{\circ}$ as dual in the larger ambient space $\operatorname{Herm}[\mathcal{M}^{\otimes n}_d].$ Therefore, by the properties of the dual, we get (for closed cones),

$$(\mathsf{PPT}_d^{(n,k)})^{\circ} = \operatorname{cl}\big(\{X \in \operatorname{Herm}[\mathcal{M}^{\otimes n}_d] \mid X^{\Gamma_{[k]}} \geq 0\} + \{X \in \operatorname{Herm}[\mathcal{M}^{\otimes n}_d] \mid \Pi^{(n)}_d X \Pi^{(n)}_d = X\}^\bot\big).$$
This equality is because in \cref{eq:ppt-intersection}, the first cone is self-dual, while the second cone is a subspace; moreover, the dual cone of an intersection is the closure of the sum.
Also note that if $O \in (\mathsf{PPT}_d^{(n,k)})^{\circ}$ where dual is considered in $\operatorname{Herm}[\vee^n \mathbb{C}^d]$, then, it is an element of the dual with respect to $\operatorname{Herm}[\mathcal{M}^{\otimes n}_d].$ Using this, we have the sequence of operators $\{P_l\}_{l \geq 1}, \{Q_l\}_{l\geq 1}$,  

$$O = \operatorname{lim}_{l \to \infty} \tilde P_l + Q_l$$ where $\forall l, \quad \tilde P_l^{\Gamma_{[k]}} \geq 0$ and $\Pi^{(n)}_d Q_l \Pi^{(n)}_d = 0.$ Denote $P_l := \tilde P^{\Gamma_{[k]}}_l$. Therefore,

$$O = \Pi^{(n)}_d O \Pi^{(n)}_d = \Pi^{(n)}_d \operatorname{lim}_{l \to \infty} (P^{\Gamma_{[k]}}_l + Q_l) \Pi^{(n)}_d =  \Pi^{(n)}_d \operatorname{lim}_{l \to \infty} (P^{\Gamma_{[k]}}_l) \Pi^{(n)}_d =  \Pi^{(n)}_d (P^{\Gamma_{[k]}}) \Pi^{(n)}_d$$ where $P := \operatorname{lim}_{l \to \infty} P_l.$ Also $P \geq 0$ as the cone of PSD matrices is closed.
\end{proof}

\begin{remark}
By the previous result in \cref{thm:symmetric-subspace-proj}, we can also show easily, 
\begin{align*}
\forall z \in \mathbb{C}^d \qquad 
\braket{z^{\otimes n} \vert O \vert z^{\otimes n}}
&= 
\braket{z^{\otimes n} \vert P^{{\Gamma}_{[k]}} \vert z^{\otimes n}} = 
\braket{\overbar{z}^{\otimes k} \otimes z^{\otimes n-k} \vert P \vert \overbar{z}^{\otimes k} \otimes z^{\otimes n-k}} \ge 0.
\end{align*}
This also shows easily that $\mathsf{Dec}^{(n,k)}_d \subseteq \mathsf{EW}^{(n)}_d.$
\end{remark}

\begin{corollary}
\label{cor:duality-ppt-dec}
We have the following conic-duality:
    \begin{align*}
        \Bigg(\bigcap^{\lfloor n/2 \rfloor}_{k=0} \mathsf{PPT}_d^{(n,k)} \Bigg)^\circ &= \operatorname{cl}\big(\sum^{\lfloor n/2 \rfloor}_{k=0} \mathsf{Dec}_d^{(n,k)}\big)=: \mathsf{Dec}^{(n)}_d,
    \end{align*}
    where the latter set can be described explicitly as
    $$\mathsf{Dec}^{(n)}_d = \Big\{O \in \operatorname{Herm}[\mathcal{M}^{\otimes n}_d] \mid \exists \{P_k \succeq 0\}^{\lfloor n/2 \rfloor}_{k=0}, \, O \;=\; \Pi^{(n)}_d \Big(\sum^{\lfloor n/2 \rfloor}_{k=0} P_k^{\Gamma_{[k]}} \Big) \Pi^{(n)}_d \Big\}.$$
\end{corollary}

We have the following inclusions and dualities of convex cones in the bosonic subspace.  
\bigskip

\[
\begin{array}{ccccccccc}
\tikzmarknode{sep}{\mathrm{Sep}_d^{(n)}} & \subseteq &
\tikzmarknode{pptk}{\mathrm{PPT}_d^{(n,k)}} & \subseteq &
\tikzmarknode{ppt0}{\mathrm{PPT}_d^{(n,0)}} & \subseteq &
\tikzmarknode{dec}{\mathrm{Dec}_d^{(n,k)}} & \subseteq &
\tikzmarknode{ew}{\mathrm{EW}_d^{(n,0)}}
\end{array}
\]

\begin{tikzpicture}[overlay, remember picture,
  every path/.style={<->, thick, {Stealth[length=2mm]}-{Stealth[length=2mm]}}]

  \draw[red, rounded corners=5pt]
    (sep.north)
    -- ++(0, 0.8cm)
    -- ([yshift=0.8cm]ew.north)
    -- (ew.north);

  \draw[blue, rounded corners=5pt]
    (pptk.north)
    -- ++(0, 0.4cm)
    -- ([yshift=0.4cm]dec.north)
    -- (dec.north);

\end{tikzpicture}

Having established the basics of bosonic entanglement, we move on to studying entanglement in the diagonally symmetric (DS) subspace. By using the tensor parametrization (ref. \cref{thm:tensor-param}) of this subspace, we establish an exact correspondence of the entanglement properties in the DS subspace with those of tensors.

\subsection{Separability and Entanglement Witnesses}
In this section, we establish an exact correspondence between the separability of the diagonally symmetric state $X \in \mathsf{DS}^{(n)}_d$ and the convex properties of the tensor $Q[X]$. We accomplish this by extending the known results about the bipartite separability problem in mixtures of Dicke states.
We now state and prove the main result of this section: the separability of a bosonic diagonally symmetric matrix $X$ is equivalent to the complete positivity of the associated tensor $Q[X]$.

\begin{theorem}
\label{thm:sep-and-cp}
Let $X \in \mathsf{DS}^{(n)}_d$ be a diagonally symmetric matrix. Then:
$$X \in \mathsf{Sep}^{(n)}_{d} \iff Q[X] \in \mathsf{CP}^{(n)}_{d}.$$
\end{theorem}
\begin{proof}
($\implies$) To show the first implication, we assume the state $X$ is separable, i.e, there exist vectors $\{\ket{x_q} \in \mathbb{C}\}^K_{q=1}$, such that,  
$$X = \sum^K_{q=1} {\ketbra{x_q}{x_q}^{\otimes n}}.$$ By \cref{thm:tensor-param}, we have $\forall i \in [d]^n, \,  Q[X]_i := \braket{i \vert X \vert i}$. Hence, we have, by the separable decomposition of $X$, 
\begin{equation}
\forall i \in [d]^n \quad Q[X]_{i} = \braket{i \vert X\vert i} = \sum^K_{q=1} |(x_q)_{i_1}|^2 |(x_q)_{i_2}|^2 \ldots |(x_q)_{i_n}|^2.
\end{equation}
By defining $v_q := x_q \odot \overbar{x_q}$, it is clear that $\forall {q}, \, v_q \geq 0$.
With this notation, the previous equation reads 
$$Q[X] = \sum^K_{q=1} v_q^{\otimes n}.$$ This implies that $Q[X] \in \mathsf{CP}^{(n)}_{d}$.

\smallskip

\noindent($\impliedby$) To show the reverse implication, we assume that $Q[X] \in \mathsf{CP}^{(n)}_{d}.$ This implies that there exists a finite number of entrywise positive vectors $\{v_q \in \mathbb{R}_+^d\}^K_{q=1}$ such that $$Q[X] = \sum^K_{q=1} v_q^{\otimes n}.$$ Consider the vectors, $\{x_q\}^K_{q=1}$ defined as $\forall q, \, \, (x_q)_i := \sqrt{(v_q)_i}$. By this definition, it immediately follows that $\forall q \in [K], \,  x_q \odot x_q = v_q$. Consider the following separable quantum state $\sigma \in \mathsf{Sep}^{(n)}_d, $ defined as $$\sigma := \sum^K_{q=1} \ketbra{x_q}{x_q}^{\otimes n}.$$ Note that, $\Pi \sigma \Pi = \sigma.$ i.e $\sigma$ is a bosonic state. We have the twirled quantum state, 
\begin{equation}
\mathcal{T}_{\diagU}(\sigma) =  \int_{\diagU_d} U^{\otimes n} \sum^K_{q=1}\ket{x_q^{\otimes n}} \bra{x_q^{\otimes n}} (U^{\otimes n})^* \, \mathrm{d}U = \sum^K_{q=1} \int_{\diagU_d} U^{\otimes n} \ket{x_q^{\otimes n}} \bra{x_q^{\otimes n}} (U^{\otimes n})^* \, \mathrm{d}U
\end{equation}
Firstly note that, $\Pi^{(n)}_d \mathcal{T}_{\diagU}(\sigma) \Pi^{(n)}_d = \mathcal{T}_{\diagU}(\sigma)$ and due to twirling $\forall U \in \diagU_d, \, [\mathcal{T}_{\diagU}(\sigma), U^{\otimes n}] = 0$, therefore by \cref{prop:symmetry-mods}, the state $\mathcal{T}_{\diagU}(\sigma) \in \mathsf{DS}^{(n)}_{d}.$ Secondly, since twirling is an LOCC operation, $\mathcal{T}_{\diagU}(\sigma) \in \mathsf{Sep}^{(n)}_{d}$. Finally note that
\begin{align*}
    Q[\mathcal{T}_{\diagU}(\sigma)]_{i} &= \braket{i \vert \mathcal{T}_{\diagU}(\sigma) \vert i} = \bra{i} \int_{\diagU_d} U^{\otimes n} \sum^K_{q=1}\ket{x_q^{\otimes n}} \bra{x_q^{\otimes n}} (U^{\otimes n})^*\, \mathrm{d}U \ket{i} \\
    &= \sum^K_{q=1} \int_{\diagU_d} \braket{i | x_q^{\otimes n}} \braket{x_q^{\otimes n}| i}\, \mathrm{d}U = Q[\sigma]_{i},
\end{align*}
showing that $Q[\mathcal{T}_{\diagU}(\sigma)] = Q[X]$; above we have used that $U^{\otimes n}\ket i$ is equal to $\ket i $ up to a phase for all diagonal unitary matrix $U \in \diagU_d$. Hence, by \cref{thm:tensor-param}, we conclude $\mathcal{T}_{\diagU}(\sigma) = X$, finishing the proof.
\end{proof}
The proof strategy we employed above is very similar to that presented in \cite{yu2016separability, quesada2017entanglement} for the bipartite case. \cref{thm:sep-and-cp} characterizes the separability properties of mixtures of Dicke states. Recall that checking for complete positivity is known to be NP-hard. 

We now use the duality relation from \cref{eq:HS-Q-W} to characterize entanglement witnesses in the diagonally symmetric subspace in terms of the copositive tensors.

\begin{proposition}
\label{prop:ew-cop}
The following equivalence holds for a matrix $O \in \operatorname{Herm}[\mathsf{DS}^{(n)}_d].$
$$O \in \mathsf{EW}^{(n)}_d \iff W[O] \in \mathsf{Cop}^{(n)}_d.$$
\end{proposition}

\begin{proof}
This follows from the duality of separable states and entanglement witnesses, read at the level of $Q$ and $W$ parametrizations.
More explicitly, by the result in \cref{prop:X-expansion-W}, we have
    $$\forall \ket{v} \in \mathbb{C}^d \quad \braket{v^{\otimes n} |O| v^{\otimes n}} =  \braket{v^{\otimes n} \vert \Pi^{(n)}_{d} \operatorname{diag}[W[O]] \Pi^{(n)}_{d} \vert v^{\otimes n}} =  \braket{W[O], (\bar v \odot v)^{\otimes n}} \geq 0.$$
\end{proof}

\begin{remark}
The connection between entanglement witnesses and copositive \emph{matrices} has already been noted in \cite{marconi2021entangled} for the bipartite case. More precisely, it is shown in \cite{marconi2021entangled} that, at $n=2$, for the witness parametrized by a matrix $C \in \operatorname{Herm}[\mathcal{M}_d]$ in the way $X[C] = \sum_{i,j} C_{ij} \ketbra{ij}{ji}$, $$X[C] \in \mathsf{EW}^{(2)}_{d} \iff C \in \mathsf{Cop}_d^{(2)}.$$ 

This is slightly different from our parametrization, as $X[C]$ is not an element of the DS subspace, in particular, it is not bosonic as $X[C] P_{(12)}  = \sum_{i,j} C_{ij} \ketbra{ij}{ji} P_{(12)} = \sum_{ij} C_{ij} \ketbra{ij}{ij} \neq X[C].$ Therefore, instead of $X[C]$, in \cref{prop:ew-cop}, we consider $\sum_{ij} C_{ij} \ketbra{ij}{ij}$ and project it into the symmetric subspace.

\end{remark}

Before discussing an example of a bosonic entanglement witness, let us recall the Arithmetic Mean-Geometric Mean (AM-GM) weighted inequality. Given, some $d$ non-negative real numbers $x_1, x_2 \ldots x_d \in \mathbb{R}_+$, and weights $\lambda_1, \lambda_2 \ldots \lambda_d \in [0,1]$ such that  $\sum^d_{i=1}\lambda_i =1$, we have the following inequality

$$\sum^d_{i=1} \lambda_i x_i \geq \prod^d_{i=1} x_i^{\lambda_i}.$$ For all $\lambda_i = 1/d$, this reduces to the usual AM-GM inequality. 
$$\frac{\sum^d_{i=1} x_i}{d} \geq \prod^d_{i=1} x_i^{1/d}.$$

\begin{example}[Projective entanglement witnesses] We show some examples of bosonic diagonally symmetric entanglement witnesses. We define 
$$R_\alpha(\mu) := \Pi^{(n)}_d - \mu \ketbra{D_{\alpha}}.$$ where $\Pi^{(n)}_d$ is a projector into the symmetric subspace $\vee^n \mathbb{C}^d$. Then, by computing the tensor parametrization of the witness, we obtain, $$\forall i \in [d]^n, \, W[R_{\alpha}]_i =  \binom{n}{\gamma(i)} \braket{i \vert \Pi^{(n)}_d - \mu \ketbra{D_{\alpha}} \vert i}  = 1 - \mu \frac{ \binom{n}{\gamma(i)}}{\binom{n}{\alpha}}\mathbb{1}_{\gamma(i) = \alpha} = 1 - \mu\mathbb{1}_{\gamma(i) = \alpha}$$
Consider the polynomial, $$\braket{W[R_{\alpha}], x^{\otimes n}} = \Big(\sum_{i=1}^d x_i\Big)^{n} - \mu\binom{n}{\alpha} x^{\alpha}.$$ By applying the AM-GM weighted inequality, we obtain for all $x \in \mathbb{R}_+^d$, $$\frac{1}{n} \Big(\sum_{i=1}^d x_i\Big) =  \frac{1}{n}(\sum^d_{i=1}\alpha_i \frac{x_i}{\alpha_i}) \geq \prod^d_{i=1} (\frac{x_i}{\alpha_i})^{\frac{\alpha_i}{n}} \iff \Big(\sum_{i=1}^d x_i\Big)^n \geq \frac{n^n}{\prod^d_{i=1} \alpha_i^{\alpha_i}} x^{\alpha}$$ Using \cref{prop:ew-cop}, we have that for all 
$$-\infty \leq  \mu \leq \frac{1}{\binom{n}{\alpha}\prod^d_{i=1} (\frac{\alpha_i}{n})^{\alpha_i}},$$
$R_\alpha(\mu)$ is a bosonic entanglement witness. In the case of $n$ qubits, the Dicke state with occupation numbers $(n/2, n/2)$ yields the entanglement witness 
$$\Pi^{(n)}_2 - \mu_* \ketbra{D_{(n/2,n/2)}}$$
with $\mu_* = \frac{2^n}{\binom{n}{{n/2}}} \sim \sqrt{\frac{\pi n}{2}} $ as $n \to \infty$. Similar computations with entanglement witnesses constructed from Dicke states were considered in \cite{toth2007detection,bergmann2013entanglement}.
\end{example}

\subsection{Decomposability and PPT}
In the preceding section, we demonstrated that the separability of the states in the diagonally symmetric subspace is equivalent to the complete positivity of the tensor $Q[X].$ The convex cones $\mathsf{PPT}^{(n,k)}_d$ are outer tractable approximation of the separable states. The pertinent question is how to describe the PPT property in terms of the tensor $Q[X].$ The analytical computation of PPT conditions in the diagonally symmetric subspace was obtained in \cite[Lemma III.4 and Lemma III.5]{romero2025multipartite} by block diagonalization of the matrix $X^\Gamma$. These conditions appear highly non-trivial to describe. This motivates us to focus on the dual objects, the sets of decomposable witnesses, $\mathsf{Dec}^{(n,k)}_d.$ Surprisingly, these sets can be shown to have a correspondence with the sum of squares property of tensors.

We now state and prove the main result of this section. Briefly, this result provides a correspondence between decomposable witnesses $O$ and the properties of the polynomial $p_{W[O]}$ in the diagonally symmetric subspace. For  the rather long and technical proof of this theorem, see \cref{sec:appendix-proof-Dec}.

\begin{theorem}\label{thm:Dec-n-k-SOS}
The following results hold for $O \in \operatorname{Herm}[\mathsf{DS}^{(n)}_d].$ 

\begin{enumerate}
    \item $O \in  \mathsf{Dec}^{(n,l)}_d \implies p_{W[O]}(x) \in \operatorname{cone}\{x^\alpha \psi \mid |\alpha| \in \{n, n-2, n-4 \ldots n-2l\}, \psi \in \Sigma^{(n-|\alpha|)}_d\}.$
    \item Let $k \leq \lfloor n/2 \rfloor$ such that $\alpha \in \mathbb{N}^{(n-2k)}_d$ and $\psi \in \homo{d}{k}.$ Then, $$p_{W[O]}(x) = x^{\alpha} \psi^2 \implies O \in  \mathsf{Dec}^{(n,l)}_d$$ for all $l \in [k,n-k].$
\end{enumerate}    
    In particular, we have the equivalence, $$O \in  \mathsf{Dec}^{(n,l)}_d \iff p_{W[O]}(x) \in \operatorname{cone}\{x^\alpha \psi \mid |\alpha| \in \{n, n-2, n-4 \ldots n-2l\}, \psi \in \Sigma^{(n-|\alpha|)}_d\}.$$
    and hence the following inclusions between the decomposable sets. 
    $$\forall l \in [\lfloor n/2 \rfloor-1], \quad \mathsf{Dec}^{(n,l)}_d \cap \operatorname{Herm}[\mathsf{DS}^{(n)}_d]  \subseteq \mathsf{Dec}^{(n,l+1)}_d \cap \operatorname{Herm}[\mathsf{DS}^{(n)}_d].$$
\end{theorem}

To make the notation less clunkier, we define the increasing series of sets, $$\mathsf{SOS}^{(n,l)}_d:= \{T \mid p_{T}(x) \in \operatorname{cone}\{x^\alpha \psi \mid |\alpha| \in \{n, n-2, n-4 \ldots n-2l\}, \psi \in \Sigma^{(n-|\alpha|)}_d\}\}.$$ Recall that an $n$-th order symmetric tensor is called a sum-of-squares (represented as $\mathsf{SOS}^{(n)}_d = \mathsf{SOS}^{(n,\lfloor n/2 \rfloor)}_d$ by \cref{prop:SOS}) tensor if its associated polynomial is a sum of squares polynomial in "$x^2$", 
$$p_{T}(x \odot x) \in \Sigma^{(2n)}_d.$$

\begin{corollary} 
\label{thm:dec-iff-sos}
For any $O \in \operatorname{Herm}[\mathsf{DS}^{(n)}_d]$, we have
    $$W[O] \in \mathsf{SOS}^{(n)}_d \iff O \in \mathsf{Dec}^{(n, \lfloor n/2 \rfloor)}_d = \sum^{\lfloor n/2 \rfloor}_{l=0}  \mathsf{Dec}^{(n,l)}_d.$$
\end{corollary}

Let us summarise what we have shown in the previous result. For a diagonally symmetric (DS) operator $O$, the $\lfloor n/2 \rfloor$- decomposability property is equivalent to the symmetric tensor $W[O]$ being SOS. Surprisingly, we also find that $\lfloor n/2 \rfloor$- decomposability implies decomposability of the DS operator for all $l \leq \lfloor n/2 \rfloor.$ Recall that for $X,Y \in \mathsf{DS}^{(n)}_d$, we have, $$\braket{X,Y} = \braket{Q[X],W[Y]}.$$ We use this fact to characterize the property of PPT in the diagonally symmetric subspace, culminating with \cref{thm:ppt-mom}, in which we show that the PPT condition with respect to the most balanced bipartion implies the PPT conditions with respect to all other bipartitions. To accomplish this, we characterise the dual of SOS tensors. 

\bigskip
\noindent \textbf{Moment Tensors}.  
Let $f : \homo{d}{n} \rightarrow \mathbb{R}$. For all $\alpha \in \mathbb{N}^d$ such that $|\alpha| = n-2j$, we define the following matrices, with the columns and rows labelled by $\beta, \beta'$ such that $|\beta| = |\beta'| = j$. 

$$(\mathcal{M}^{(\alpha,j)}[f])_{\beta, \beta'} := f(x^{\alpha + \beta + \beta'}).$$ This defines a finite family of \emph{moment matrices} labeled by the monomial vector $\alpha$ and $j$. 

Define the linear functional over the homogeneous polynomials.
\[
\mathcal{L}_T : \homo{d}{n} \to \mathbb{R},\qquad \mathcal{L}_T(x^\mu):= T_{\idx(\mu)}
\quad (|\mu|=n),
\]
and extend it by linearity. Then \(\mathcal{L}_T(p_Q)=\langle T,Q\rangle\) for every symmetric tensor \(Q\), and the moment-matrix entries are
\[
\big(\mathcal{M}^{(\alpha,j)}[\mathcal{L}_T]\big)_{\beta,\beta'}
:= \mathcal{L}_T\big(x^{\alpha+\beta+\beta'}\big),
\qquad |\alpha|=n-2j,\; |\beta|=|\beta'|=j.
\]

We can show the following result. 
\begin{theorem}
\label{thm:dual-SOS}
The following statements are equivalent and characterise the SOS tensors.
\begin{enumerate}
    \item[(1)] $T \in (\mathsf{SOS}^{(n,k)}_d)^\circ$
    \item[(2)] $\forall j = 0, \ldots, k , |\alpha| = n-2j,  \quad \mathcal{M}^{(\alpha,j)}[\mathcal{L}_T] \succeq 0.$
\end{enumerate}
\end{theorem}

\begin{proof}
    We denote by $\Delta_d^{(n)} := \{p \in \homo{d}{n} \mid p(x \odot x) \in \Sigma^{(2n)}_d\}.$
    By \cref{prop:SOS}, we have the general characterization of polynomials $p(x)$ such that $p(x \odot x)$ is SOS. Note that such polynomials belong to the following convex cone. 
    $$p (x) \in\operatorname{cone}\{x^\alpha \phi^2 (x) | \,  |\alpha| + 2\deg (\phi) = n\}.$$
    Therefore $T \in (\mathsf{SOS}^{(n,k)}_d)^\circ \iff \mathcal{L}_T \in (\Delta_d^{(n)})^\circ \iff \forall |\alpha| + 2\deg (\phi) = n$ we have $\mathcal{L}_T(x^\alpha \phi^2 (x)) \geq 0$. Let $\phi(x) = \sum_{|\beta| = \frac{\deg (\phi)}{2}} c_{\beta}x^\beta$. Then, by plugging this into the functional,

    $$\mathcal{L}_T(\sum_{\beta \beta'} c_{\beta} c_{\beta'} x^{\alpha + \beta + \beta'}) = \sum_{\beta \beta'} c_{\beta} c_{\beta'} \mathcal{M}^{\alpha,j}[\mathcal{L}_T]_{\beta, \beta'} \geq 0.$$  Since this is true for all $c_\beta \in \mathbb{R}^{|\beta|}$, $\mathcal{M}^{\alpha,j}[\mathcal{L}_T] \succeq 0$.
\end{proof} 

In the remainder of this section, we want to interpret this family of moment matrices in the language of tensors rather than linear functionals. We do this by showing that the conditions in \cref{thm:dual-SOS} can be exactly obtained by the combination of the following two operations on tensors. 

\begin{itemize}
    \item \underline{Tensor flattenings}: Let $T \in \vee^{2n} \mathbb{R}^d$ be a symmetric tensor of even order $2n$ and dimension $d$. Then, such a tensor can be associated to a matrix of size $d^n \times d^n$ by dividing the indices equally to arrange it as a matrix, 

    $$M_{\alpha \beta} := T_{\alpha_{1} \alpha_{2} \ldots \alpha_n  \beta_{1} \beta_2 \ldots \beta_n}.$$

    This is called a tensor flattening (represented diagrammatically in \cref{fig:tensor-flattening}). Note that we consider only the flattenings that are \emph{balanced}, i.e, we consider an equal number of indices to create a matrix. 

    \begin{figure}[htb]
        \centering
        \includegraphics[width=0.8\linewidth]{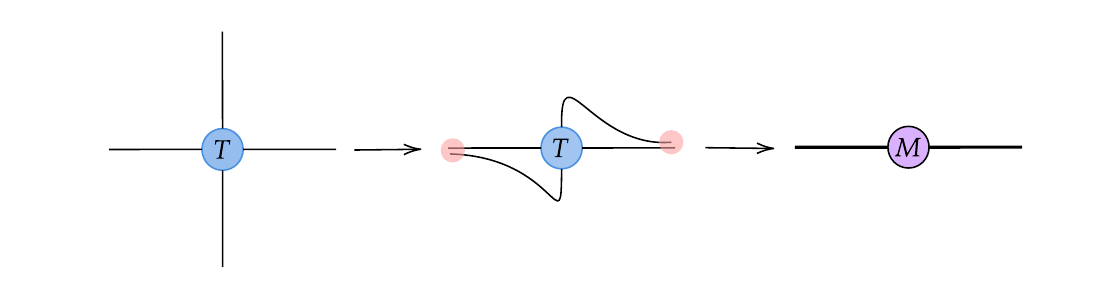}
        \caption{The figure represents a balanced flattening of a tensor $T \in \vee^4\mathbb{R}^d.$}
        \label{fig:tensor-flattening}
    \end{figure}

    \item \underline{Tensor slices}: Let $T \in \vee^{n} \mathbb{R}^d$ be a symmetric tensor of even order $n$ and dimension $d$. For a subset of indices $I \subset [n]$, we define the tensor $S \in \vee^{(n - |I|)} \mathbb{R}^d$ by fixing the indices of the tensor $T$ belonging to $I$, i.e $\forall k \in I$, $i_k = a_k \in \mathbb{N}$: 

    $$S_{j_1 \ldots j_{n-|I|}} := T_{a_1 \ldots a_{|I|} j_1 \ldots j_{n-|I|}}.$$
    
    \begin{figure}[htb]
        \centering
        \includegraphics[width=0.8\linewidth]{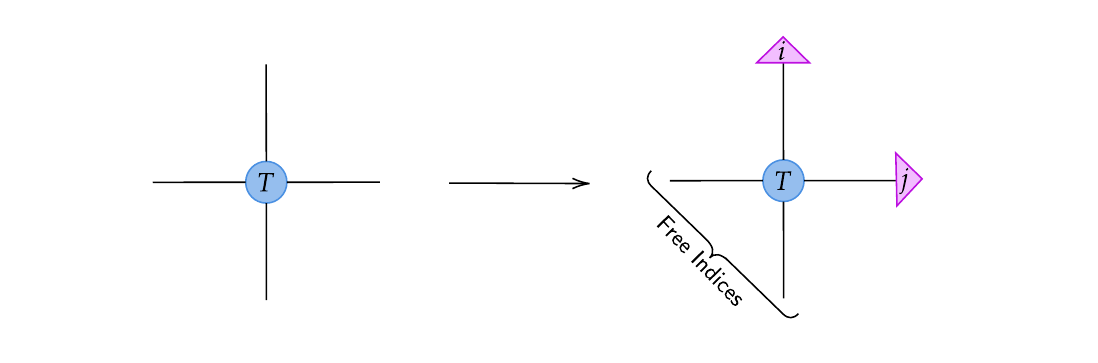}
        \caption{The figure represents a slice of a tensor $T \in \vee^4\mathbb{R}^d.$ by fixing two indices to be $i,j$.}
        \label{fig:tensor-slices}
    \end{figure}

    This is represented diagrammatically in \cref{fig:tensor-slices}.
\end{itemize}

In what follows, we shall consider all the even-order slices of a given tensor. For tensors of even order $2n$, we have even slices of order $2n, 2n-2, 2n-4, \ldots, 2, 0$, and for tensors of odd order $2n-1$, we have even slices of order $2n-2, 2n-4, \ldots,2, 0$. We consider the family of \emph{flattenings} of all the \emph{even} slices of the tensor, and we call this family of square matrices, the \emph{slice flattenings} associated with a tensor $T$, denoted $\operatorname{SF}[T]$. We also use $\operatorname{SF}^{(k)}[T]$ to denote the flattenings of even slices of size $\leq 2k$ (even order tensors) and  $\leq 2k+1$ (odd order tensors).

\begin{proposition}
\label{prop:cp-slice}
    Let $Q$ be a completely positive tensor. 
    \begin{itemize}
        \item[(1)]  Any slice of the tensor is completely positive tensor. 
        \item[(2)]  Any balanced flattening of the tensor is a completely positive matrix. 
        \item[(3)] All the slice flattenings are completely positive matrices. 
    \end{itemize}
   
\end{proposition}
\begin{proof}
Let $Q$ have the completely positive decomposition, $Q = \sum^N_{q=1} v^{\otimes n}_q.$ 
\begin{itemize}
    \item[(1)]  Let $S$ be the slice of the tensor over the indices $I$, by fixing indices $\forall k \in I,  i_k = a_k.$ Then, we have 
    $$S = \sum^N_{q=1} \prod_{k \in I} (v_q)_{a_k}v^{\otimes n-|I|}_q \in \mathsf{CP}^{(n-|I|)}_d$$ as $\prod_{k \in I} (v_q)_{a_k} \geq 0.$
    \item[(2)] For $n$ odd there is no balanced flattening. Hence, we assume $n = 2m$ is even. Then, we have that the balanced flattening of the tensor $$F = \sum^N_{q=1} \ketbra{v_q^{\otimes m}}{v_q^{\otimes m}}.$$ 
    \item[(3)] The proof follows from combining the result in Item (1) and Item (2).
\end{itemize}
\end{proof} 

We define the following convex cones of moment tensors:
$$\mathsf{Mom}^{(n,k)}_d := \{T \in \vee^n \mathbb{R}^d \mid \forall X \in \operatorname{SF}^{(k)}[T], \, \, X \succeq 0\},$$
and for $k =\lfloor n/2 \rfloor$ we denote the set of \emph{all slice flattenings} positive semidefinite by

\begin{equation}
\label{eq:mom-definition}
    \mathsf{Mom}^{(n)}_d := \mathsf{Mom}^{(n,\lfloor n/2 \rfloor)}_d.
\end{equation}

By \cref{prop:cp-slice}, this is outer relaxation of completely positive tensors. Moreover, we have the following equivalence that establishes that moment tensors are dual cones to SOS tensors.
\begin{theorem}
\label{thm:duality-moment-sos}
We have the following conic-duality:
    $$\forall k \in [\lfloor n/2 \rfloor], \quad (\mathsf{Mom}^{(n,k)}_d)^\circ  = \mathsf{SOS}^{(n,k)}_d.$$
\end{theorem} 

\begin{proof}
Recall that we have the following equivalence 
\[
    T \in (\mathsf{SOS}^{(n,k)}_d)^\circ 
    \quad\Longleftrightarrow\quad
    \forall\, j = 0,\ldots, k \; |\alpha| = n-2j,\;
    \mathcal{M}^{(\alpha,j)}[\mathcal{L}_T] \succeq 0 .
\]
We now explain how the moment matrices \(\mathcal{M}^{(\alpha,j)}[\mathcal{L}_T]\) are obtained from slice--flattenings of \(T\). Recall the entries of this matrix
\[
\big(\mathcal{M}^{(\alpha,j)}[\mathcal{L}_T]\big)_{\beta,\beta'}
= \mathcal{L}_T\big(x^{\alpha+\beta+\beta'}\big),
\qquad |\alpha|=n-2j,\; |\beta|=|\beta'|=j.
\]

Fix \(\alpha\in\mathbb{N}_d^{(n-2j)}\). Consider the even slice \(S^{(\alpha)}\in\vee^{2j}\mathbb{R}^d\) of \(T\) obtained by fixing the first \(n-2j\) indices to the canonical multi-index \(\idx(\alpha)\):
\[
S^{(\alpha)}_{u\cup v} := T_{\idx(\alpha)\cup u\cup v},
\qquad u,v\in[d]^j.
\]
Its balanced flattening is the \(d^j\times d^j\) matrix
\[
F^{(\alpha)}_{u,v}:=S^{(\alpha)}_{u\cup v}=T_{\idx(\alpha)\cup u\cup v}.
\]
Because \(T\) is symmetric, the entries \(F^{(\alpha)}_{u,v}\) only depend on the occupation vectors \(\beta:=\gamma(u)\) and \(\beta':=\gamma(v)\). Consequently, \(F^{(\alpha)}\) has repeated rows/columns whenever two multi-indices have the same occupation pattern.

Let \(D\) be the \(\{0,1\}\)-matrix that duplicates rows/columns according to occupation vectors (one column per \(\beta\in\mathbb{N}_d^{(j)}\), one row per \(u\in[d]^j\), with \(D_{u,\beta}=1\) iff \(\gamma(u)=\beta\)). Then one checks that
\[
F^{(\alpha)} = D\,\mathcal{M}^{(\alpha,j)}[\mathcal{L}_T] \,D^{\top},
\]
so \(\mathcal{M}^{(\alpha,j)}[\mathcal{L}_T]\) is exactly the \emph{effective} flattening obtained from \(F^{(\alpha)}\) by deleting repeated rows and columns.
In particular, positive semidefiniteness is preserved (and reflected) by this deletion/duplication: since \(D\) has full column rank, we have
\(
F^{(\alpha)}\succeq 0 \iff \mathcal{M}^{(\alpha,j)}[\mathcal{L}_T]\succeq 0.
\)

\end{proof}

\begin{remark}
    The size of the tensor flattenings of the general tensor of $2n$ legs is $d^n \times d^n$. But due to symmetry in the tensor, this leads to repeated multiple rows and columns in the matrix, leading to the size of the largest \emph{effective} flattening (after deletion of the repeated rows and columns) being $\binom{n+d-1}{d-1} \times \binom{n+d-1}{d-1}$. For $d = 2$, we get the largest effective flattening of size $(n+1) \times (n+1)$.
\end{remark}

We end this section by showing the exact PPT conditions for the quantum states in the diagonally symmetric subspace. We present these conditions in terms of the slice flattening of the tensor $Q[X].$ Note that the same conditions for PPT have been obtained in \cite{romero2025multipartite}, in which they also pose the conjecture whether in the diagonally symmetric subspace, the PPT just across the \emph{largest} bipartition of the $n$-partite state implies PPT across all bipartitions (and also positivity of the state). The next result discusses the PPT conditions and the positive resolution of this conjecture. 

\begin{theorem}
\label{thm:ppt-mom}
    Let $X \in \operatorname{Herm}[\mathsf{DS}^{(n)}_d]$. Then, we have the following equivalence: 
    $$X \in \mathsf{PPT}_d^{(n,k)} \iff Q[X] \in \mathsf{Mom}^{(n,k)}_d.$$
    In particular, we have for the PPT sets:
    $$\forall l \in [\lfloor n/2 \rfloor-1], \quad \mathsf{PPT}^{(n,l+1)}_d \cap \operatorname{Herm}[\mathsf{DS}^{(n)}_d]  \subseteq \mathsf{PPT}^{(n,l)}_d \cap \operatorname{Herm}[\mathsf{DS}^{(n)}_d].$$
\end{theorem}
\begin{proof}
    Firstly, note the following equality between the sets.
    $$\mathsf{Dec}^{(n,k)}_d \cap \operatorname{Herm}[\mathsf{DS}^{(n)}_d] = \{\mathcal{T}_{\diagU}(O) \mid O \in \mathsf{Dec}^{(n,k)}_d\}.$$
    $$\mathsf{PPT}^{(n,k)}_d \cap \operatorname{Herm}[\mathsf{DS}^{(n)}_d] = \{\mathcal{T}_{\diagU}(X) \mid X \in \mathsf{PPT}^{(n,k)}_d\}.$$ Moreover, we have for all $W \in \operatorname{Herm}[\vee^n \mathbb{C}^d]$ and $X \in \operatorname{Herm}[\mathsf{DS}^{(n)}_d]$, the following trace relation.
    $$\operatorname{Tr}(XO) = \operatorname{Tr}(\mathcal{T}_{\diagU}(X) O) = \operatorname{Tr}(X \mathcal{T}_{\diagU}(O)).$$ Then, we have the following equivalences for $X \in \operatorname{Herm}[\mathsf{DS}^{(n)}_d]$:
    \[
\begin{aligned}
X \in \mathsf{PPT}^{(n,k)}_d
&\iff 
\forall\, O \in \mathsf{Dec}^{(n,k)}_d,\;
\operatorname{Tr}(X O) \ge 0 \\[0.4em]
&\iff 
\forall\, O \in \mathsf{Dec}^{(n,k)}_d,\;
\operatorname{Tr}\!\bigl(X\,\mathcal{T}_{\diagU}(O)\bigr) \ge 0 \\[0.6em]
&\iff 
\forall\, O \in \mathsf{Dec}^{(n,k)}_d 
\cap \operatorname{Herm}\!\bigl[\mathsf{DS}^{(n)}_d\bigr],\;
\operatorname{Tr}(X O) \ge 0 \\[0.6em]
&\iff 
\forall\, O \in \mathsf{Dec}^{(n,k)}_d 
\cap \operatorname{Herm}\!\bigl[\mathsf{DS}^{(n)}_d\bigr],\;
\braket{Q[X], W[O]} \ge 0 \\[0.6em]
&\iff 
\forall\, W[O] \in \mathsf{SOS}^{(n,k)}_d,\;
\braket{Q[X], W[O]} \ge 0 \\[0.6em]
&\iff 
Q[X] \in \mathsf{Mom}^{(n,k)}_d.
\end{aligned}
\]
\end{proof}

The following direct corollary of the theorem above settles in the affirmative the Conjecture III.1 in \cite{romero2025multipartite}.

\begin{corollary}\label{cor:PPT-balanced-strongest}
    Let $X \in \operatorname{Herm}[\mathsf{DS}^{(n)}_d]$ be a mixture of Dicke states on $n$ particles with local dimension $d$. If $X$ is PPT with respect to the most balanced bipartition $n = \lfloor n/2 \rfloor + \lceil n/2 \rceil$, then $X$ is PPT with respect to all other bipartitions:
    $$\forall l \in [\lfloor n/2\rfloor], \quad \mathsf{PPT}^{(n,\lfloor n/2 \rfloor)}_d \cap \operatorname{Herm}[\mathsf{DS}^{(n)}_d]  \subseteq \mathsf{PPT}^{(n,l)}_d \cap \operatorname{Herm}[\mathsf{DS}^{(n)}_d].$$
    In other words, in the case of mixture of Dicke states, the PPT criterion with respect to the most balanced bipartion is stronger than the PPT criterion with respect to all other bipartitions. 
\end{corollary}

The previous theorem implies that if we desire PPT of a partition of size $k$, this means that flattenings of all even slices of size $\geq n-2k$ are positive-semidefinite. For $k = 0$, the PPT condition reduces to the positivity of the matrix $X \in \mathsf{DS}^{(n)}_d$, which is equivalent to the entrywise positivity of the tensor $Q[X].$ We present the PPT conditions for the case of $n = 2$ and $n=3$ in the following tables. 

\definecolor{headergray}{RGB}{240,240,240}
\definecolor{rowblue}{RGB}{245,248,252}

\begin{table}[H]
\centering

\label{tab:ppt-n2}
\renewcommand{\arraystretch}{1.5}
\begin{tabular}{|c|c|}
\hline
\rowcolor{headergray}
\textbf{Property} & \textbf{Condition on $Q[\cdot]$} \\ \hhline{|=|=|}
\rowcolor{rowblue}
\(X \in \mathsf{PPT}^{(2,0)}_d\)
& \(\forall i,j \in [d],\; Q[X]_{ij} \geq 0\) \\ \hline
\(X \in \mathsf{PPT}^{(2,1)}_d \)
& \(\forall i,j \in [d],\; Q[X]_{ij} \geq 0 \text{ and }
\{Q[X]_{ij}\}_{i,j=1}^d \succeq 0\) \\ \hhline{|=|=|}
\rowcolor{rowblue}
\(X \in \mathsf{Sep}^{(2)}_d\)
& \(Q[X] \in \mathsf{CP}^{(2)}_d\) \\ \hline
\end{tabular}
\caption{PPT and separability conditions for \(n=2\).}
\end{table}

\begin{table}[H]
\centering

\label{tab:ppt-n3}
\renewcommand{\arraystretch}{1.5}
\begin{tabular}{|c|c|}
\hline
\rowcolor{headergray}
\textbf{Property} & \textbf{Condition on $Q[\cdot]$} \\ \hhline{|=|=|}
\rowcolor{rowblue}
\(X \in \mathsf{PPT}^{(3,0)}\)
& \(\forall i,j,k \in [d],\; Q[X]_{ijk} \geq 0\) \\ \hline
\(X \in \mathsf{PPT}_d^{(3,1)}\)
& \(\forall i,j,k \in [d],\; Q[X]_{ijk} \geq 0 \text{ and }
\forall k,\; \{Q[X]_{ijk}\}_{i,j=1}^d \succeq 0\) \\ \hhline{|=|=|}
\rowcolor{rowblue}
\(X \in \mathsf{Sep}^{(3)}_d\)
& \(Q[X] \in \mathsf{CP}^{(3)}_d\) \\ \hline
\end{tabular}
\caption{PPT and separability conditions for \(n=3\).}
\end{table}

We summarize the main connections between the entanglement properties in the diagonally symmetric subspace and the corresponding tensor properties in \cref{fig:entanglement-DS-tensors}. 

\begin{figure}[htb]
    \centering

    \usetikzlibrary{matrix, arrows.meta}
    \begin{center}
    \begin{tikzpicture}[>=Stealth]
      \matrix (m) [matrix of math nodes, row sep=3em, column sep=4em, text height=1.5ex, text depth=0.25ex]
      {
        \mathsf{CP}_{d}^{(n)}  & \mathsf{Mom}_{d}^{(n)} & \mathsf{NN}_{d}^{(n)} & \mathsf{SOS}_{d}^{(n)} & \mathsf{Cop}_{d}^{(n)} \\
        \mathsf{Sep}_{d}^{(n)} & \mathsf{PPT}_{d}^{(n)} & \mathsf{PSD}_{d}^{(n)} & \mathsf{Dec}_{d}^{(n)} & \mathsf{EW}_{d}^{(n)}  \\
      };
    
      \foreach \i/\j in {1/2, 2/3, 3/4, 4/5} {
        \path (m-1-\i) -- node {$\subseteq$} (m-1-\j);
        \path (m-2-\i) -- node {$\subseteq$} (m-2-\j);
      }
    
      \draw[<->] (m-1-1) -- node[left]  {\scriptsize \cref{thm:sep-and-cp} [Q]}  (m-2-1);
      \draw[<->] (m-1-2) -- node[left] {\scriptsize \cref{thm:ppt-mom} [Q]}  (m-2-2);
      \draw[<->] (m-1-3) -- node[left] {\scriptsize \cref{thm:tensor-param} [Q/W]}(m-2-3);
      \draw[<->] (m-1-4) -- node[left] {\scriptsize \cref{thm:dec-iff-sos} [W]} (m-2-4);
      \draw[<->] (m-1-5) -- node[left] {\scriptsize \cref{prop:ew-cop} [W]} (m-2-5);
    
    \end{tikzpicture}
    \end{center}
    
    \caption{The different notions of positivity for Hermitian matrices in the DS subspace (bottom row) and the corresponding class of real symmetric tensors (top row). References for the proof of the correspondence are given, along with the parametrization ($Q$ and/or $W$) realizing the correspondance.}
    \label{fig:entanglement-DS-tensors}
\end{figure}

\subsection{Existence of PPT entangled states}
\label{sec:ppt-ent-ds-subspace}
In the last section, we established a complete dictionary of entanglement properties in the diagonally symmetric subspace. We  now proceed to do a systematic study of PPT entanglement by using the correspondences. Recall that a mixture of Dicke states, $X \in \mathsf{DS}^{(n)}_d$, is PPT for bipartitions if and only if $Q[X] \in \mathsf{Mom}^{(n)}_d.$ On the dual side, we have that a witness $O \in \mathsf{DS}^{(n)}_d$ is decomposable if and only if $W[O] \in \mathsf{SOS}^{(n)}_d.$ The PPT entangled states correspond exactly to the moment tensors, $\mathsf{Mom}^{(n)}_d$ that are not in $\mathsf{CP}^{(n)}_d,$ and on the dual side, the witnesses correspond to $\mathsf{Cop}^{(n)}_d$ which are not $\mathsf{SOS}_d^{(n)}.$ 

In the remainder of this section, we analyze the three different regimes: multipartite qubit systems, bipartite systems, and the rest. The first two have already been scrutinized in \cite{yu2016separability, tura2018separability, quesada2017entanglement}. 

\bigskip

\noindent\textbf{Qubit systems}. We begin by recalling the following standard fact \cite{hilbert1888darstellung} about homogeneous polynomials of $2$ variables (see also \cref{tab:homo-sos}):  non-negative homogeneous polynomials in two variables of even degree are sums-of-squares; equivalently, every non-negative polynomial in one real variable is a sum of (two) squares. This directly implies the following result in \cite[Theorem 2]{yu2016separability}.

\begin{theorem}[\cite{yu2016separability}]
\label{thm:no-ppt-in-qubits}
     Let $X \in \mathsf{DS}^{(n)}_2.$ Then, $X$ is separable if and only if it is PPT across all bipartitions. 
\end{theorem}

\begin{proof}
    By Hilbert's result \cite{hilbert1888darstellung} it follows that $\mathsf{Cop}^{(n)}_2 = \mathsf{SOS}^{(n)}_2.$ Therefore, on the dual side it follows that $\mathsf{CP}^{(n)}_2 = \mathsf{Mom}^{(n)}_2$ showing the result.
\end{proof}

We also explicitly show the conditions for PPT and  hence, the necessary and sufficient conditions for separability in the \cref{sec:appendix-proof-qubits}. Moreover, note that the PPT across all bipartitions is implied by just the PPT of the most balanced bipartition. 

\noindent\textbf{Bipartite qudit systems}. For the bipartite mixtures of Dicke states, this problem reduces to the matrix case. Recall the results about homogeneous polynomials of degree $4$ \cite{hilbert1888darstellung}, as presented in \cref{tab:homo-sos-states-bi}. The table implies that for $d = 3$, $\mathsf{SOS}^{(2)}_3 = \mathsf{Cop}_3^{(2)}$ and similarly, on the dual side, $\mathsf{CP}^{(2)}_3 = \mathsf{Mom}^{(2)}_3$. Surprisingly, this equivalence also holds for $d=4$, as stated in the next theorem.

\begin{table}[H] 
\centering 
\renewcommand{\arraystretch}{1.3}
\begin{tabular}{|c|c|}
\hline
\rowcolor{green!10}
$d$ & $2n = 4$ \\
\hline
\rowcolor{yellow!10}
2 & \cellcolor{green!20}Yes \\
\hline
\rowcolor{yellow!10}
3 & \cellcolor{green!20}Yes \\
\hline
\rowcolor{yellow!10}
$\ge 4$ & \cellcolor{red!20}No \\
\hline
\end{tabular}
\caption{Can any homogeneous polynomial of degree $2n$ in $d$ variables be represented as a sum of squares? \label{tab:homo-sos-states-bi}}
\end{table}

\begin{theorem}[\cite{shaked2021copositive,tura2018separability,yu2016separability}]\label{thm:cp-dnn-matrices}
     For $d \leq 4$, $\mathsf{CP}^{(2)}_d = \mathsf{Mom}^{(2)}_d.$ Therefore, for $d \leq 4$, a bipartite mixture of Dicke states is separable if and only if it is PPT.  
\end{theorem}

Recall that $\mathsf{Mom}^{(2)}_d = \mathsf{NN}_d \, \cap \, \mathsf{PSD}_d := \mathsf{DNN}_d.$ For the $d=5$, the following matrix is an element of doubly non-negative ($\mathsf{DNN}_d$) but not $\mathsf{CP}_d$ as shown in \cite{gulati2025entanglement}.

\[
A =
\begin{pmatrix}
2\cos{\tfrac{\pi}{5}} & 1 & 0 & 0 & 1 \\
1 & 2\cos{\tfrac{\pi}{5}} & 1 & 0 & 0 \\
0 & 1 & 2\cos{\tfrac{\pi}{5}} & 1 & 0 \\
0 & 0 & 1 & 2\cos{\tfrac{\pi}{5}} & 1 \\
1 & 0 & 0 & 1 & 2\cos{\tfrac{\pi}{5}}
\end{pmatrix}.
\]

\begin{remark}
    The polynomial $f(x,y,z,w) = x^2y^2 + y^2z^2 + z^2x^2 + w^4 - 4 xyzw$ is a positive polynomial which is not SOS, in line with \cref{tab:homo-sos-states-bi}. But the previous theorem implies that there doesn't exist any "square root" of this polynomial, i.e $f(x,y,z,w) \neq p (x^2, y^2, z^2, w^2)$ for any polynomial $p$.
\end{remark}

\bigskip

\noindent\textbf{Three qutrits.} It was conjectured in \cite{romero2025multipartite} that for local dimensions $d=3,4$, there is no PPT entanglement in the class of mixtures of Dicke states for \emph{any} number of systems, analogous to the bipartite case. We will show that this conjecture is false by constructing explicit examples of PPT entangled states in three qutrits. This also implies (\cref{thm:qutrits-and-beyond}) that there exist PPT entangled states in all the remaining cases, $d \geq 3$ and $n \geq 3$.

Let us first focus on the construction of witnesses for PPT entanglement, i.e, finding explicit copositive tensors that are not SOS. This is also equivalent to finding homogeneous polynomials $p$ such that $p(x \odot x)$ is globally non-negative but not a sum of squares. Surprisingly, a lot of these examples already exist in the literature. Let us present a few of them. 

\begin{tcolorbox}[
  colback=gray!4,
  colframe=gray!75!black,
  enhanced,
  boxrule=0.8pt,
  arc=2mm,
  left=6pt,right=6pt,top=6pt,bottom=6pt
]

\textsc{Motzkin Polynomial} \cite{Motzkin1967}

\[
p_{\mathrm{mot}}(x,y,z)
  = x^2y + y^2x + z^3 - 3xyz
\]

\vspace{4mm}

\textsc{Robinson Polynomial} \cite{robinson1973some}

\[
\begin{aligned}
p_{\mathrm{rob}}(x,y,z)
  &= x^6 + y^6 + z^6 + 3x^2y^2z^2 - \bigl(
      x^2y^4 + x^4y^2
    + x^2z^4 + x^4z^2
    + y^2z^4 + y^4z^2
    \bigr)
\end{aligned}
\]

\end{tcolorbox}

The exact conditions for positivity and sum of squares for symmetric polynomials (symmetric under the exchange of any variables) have been considered in \cite{choi1987even}, which is a rich source of examples of positive polynomials that are not SOS. The exact conditions for positivity of such polynomials, along with the SOS conditions, are outlined in \cite{choi1987even}. 
\begin{example}[{{\cite[Theorems 3.7 and 4.25]{choi1987even}}}]
\label{ex:positive-non-sos}
    Consider $M_r = \sum^d_{i=1} x_i^r$ and the following polynomial.
    $$p_{a,b,c}(x) = a M_3 + b M_1 M_2 + c M^3_1.$$
    This polynomial is a homogeneous polynomial of degree $3$ in $d$ variables. Define the polynomial $p^*_{a,b,c}(t) = a + bt + ct^2$. 
    Then, we have the following equivalences: 
    \begin{itemize}
        \item $\forall x \geq 0, \quad p_{a,b,c}(x) \geq 0 \iff \forall k \in [d], \quad p^*_{a,b,c}(k) \geq 0.$ 
        \item $p_{a,b,c}(x \odot x) \in \Sigma^{(6)}_d \iff \forall t \in \{1\} \cup [2,d], \quad p^*_{a,b,c}(t) \geq 0.$
    \end{itemize}

In particular, Robinson's polynomial $R$ is given by $a = 3, b= -5/2, c= 1/2$. The corresponding $R^*$ polynomial can be factorized as, 
$R^*(t) = \frac{1}{2} (6 - 5t+ t^2) = \frac{1}{2} (2-t)(3-t).$ This implies that:
$$R^*(t) < 0 \iff 2 < t < 3$$
that violates the SOS condition. Surprisingly, in these constructions, the positivity and SOS conditions are independent of the number of variables $d$. Therefore, for the same parameter values, this provides examples of tensors in $\mathsf{Cop}^{(3)}_d$ that are not in $\mathsf{SOS}^{(3)}_d$ for all $d$. 

\end{example}

Let us discuss the case of Motzkin's polynomial in detail, which, in fact, was the first explicit positive polynomial that wasn't a sum of squares. Firstly, it is easy to check that this polynomial is positive for all $x \geq 0$  by the AM-GM inequality. 

$$p_{\mathrm{mot}}(x,y,z) = x^2 y + y^2 x + z^3 - 3 xyz = 3 (\frac{x^2 y + y^2 x + z^3}{3}- xyz) \geq 0.$$

We now show that the Motzkin polynomial is not SOS.  This proof strategy is well-known in the literature. 

\begin{theorem}
    The Motzkin polynomial "entrywise squared" is not a sum of squares.
    $$p_{\mathrm{mot}} (x \odot x) \notin \Sigma^{(2n)}_d.$$
\end{theorem}

\begin{proof}
    We recall that the Newton polytope of a general polynomial expanded in the polynomial basis  is given by 
    $$\mathcal{N}(p) = \operatorname{conv}\{\alpha \in \mathbb{N}^d \mid p(x) = \sum_{\alpha \in \mathbb{N}^d} \lambda_\alpha x^\alpha, \, \lambda_\alpha  \neq 0\}$$
    Then, we have the following property, 
    $\mathcal{N}(pq) = \mathcal{N}(p)+\mathcal{N}(q)$ and hence $\mathcal{N}(p^2) = 2\mathcal{N}(p).$
    Moreover, we have $\mathcal{N}(p+q) = \operatorname{conv}\{\mathcal{N}(p) \cup \mathcal{N}(q)\}$. Now, for any polynomial that is a sum of squares, i.e, $p_{\mathrm{mot}}(x \odot x) = \sum_k f^2_k(x)$, we have, 

    $$\mathcal{N}(p_{\mathrm{mot}}(x \odot x)) = 2 \operatorname{conv}\{\bigcup_k \mathcal{N}(f_k )\}$$    

    Let us assume to the contrary that the Motzkin polynomial is a sum of squares. The Newton polytope of the Motzkin polynomial is given by
    $$\mathcal{N}(p_{\mathrm{mot}}(x \odot x)) = \operatorname{conv}\{(4,2,0), (2,4,0), (0,0,6), (2,2,2)\}$$ and thus
    $$1/2 \cdot \mathcal{N}(p_{\mathrm{mot}}(x \odot x)) = \operatorname{conv}\{(2,1,0), (1,2,0), (0,0,3), (1,1,1)\}.$$
    Note that $\mathcal{N}(f_k) \subseteq 1/2 \cdot \mathcal{N}(p_{\mathrm{mot}}(x \odot x))$. Therefore, we can assume that $f_k = (a_k x_1x^2_2 + b_k x^2_1x_2 + c_k x_1 x_2 x_3 + d_k x^3_3)$ and hence the coefficient of $x^2_1 x^2_2 x^2_3$ is $\sum_k c_k^2 \geq 0$, which is a contradiction with the value $-3$ in $p_{\mathrm{mot}}(x \odot x)$. Therefore, the Motzkin polynomial cannot be a sum of squares.
\end{proof} The proof of the previous theorem actually works for all polynomials of the form,

$$p(x_1,x_2,x_3) = x^4_1 x^2_2 + x^4_2 x^2_1 + x^6_3 - 
    \alpha x^2_1x^2_2x^2_3$$ which is a positive, not SOS polynomial for $\alpha \in (0, 3]$. For a historical perspective on the problem of positive polynomials and sum of squares, we recommend the book \cite[Chapter 2]{powers2021certificates}.

Therefore, since we showed that there exist copositive tensors in $\vee^3 \mathbb{R}^3$ that are not SOS tensors, by duality, there are PPT entangled diagonally symmetric states for $3$ qutrits. 
For any polynomial $p$, we can define the inequality on completely positive tensors by using the tensor-polynomial correspondence. The strategy is as follows:  Select a polynomial $f(x)$ such that $\forall x \geq 0, f(x) \geq 0.$ Then, we use the correspondence to construct a copositive tensor $T$ such that $p_T = f.$ Then, for all $Q \in \mathsf{CP}^{(n)}_d$ we have, 
$$\braket{Q,T} \geq 0.$$ Moreover, if $f(x \odot x)$ is not a sum of squares, the tensor $T \notin \mathsf{SOS}^{(n)}_d$ acts as an indecomposable witness. Let us now present an explicit example of a PPT entangled state in $3$ qutrits. 

\begin{example}
Let $Q \in \vee^3 \mathbb{R}^3$ have entries
\[
Q_{ijk} =
\begin{cases}
p, & \text{if } i=j=k, \\[4pt]
r, & \text{if } i,j,k \text{ are all distinct}, \\[4pt]
q, & \text{otherwise}.
\end{cases}
\]
For each $i$, define the slice matrix $S^{(i)} \in \mathbb{R}^{3\times 3}$ by
\[
S^{(i)}_{jk} := Q_{ijk}.
\]
All slices are proportional to
\[
S^{(i)} \sim
\begin{bmatrix}
p & q & q \\[6pt]
q & q & r \\[6pt]
q & r & q
\end{bmatrix}.
\] The PPT conditions (see \cref{tab:ppt-n3}) imply that for all $i \in \{1,2,3\}$ we must have
$S^{(i)} \succeq 0$. This matrix is PSD if and only if \[
p \ge q \ge r \ge 0,
\qquad
p(q+r) \ge 2 q^2.
\] Using Robinson's polynomial defined previously, we obtain the inequality
\[
\forall Q \in \mathsf{CP}^{(3)}_d \quad
Q_{111} + Q_{222} + Q_{333} + 3 Q_{123}
- Q_{112} - Q_{122} - Q_{113} - Q_{133}
- Q_{223} - Q_{233} \ge 0.
\]
Under the present parametrization, this becomes
\[
3p + 3r - 6q \;\ge\; 0.
\] Hence, we obtain the optimization problem.
\[
\begin{aligned}
\eta^{\star} =\;& \text{minimize} && 3p + 3r - 6q \\[4pt]
& \text{subject to} &&
p \ge q \ge r \ge 0, \\
&&& p(q+r) \ge 2 q^2, \\
&&& 3p + 18q + 6r = 1.
\end{aligned}
\]

The last constraint comes from the normalization of the state. A numerical computation yields an optimal value $\eta^{\star} \approx -0.02 < 0$. In particular, there exists a feasible triple $(p,q,r)$ with $3p + 3r - 6q < 0$, showing that the tensor is not CP even though all slices are PSD. This tensor corresponds to a PPT entangled $3$ qutrit mixture of Dicke states.
\end{example}

\begin{remark}
    Note that the PPT entangled state that we construct has no bipartite reduction that is entangled. This is because the marginals of diagonally symmetric bosonic states are in the same class (ref. \cref{thm:marginals-dicke-states}), and in $3 \vee 3$ systems, the states in the DS subspace have no PPT entanglement (ref. \cref{thm:cp-dnn-matrices}). 
    In terms of the permutation symmetric mixtures of Dicke states from \cref{sec:sd-symmetry}, the tensor $Q$ in the example above corresponds to the matrix 
    $$3\alpha \rho_{(3,0,\ldots,0)} + 6 \beta \rho_{(2,1,0\ldots,0)} +  \gamma \rho_{(1,1,1,0\ldots,0)}.$$
    The values discussed in the example above allow us to construct a matrix with $S_d$ symmetry which is PPT entangled in the case of 3 qutrits ($n=d=3$). 

    This is because the witness we consider based on the Robinson's polynomials is symmetric under variable permutation; such polynomials correspond dually to the sub-family of mixtures of Dicke states that have local basis permutation symmetry (\cref{sec:sd-symmetry}). 
\end{remark}

We conclude the section by showing that such positive non-SOS polynomials can be constructed for all $d > 3$ and $n> 3$. This conclusively settles the question of the existence of PPT entangled states in the diagonally symmetric subspace.  
\begin{theorem}[PPT entanglement in $d \geq 3, n\geq 3$]
\label{thm:qutrits-and-beyond}
    There exist entangled states that are PPT across all bipartitions for all $d \geq 3, n \geq 3$ in the DS subspace. 
\end{theorem}

\begin{proof}
We will exhibit homogeneous polynomials $p \in \homo{d}{n}$ in $d$ variables and degree $n$ such that $p(x \odot x) \geq 0$ but $p(x \odot x)$ is not a sum of squares. Note that for $d=3, n=3$, the Motzkin/Robinson polynomials are of this type. Moreover, \cref{ex:positive-non-sos} includes such polynomials for all $d \geq 3, n=3$. We show that for each $d$, examples for each $n$ can be obtained by induction. 

Let $p$ be a homogeneous polynomial of degree $k$ such that $p(x \odot x)$ is not a sum of squares. Then, $\tilde p = x_1 \cdot p$ is a homogeneous polynomial of degree $k + 1$. It is easy to verify that the polynomial $\tilde p (x \odot x)$ is positive as $\tilde p (x \odot x) = x_1^2 p(x \odot x) \geq 0.$ Assume that it has an SOS decomposition,
$$\tilde p (x \odot x) = \sum_{i} h^2_i (x).$$ Moreover, if $x_1 = 0$, it implies that $\tilde p (x \odot x) = 0$ and hence $h_i(x) = 0.$ This implies that $h_i = x_1 z_i$ for some homogenous polynomials $z_i \in \homo{d}{k}.$ Therefore, $x_1^2 p(x \odot x) = \tilde p (x \odot x)= \sum_{i} h^2_i (x) = x^2_1 z^2_i.$ This implies $p(x \odot x)$ is a sum of squares. Therefore, by contradiction, $\tilde p (x \odot x)$ is not a SOS polynomial. By starting this induction from the case $n=3$, we prove the theorem.
\end{proof}

Let us summarise our results in the table \ref{tab:ppt-entanglement-mods}.

\begin{table}[H]

\centering
\renewcommand{\arraystretch}{1.3}
\setlength{\arrayrulewidth}{0.5pt}
\arrayrulecolor{black}
\definecolor{lightblue}{RGB}{173, 216, 230}
\definecolor{lightorange}{RGB}{255, 204, 153}
\definecolor{lightgreen}{RGB}{204, 255, 204}

    \begin{minipage}[c]{0.55\linewidth}
        \centering
        \begin{tabular}{|c|c|c|c|}
        \hline
        \rowcolor{gray!20}
         & $n=2$ & $n=3$ & $n \geq 4$\\
        \hline
        $d=2$       & \cellcolor{lightblue}No & \cellcolor{lightblue}No & \cellcolor{lightblue}No \\
        \hline
        $d=3$       & \cellcolor{lightblue}No & \cellcolor{lightgreen}Yes & \cellcolor{lightgreen}Yes\\
        \hline
        $d=4$ & \cellcolor{lightblue}No & \cellcolor{lightgreen}Yes  & \cellcolor{lightgreen}Yes \\
        \hline
        $d \ge 5$ & \cellcolor{lightblue}Yes & \cellcolor{lightorange}Yes  & \cellcolor{lightorange}Yes \\
        \hline
        \end{tabular}
    \end{minipage} \hfill
    \begin{minipage}[c]{0.4\linewidth}
        \raggedright 
        \begin{tabular}{@{}ll@{}}
        \cellcolor{lightblue}\hspace{1.5em} & \cite{tura2018separability, yu2016separability} \\[4pt]
        \cellcolor{lightgreen}\hspace{1.5em} & This article \\[4pt]
        \cellcolor{lightorange}\hspace{1.5em} & \cite{romero2025multipartite, gulatiposext} 
        \end{tabular}
    \end{minipage}
    \caption{Is there PPT entanglement in mixtures of Dicke states of $n$ parties and local dimension $d$?}
    \label{tab:ppt-entanglement-mods}
\end{table}

\subsection{Marginals of entangled Dicke states are NPT}

In \cite{szalay2025dicke}, it was shown that for any marginal of a pure Dicke state $\ket{D_\alpha}$, such that $\operatorname{supp}(\alpha) > 1,$ is NPT (Negative Partial Transpose) across all bipartitions. If $\operatorname{supp}(\alpha) = 1$, then the state $\ket{D_\alpha}$ is separable. We recover the results slightly differently; we show that all the $2$-body marginals are entangled as they have negative partial transpose. Our proof is relatively straightforward, thanks to the characterization of marginals of Dicke mixtures from \cref{thm:marginals-dicke-states}.

\begin{proposition}
\label{prop:2-body-marginals}
    Consider an entangled Dicke state $\ket{D_\alpha}$ with $\alpha \in \knd{n}{d}$ such that $\operatorname{supp}(\alpha) >1$. Then, all the $2$-body marginals of $\ket{D_\alpha}$ have negative partial transpose (i.e.~they are NPT).
\end{proposition}
\begin{proof}

    Consider $l \neq m \in \operatorname{supp}(\alpha)$, and let $Q$ be the $Q$ parametrization of the marginal of the Dicke state corresponding to indices $l,m \in [d]$. By \cref{prop:marginal-dicke-pure} we have
    $$Q_{lm}  = \frac{\binom{n-2}{\alpha-\gamma(lm)}}{\binom{n}{\alpha}}\,\mathbb{1}_{\gamma(lm)\le \alpha},$$
    where $\gamma(lm)$ is the vector having 1's in positions $l$ and $m$ and 0's elsewhere. 
    We now claim that $Q_{ll} Q_{mm} < Q^2_{lm}$. To see this, we expand the two expressions:
\begin{align*}
Q_{ll} Q_{mm}
&= 
\frac{\binom{n-2}{\alpha-\gamma(ll)}}{\binom{n}{\alpha}}\,
\mathbb{1}_{\{\gamma(ll)\le \alpha\}}\;
\frac{\binom{n-2}{\alpha-\gamma(mm)}}{\binom{n}{\alpha}}\,
\mathbb{1}_{\{\gamma(mm)\le \alpha\}} \\[0.8em]
&=
\frac{1}{\binom{n}{\alpha}^2}\,
\frac{(n-2)!^2}{
\displaystyle 
\prod_{i\neq l,m} (\alpha_i!)^2\,
(\alpha_l-2)!\,\alpha_l!\,(\alpha_m-2)!\,\alpha_m!
}\,
\mathbb{1}_{\{\gamma(ll),\,\gamma(mm)\le \alpha\}} \\[0.8em]
& < 
\frac{1}{\binom{n}{\alpha}^2}\,
\frac{(n-2)!^2}{
\displaystyle 
\prod_{i\neq l,m} (\alpha_i!)^2\,
(\alpha_l-1)!^2\,(\alpha_m-1)!^2
}\,
\mathbb{1}_{\{\gamma(lm)\le \alpha\}} = Q^2_{lm}.
\end{align*}
where the strict inequality follows from noting that 
$\mathbb{1}_{\{\gamma(lm)\le \alpha\}} = 1$ and $$(\alpha_l-2)!\,\alpha_l!\,(\alpha_m-2)!\,\alpha_m! > (\alpha_l-1)!^2 (\alpha_m-1)!^2.$$ This shows that the matrix $Q$ is not PSD, and hence the $2$-body marginal state $X$ is NPT. 
\end{proof}

\begin{corollary}[\cite{szalay2025dicke}]
  Let $\alpha \in \knd{n}{d}$ such that $\operatorname{supp}(\alpha) >1$. Then, all the marginals of the entangled pure Dicke state $\ket{D_\alpha}$ are NPT across all bipartitions.
\end{corollary}
\begin{proof}
    Let $X = \operatorname{Tr}_{[r]}(\ketbra{D_\alpha}{D_\alpha})$ for $1 \leq r < n$. Then, we claim that $X^{\Gamma_{[k]}} \nsucceq 0$ for all $1 \leq k < r$. This can be shown by noting that
    $$\operatorname{Tr}_{[r] \backslash \{1,k+1\}}(X^{\Gamma_{[k]}}) = \big(\operatorname{Tr}_{[r] \backslash \{1,k+1\}}(X)\big)^\Gamma = \big(\operatorname{Tr}_{[n] \backslash \{1,k+1\}}(\ketbra{D_\alpha}{D_\alpha})\big) ^\Gamma.$$ By \cref{prop:2-body-marginals}, $\big(\operatorname{Tr}_{[n] \backslash \{1,k+1\}}(\ketbra{D_\alpha}{D_\alpha})\big) ^\Gamma\nsucceq 0$ implying $X^{\Gamma_{[k]}} \nsucceq 0$.
\end{proof}

\section{Bosonic extendibility}
\label{sec:bosonic-extendibility}
In the final section of this article, we shift our focus to another important entanglement-theoretic property of quantum states known as \emph{extendibility} \cite{doherty2004complete, terhal2004entanglement,christandl2007one}. This property is concerned with the existence of an extension, i.e, a state supported on a larger number of systems, such that some chosen reductions of such a state are equal to the original quantum state. There exist various notions of the extendibility of quantum states, see \cite{doherty2004complete,navascues2009power, allerstorfer2026monogamy,solymos2024extendibility,gulatiposext}. In this article, we study the one that is well-suited to bosonic states, which we call the \emph{hypergraph extendibility}, generalizing the complete graph extendibility from \cite{allerstorfer2026monogamy}. Let us begin by defining this notion.

\begin{definition}
Let $r,n \in \mathbb{N}$ be non-negative integers. A matrix $\tilde X \in \operatorname{Herm}[\vee^{n+r} \mathbb{C}^d]$ is called a \emph{$(n,r)$-hypergraph bosonic extension} of a matrix $X \in  \operatorname{Herm}[\vee^n \mathbb{C}^d]$ if both of the following conditions hold: 
\begin{itemize}
    \item[(1)] $\tilde X \succeq 0$
    \item[(2)] $\forall \big(A \subseteq [n+r], \,  |A| = r\big) \quad  \operatorname{Tr}_{A}(\tilde X) = X$.
\end{itemize}
We call such an extension a \emph{$(n,r)$-hypergraph PPT bosonic extension} if the extension $\tilde X$ is also PPT across all bipartitions. Note that, thanks to the symmetry property, one could have replace the ``$\forall$'' in the second point above by an ``$\exists$''.
\end{definition}

If $X, Y$ have a $(n,r)$-hypergraph (PPT) bosonic extension $\tilde X, \tilde Y$ (resp.), then $\tilde X + \tilde Y$ is the $(n,r)$-hypergraph (PPT) bosonic extension of $X+Y$. Similarly, $\tilde X$ is a $(n,r)$-hypergraph (PPT) bosonic extension of $X$ , then, $\lambda \tilde X$ is the  $(n,r)$-hypergraph (PPT) bosonic extension of $\lambda X.$ We denote these cones of $(n,r)$-hypergraph bosonic extendable states as $\mathsf{BExt}_d^{(n,r)}$ and $(n,r)$-hypergraph PPT bosonic extendable states as $\mathsf{PPTBExt}_d^{(n,r)}$, where $r \geq 0$:
\begin{align}
\label{eq:bext-sets}
    \mathsf{BExt}_d^{(n,r)} &:= \{ X \in \operatorname{Herm}[\vee^n \mathbb{C}^d] \mid \exists \tilde X \in \mathsf{PSD}^{(n+r)}_d \text{ s.t. } \Tr_{[r]}\tilde X  = X\}\\
    \mathsf{PPTBExt}_d^{(n,r)} &:= \{ X \in \operatorname{Herm}[\vee^n \mathbb{C}^d] \mid \exists \tilde X \in \mathsf{PPT}^{(n+r)}_d \text{ s.t. } \Tr_{[r]}\tilde X  = X \}.
\end{align}

\begin{remark}
Note that for $r = 0$, $\mathsf{BExt}_d^{(n,0)}$ reduces to the set of all $n$-partite bosonic states, while $\mathsf{PPTBExt}_d^{(n,0)}$ reduces to the set of $n$-partite bosonic states with all bipartitions PPT:
\begin{align*}
    \mathsf{BExt}_d^{(n,0)} &= \mathsf{PSD}_d^{(n)}\\
    \mathsf{PPTBExt}_d^{(n,0)} &= \mathsf{PPT}_d^{(n)}.
\end{align*}
\end{remark} 

Recall that any separable state in the bosonic subspace, $X \in \operatorname{Herm}[\vee^n \mathbb{C}^d]$ is of the form
$$X = \sum^K_{k=1} \ketbra{v_k}{v_k}^{\otimes n}\text{ for }\ket{v_k} \in \mathbb{C}^d.$$ 
Consider, for any $r \geq 0$, the following operator $\tilde X \in \operatorname{Herm}[\mathcal{M}^{\otimes (n+r)}_d]$:  
$$\tilde X := \sum^K_{k=1} \frac{1}{||v_k||^{2r}}\ketbra{v_k}{v_k}^{\otimes(n+r)}.$$
It is easy to verify that $\tilde X$ is an $(n,r)$-hypergraph bosonic extension of $X$. Moreover, this extension is also PPT across all bipartitions.  Hence, we have the following inclusions, showing that the newly introduced sets  form a decreasing hierarchy of outer approximations of the separable set:
{
\setlength{\arraycolsep}{3pt} 
\[
\begin{array}{cccccccccccccc} 
\mathsf{Sep}^{(n)}_{d}& \subseteq & \cdots
& \subseteq & \mathsf{BExt}_d^{(n,r)}
& \subseteq & \cdots
& \subseteq & \mathsf{BExt}_d^{(n,1)}
& \subseteq & \mathsf{BExt}_d^{(n,0)}
& = & \mathsf{PSD}_d^{(n)}\\[8pt]

\upwardeq & & \cdots & & \upwardsubseteq & & \cdots & & \upwardsubseteq & & \upwardsubseteq & & \upwardsubseteq \\[8pt]

\mathsf{Sep}^{(n)}_{d}& \subseteq & \cdots
& \subseteq & \mathsf{PPTBExt}_d^{(n,r)}
& \subseteq & \cdots
& \subseteq & \mathsf{PPTBExt}_d^{(n, 1)}
& \subseteq & \mathsf{PPTBExt}_d^{(n, 0)}
& = & \mathsf{PPT}_d^{(n)}
\end{array}
\]
}

We now provide a formal characterization of the dual sets, i.e, the set of \emph{extendibility witnesses}. Such a characterization has already been obtained in \cite[Theorem 8.9]{gulatiposext} and \cite[Theorem 4.1]{britz2025semidefinite} for the bipartite case $n=2$. The essential argument can be easily extended to the case of $n > 2$; we provide the proof for the sake of completeness. 

\begin{theorem}[PPT Extendibility Witness]
\label{thm:ppt-ext-witness}
The following statements are equivalent,
\begin{itemize}
    \item[(1)] $X \in (\mathsf{PPTBExt}^{(n,r)}_d)^\circ =:  (\mathsf{DecExtW}^{(n,r)}_d)$
    \item[(2)] $\Pi^{(n+r)}_{d} (X \otimes I_d^{\otimes r}) \Pi^{(n+r)}_{d} \in \sum^{\lfloor (n+r)/2 \rfloor}_{k=0}\mathsf{Dec}^{(n+r,k)}_d$, i.e. there exist PSD matrices $W_k \in \operatorname{Herm}[\mathcal{M}^{\otimes (n+r)}_d]$ for $k=0, \ldots, \lfloor (n+r)/2 \rfloor$ such that 
    $$\Pi^{(n+r)}_{d} \big( X \otimes I_d^{\otimes r} \big) \Pi^{(n+r)}_{d} = \Pi^{(n+r)}_d\Biggl( W_0 + \sum^{\lfloor (n+r)/2 \rfloor}_{k=1} W_k^{\Gamma_{[k]}} \Biggr)  \Pi^{(n+r)}_d.$$ 
\end{itemize}
\end{theorem} 

\begin{proof} For the implication $(1)\implies (2)$, consider an arbitrary $X \in (\mathsf{PPTBExt}_d^{(n,r)})^\circ$. By definition, we have
    $\Tr(XW)\ge 0$ for all $W \in \mathsf{PPTBExt}_d^{(n,r)}$. This also implies that for all $\tilde W \in \bigcap^{\lfloor (n+r)/2 \rfloor}_{k=0}\mathsf{PPT}_d^{(n+r,k)}$, we have 
    $$\operatorname{Tr}((X \otimes I_d^{\otimes r}) \tilde W) = \operatorname{Tr}(X \operatorname{Tr}_{[r]} (\tilde W)) \geq 0.$$
    Therefore, $\Pi^{(n+r)}_{d} (X \otimes I_d^{\otimes r}) \Pi^{(n+r)}_{d} \in \sum^{\lfloor (n+r)/2 \rfloor}_{k=0} \mathsf{Dec}^{(n+r,k)}_d =\mathsf{Dec}^{(n+r)}_d$, see  \cref{cor:duality-ppt-dec}. This shows the first claim.
    
    For the reverse direction, we will show that $\Tr(XW)\ge 0$ for all $W \in \mathsf{PPTBExt}_d^{(n,r)}$.
    If $W \in \mathsf{PPTBExt}_d^{(n,r)}$, then there exists a $\tilde W \in \bigcap^{\lfloor (n+r)/2 \rfloor}_{k=0}\mathsf{PPT}_d^{(n+r,k)}$ such that $\operatorname{Tr}_{[r]} (\tilde W) = W.$ Then, we can show that
    $$\operatorname{Tr}(XW) = \operatorname{Tr}((X \otimes I^{\otimes r}_d) \tilde W) \geq 0.$$
\end{proof}

If we restrict to $k=0$ in the theorem above, we eliminate the non-trivial PPT conditions to obtain the following result. 

\begin{corollary}
\label{cor:ext-witness}
    Let $X \in \operatorname{Herm}[\mathcal{M}^{\otimes n}_d]$ be a self-adjoint matrix such that
    $$\Pi^{(n)}_{d} X \Pi^{(n)}_{d} =  X.$$ 
    Then, $X \in \big(\mathsf{BExt}_d^{(n,r)}\big)^\circ := (\mathsf{ExtW}^{(n,r)}_d)$ if and only if $$\Pi^{(n+r)}_{d} \big( X \otimes I_d^{\otimes r}\big) \Pi^{(n+r)}_{d} \succeq 0.$$ 
\end{corollary}

\subsection{Moment hierarchy for CP Tensors}
In general, checking for the copositivity of a matrix is computationally
co-NP-hard \cite{murty1987some}, while on the dual side, checking for
complete positivity of a matrix is NP-hard \cite{dickinson2014computational}. Despite this, there is a systematic way to approximate these sets through inner and outer approximations. We will briefly review the two widely used inner approximations for the set of copositive matrices: the SOS hierarchy and the non-negative hierarchy. The membership at any level of the hierarchy can be decided using a semi-definite/linear program. On the dual side, the duals of these hierarchies outer approximate the completely positive matrices. We also discuss the natural extension of these hierarchies to tensors, and in the next section, provide a surprising connection to the extendibility properties in the diagonally symmetric (DS) subspace. Recall the following tensor duality from \cref{thm:duality-moment-sos}, 

$$(\mathsf{Mom}^{(n)}_d)^{\circ} = \mathsf{SOS}^{(n)}_d.$$ 
Then, we define the following convex cones:
\begin{equation}
\label{eq:mom-ext}
\mathsf{MomExt}^{(n,r)}_d := \{\operatorname{tr}_{[r]}(X) \mid  X \in \mathsf{Mom}^{(n+r)}_d\}.
\end{equation}
For a particular $r \geq 0$, this is the convex cone of tensors in $\vee^n \mathbb{R}^d$ that have a moment tensor extension in the larger space $\vee^{(n+r)} \mathbb{R}^d.$ Note that any completely positive tensor $Q = \sum_{k} v^{\otimes n}_k$ with $v_k \in \mathbb R_+^d$ has a moment tensor extension for all $r$, given by
$$\tilde Q = \sum_{k} \frac{1}{|v_k|^r}v^{\otimes (n+r)}_k.$$

Therefore, we have the following infinite hierarchy of sets that approximate completely positive matrices from the outside:
$$\mathsf{CP}^{(n)}_d \subseteq \cdots \subseteq \mathsf{MomExt}^{(n,r)}_d \subseteq  \cdots \subseteq \mathsf{MomExt}^{(n,1)}_d \subseteq \mathsf{MomExt}^{(n,0)}_d = \mathsf{Mom}^{(n)}_d.$$ 

Unsurprisingly, this is dual to the sum of squares hierarchy, which approximates copositive matrices from the inside. This is the content of our next theorem.

\begin{theorem}[Dual of SOS hierarchy]
\label{thm:sos-mom-duality}
We have the following dual coneity:
$$\forall n,r,d \in \mathbb{N}, \quad \big(\mathsf{MomExt}^{(n,r)}_d \big)^\circ = \mathsf{RSOS}^{(n,r)}_d.$$
\end{theorem}
\begin{proof}

Let us first note that $\mathsf{RSOS}^{(n,r)}_d = \{Q \in \vee^n \mathbb{R}^d \mid \Pi^{(n+r, \mathbb{R})}_d[Q \otimes \ket e^{\otimes n}] \in \mathsf{SOS}^{(n+r)}_d\}$. Note also that the trace operation is dual to the tensoring operation with the normalized all-ones tensor:
$$\braket{\operatorname{tr}_{[r]}(T),Q} = \braket{T,Q \otimes \ket{e}^{\otimes r}}.$$ Also if the tensor $T$ is symmetric, we have
\[
\braket{T, Q \otimes \ket{e}^{\otimes r}}
= \braket{\Pi^{n+r,\mathbb{R}}_d[T], Q \otimes \ket{e}^{\otimes r}} = \braket{T,\Pi^{n+r,\mathbb{R}}_d[Q \otimes \ket{e}^{\otimes r}]}.
\]
Finally, we have the following equivalences:
\begin{align*}
X \in \big(\mathsf{MomExt}^{(n,r)}_d \big)^\circ
&\iff \forall\, Y \in \mathsf{MomExt}^{(n,r)}_d,\quad \langle X, Y \rangle \ge 0 \\
&\iff \forall\, \tilde Y \in \mathsf{Mom}^{(n+r)}_d,\quad 
    \big\langle \operatorname{tr}_{[r]}(\tilde Y), X \big\rangle \ge 0 \\
&\iff \forall\, \tilde Y \in \mathsf{Mom}^{(n+r)}_d, \quad
    \big\langle \tilde Y,
    \Pi^{n+r,\mathbb{R}}_d\big[ X \otimes \ket{e}^{\otimes r} \big]
    \big\rangle \ge 0. \\
&\iff \Pi^{n+r,\mathbb{R}}_d\big[ X \otimes \ket{e}^{\otimes r} \big] \in (\mathsf{Mom}^{(n+r)}_d)^\circ = \mathsf{SOS}^{(n+r)}_d.
\end{align*}
This completes the proof. 
\end{proof}

Analogous to the case of the sum of squares hierarchy, we will characterize the dual of Polya's non-negative hierarchy. To this end, we define the following convex cones of tensors:
\begin{equation}
\label{eq:nn-ext}
\mathsf{NNExt}^{(n,r)}_d := \{\operatorname{tr}_{[r]}(X) \mid  X \in \mathsf{NN}^{(n+r)}_d\}.
\end{equation}
These are exactly the convex cones of tensors that have an entrywise positive extension of order $n+r$. Therefore, we have the following new infinite hierarchy of sets that approximate completely positive tensors from the outside:
$$\mathsf{CP}^{(n)}_d \subseteq \cdots \subseteq \mathsf{NNExt}^{(n,r)}_d \subseteq \cdots \subseteq \mathsf{NNExt}^{(n,1)}_d  \subseteq \mathsf{NNExt}^{(n,0)}_d = \mathsf{NN}^{(n)}_d.$$ 

In the next theorem, we present the conic-duality of the non-negative hierarchy. We skip the proof, as it is essentially the same argument as in \cref{thm:sos-mom-duality}, by replacing "$\mathsf{Mom}$" with "$\mathsf{NN}$".

\begin{theorem}[Dual of NN hierarchy]
We have the following duality:
$$\forall n,r,d \in \mathbb{N}, \quad \big(\mathsf{NNExt}^{(n,r)}_d \big)^\circ = \mathsf{PNN}^{(n,r)}_d.$$
\end{theorem}

\begin{remark}
    The entrywise positive hierarchy that we discuss here has been abstractly studied in the probability literature as ``finitely extendable and exchangeable probability measures''. We refer to the recent article \cite{konstantopoulos2019extendibility}, and all references therein, for the finite-dimensional case of this problem. 
\end{remark}

\subsection{Bosonic extendibility: extending the dictionary}
In the previous section, we discussed two important hierarchies  (and their duals) that approximate the set of copositive tensors (and completely positive tensors). In this section, we connect these hierarchies to the bosonic extendibility problem in the diagonally symmetric subspace. To this end, we have the following lemma, which states that a diagonally symmetric matrix has an $(n,r)$-hypergraph extension if and only if there is an extension in the DS subspace. This result follows from the local symmetry properties of the DS subspace, and has already been noted in \cite{gulatiposext, britz2025semidefinite}. We provide the proof for completeness.

\begin{lemma}
\label{lem:extension-also-symmetric}
    Let $X \in \mathsf{DS}^{(n)}_d$. Then, the following statements are equivalent:
    \begin{itemize}
        \item[(1)] $X$ has a (PPT) $r$-bosonic extension.
        \item[(2)] $X$ has a (PPT) $r$-bosonic extension that is in $\mathsf{DS}^{(n+r)}_d$.
    \end{itemize}
\end{lemma}

\begin{proof}
    We only do the case without the PPT property, as the PPT case is analogous. The reverse implication is trivial. For the forward implication, let $X \in \mathsf{BExt}^{(n,r)}_d,$ i.e there exists an extension $\tilde X \in \mathcal{M}^{\otimes (n+r)}_d$ with $\Pi^{(n+r)}_d \tilde X \Pi^{(n+r)}_d = \tilde X$. We claim that the twirl with respect to the diagonal unitary group $\diagU_d$ 
    $$\mathcal{T}_{\diagU} (\tilde X) := \int_{\diagU_d} U^{\otimes (n+r)} \tilde X \overbar{U}^{\otimes (n+r)}\, \mathrm{d}U$$
    is also an bosonic extension of $X$. Indeed, use $[\Pi^{(n+r)}_d, U^{\otimes (n+r)}] = 0$ and compute 
\begin{align*}
\operatorname{Tr}_{[r]} \big(\mathcal{T}_{\diagU} (\tilde X)\big)
&= \int \operatorname{Tr}_{[r]}\!\big(U^{\otimes (n+r)} \tilde X\, \overbar{U}^{\otimes (n+r)} \big) \, \mathrm{d}U\\
&= \int U^{\otimes n}\, \operatorname{Tr}_{[r]} (\tilde X)\, \overbar{U}^{\otimes n} \, \mathrm{d}U \\
&= \int U^{\otimes n} X \overbar{U}^{\otimes n}\, \mathrm{d}U \\
&= X,
\end{align*}
where the last equality follows from the fact that $X \in \mathsf{DS}^{(n)}_d$.
\end{proof}

Now, we use the previous lemma to establish the main result of this section, the connection between the bosonic-extendibility of mixtures of Dicke states and hierarchies for completely positive tensors. This has already been proven in the recent papers \cite{gulatiposext, britz2025semidefinite} for the bipartite case $n=2$.

\begin{theorem}\label{thm:extendibility-dicke-states}
   Let $X \in \mathsf{DS}^{(n)}_d.$ Then, for all level of the hierarchies $r \geq 0$:
    \begin{itemize}
        \item[(1)] $X \in \mathsf{BExt}^{(n,r)}_d \iff Q[X] \in \mathsf{NNExt}_d^{(n,r)}$
        \item[(2)] $X \in \mathsf{PPTBExt}^{(n,r)}_d \iff Q[X] \in \mathsf{MomExt}_d^{(n,r)}$.
    \end{itemize}
\end{theorem}
\begin{proof}
    For the forward implication of item (1), assume that $X \in \mathsf{BExt}^{(n,r)}_d.$ Then, by \cref{lem:extension-also-symmetric}, there exists  $\tilde X \in \mathsf{DS}_d^{(n,r)}$ which is PSD, and thus $Q[\tilde X ] \in \mathsf{NN}^{(m+r)}_d.$ We also have $Q[X] = Q[\operatorname{Tr}_{[r]}(\tilde X)] = \operatorname{tr}_{[r]}(Q[\tilde X])$ where $Q[\tilde X] \in \mathsf{NN}^{(n+r)}_d$. This implies that $Q[X] \in \mathsf{NNExt}_d^{(n,r)}.$  For the reverse implication, let  $Q[X] \in \mathsf{NNExt}_d^{(n,r)}$. There exists $\tilde Q \in \mathsf{NN}^{(n+r)}_d$ such that $\operatorname{tr}_{[r]}(\tilde Q) = Q$. Then, the state $\tilde X \in \mathsf{DS}^{(n+r)}_d$ such that $Q[\tilde X] = \tilde Q$ is the bosonic extension of $X$.   
    
    The proof of the second item follows from replacing $\mathsf{NN}$ by $\mathsf{Mom}$, PSD by PPT, and bosonic extension by PPT bosonic extension in part (1).
\end{proof}

\begin{remark}
    Our proof of the previous theorem is considerably shorter than the proofs in \cite{gulatiposext, britz2025semidefinite}, as the hard part of the proof was relegated to establishing a more general dictionary in the previous sections of this work.
\end{remark}

Finally, we can also establish the dual result for the extendibility witnesses.
\begin{corollary}
\label{cor:extendibility-witnesses}
    The following hold for operators $O \in \mathsf{DS}_d^{(n)}$:
   \begin{itemize}
        \item[(1)] $O \in \mathsf{ExtW}^{(n,r)}_d\iff W[O] \in \mathsf{PNN}_d^{(n,r)}$
        \item[(2)] $O \in \mathsf{DecExtW}^{(n,r)} \iff W[O] \in \mathsf{RSOS}_d^{(n,r)}$
    \end{itemize}
\end{corollary}

\begin{remark}
    Note that the theorem provides a much more powerful SDP hierarchy to check separability in diagonally symmetric states. The largest slice flattening of a symmetric tensor of $2n$ legs is of size $\binom{n+d-1}{d-1}.$ Therefore, to check $Q[X] \in \mathsf{MomExt}_d^{(n,r)},$ we need to check the PSD of the largest matrix of size $\binom{\lfloor \frac{n+r}{2} \rfloor+d-1}{d-1}.$
\end{remark}

This result completes the picture of all the hierarchies for bosonic extendibility and their duals, see \cref{fig:full-diagram} for the full picture of correspondences, inclusions, and dualities. 


\section{Conclusion}
In this article, we present a detailed study of multipartite entanglement in the diagonally symmetric subspace, with a particular focus on PPT entanglement. This forms a crucial subclass of matrices with bosonic symmetry. By defining a tensor-based parametrization of these states, we translate the entanglement properties of these states as convex cones of symmetric tensors. In particular, the important entanglement-theoretic cones of separable and PPT states can be described equivalently in terms of the complete positivity and moment positivity of the parametrizing tensor. This extends the results in \cite{tura2018separability, yu2016separability, romero2025multipartite, marconi2021entangled} to provide a comprehensive theory for this class of states. Our results provide a bridge between the study of multipartite entanglement, semi-algebraic geometry, and optimization theory.

Our results also lead to several interesting corollaries in the DS subspace. Firstly, we show the existence of PPT entangled states in $3$ or more qutrits and ququarts, disproving the recent conjecture from \cite{romero2025multipartite}. Quite surprisingly, the existence of entanglement witnesses turns out to be closely related to the theory of positive polynomials that are not sums of squares. Using this perspective, we develop a dictionary that links extendibility witnesses on the diagonally symmetric subspace with the sum-of-squares (SOS) hierarchy for copositive tensors. This viewpoint not only clarifies the underlying structure of these witnesses, but also extends earlier extendibility results known in the bipartite setting~\cite{gulatiposext, britz2025semidefinite} to a broader symmetric framework.

As a further application, we show that the two-body marginals of pure Dicke states are NPT and therefore necessarily entangled. Our approach yields a shorter and more transparent proof of the corresponding results in~\cite{szalay2025dicke}.

More broadly, we expect that the connections developed here between entanglement theory, polynomial positivity, and SOS methods will open several new research directions. We conclude by outlining two concrete avenues for future work.

Firstly, it will be interesting to understand the multipartite states that are invariant under the diagonal unitary group, generalizing the LDUI (Local Diagonal Unitary Invariant) states from \cite{singh2021diagonal, Johnston_2019}. The bipartite versions of LDUI states can be parameterized using matrix pairs, and entanglement-theoretic properties can be understood as convex cones of matrix pairs. The bipartite mixtures of Dicke states are exactly the LDUI states where these matrices are equal. Therefore, it seems possible that multipartite LDUI states allow a similar description, replacing matrices with \emph{tensors}.

The second question centers on extendibility hierarchies. In this paper, we consider a hierarchy that was well-suited for bosonic states, which we call the $(n,r)$-hypergraph extendibility. There are other extendibility hierarchies introduced for quantum states in the literature. Since bosonic states are permutation symmetric (i.e. they satisfy $P_{\pi} \rho = \rho P_{\pi}$ for all $\pi \in S_n$), can one imagine another extendibility hierarchy in which the extensions are required to be permutation symmetric instead of bosonic? Is there a gap between this hierarchy and the one we consider for bosonic states?

\section{Author Contribution}
All authors contributed equally to this work. LLM Models were used to produce the Tikz code for Figure 1, which was independently verified by the authors.

\section{Acknowledgments}
A.G. thanks Jordi Romero-Pallejà and Anna Sanpera for the invitation to discuss a preliminary version of this work, and the connection to \cite{romero2025multipartite}. The authors thank Jonas Britz and Monique Laurent for discussions surrounding their recent work \cite{britz2025semidefinite}. The authors thank Sang-Jun Park for insightful discussions on this topic and for useful comments on the preliminary draft. 
The three authors were supported by the ANR project \href{https://esquisses.math.cnrs.fr/}{ESQuisses} grant number ANR-20-CE47-0014-01. A.G received support from the University Research School EUR-MINT (State support managed by the National Research Agency for Future Investments program bearing the reference ANR-18-EURE-0023). I.N.~was supported by the ANR project \href{https://www.ceremade.dauphine.fr/dokuwiki/anr-tagada:start}{TAGADA} grant number ANR-25-CE40-5672.

\bibliography{references}
\bibliographystyle{quantum}

\appendix
\section{Proof of the decomposability result}\label{sec:appendix-proof-Dec}
\begin{proof}[Proof of \cref{thm:Dec-n-k-SOS}]

We shall prove separately the two items in the statement. 

\medskip

\noindent\underline{Proof of the Item 1.}

Consider $O \in \mathsf{Dec}_d^{(n,k)}$ for some integer $k \leq n/2$, and let $p:=p_{W[O]}$ the associated polynomial. Our goal is to show that the polynomial $p$ belongs to the cone 
$$\operatorname{cone}\{x^\alpha q \mid |\alpha| \in \{n, n-2, n-4 \ldots n-2k\}, q \in \Sigma^{(n-|\alpha|)}_d\}.$$
From $O \in \mathsf{Dec}_d^{n,k}$, we have 
$$p(|z|^2) = \langle z^{\otimes n} | P^{\Gamma_{[k]}} | z^{\otimes n} \rangle$$
for some positive senmidefinite matrix $P$. Using the spectral decomposition of $P$, we can write 
$$p(|z|^2) = \sum_{i=1}^r \langle z^{\otimes n} | \ketbra{t_i}{t_i}^{\Gamma_{[k]}} | z^{\otimes n} \rangle$$
for tensors $t_i \in (\mathbb R^d)^{\otimes n}$. By linearity, it is enough to prove the claim for the polynomial corresponding to the general term above, so, without loss of generality, we can assume that $O = \ketbra{t}{t}^{\Gamma[k]}$ for some real tensor $t$. On the one hand, we have 
\begin{align}
    \nonumber p(|z|^2) &= \langle W[O], (|z|^2)^{\otimes n} \rangle \\
    \nonumber &= \sum_{i \in [d]^n} W[O]_i |z_{i_1}|^2 \cdots |z_{i_n}|^2\\
    &= \sum_{\alpha \in \mathbb N_d^{(n)}} \binom{n}{\alpha} W[O]_{\idx(\alpha)} \bar z^\alpha z^\alpha,\label{eq:p-z2-alpha-alpha}
\end{align}
where 
$$\mathbb N_d^{(n)} := \{ \alpha \in \mathbb Z^d \mid \alpha_i \geq 0 \, \forall i \in [d] \text{ and } \sum_i \alpha_i = n\}$$
is the set of $d$-tuples of non-negative integers that sum up to $n$ and corresponds to the types of multi-indices $i \in [d]^n$. 

On the other hand, using the form of $O$ and moving the partial transposition on the ``$z$'' factors, we also have:
$$p(|z|^2) = \langle z^{\otimes n} | \ketbra{t}{t}^{\Gamma_{[k]}} | z^{\otimes n} \rangle =  \langle \bar z^{\otimes k} \otimes z^{\otimes(n-k)} | \ketbra{t}{t} | \bar z^{\otimes k} \otimes z^{\otimes(n-k)}  \rangle  = |\langle \bar z^{\otimes k} \otimes z^{\otimes(n-k)} | t \rangle|^2.$$
In the last scalar product above, we have 
$$\bar z^{\otimes k} \otimes z^{\otimes(n-k)} \in \vee^k \mathbb C^d \otimes \vee^{n-k} \mathbb C^d,$$
hence we can assume, without loss of generality, $t \in \vee^k \mathbb R^d \otimes \vee^{n-k} \mathbb R^d$, by projecting it on the corresponding tensor product of symmetric subspaces of $\mathbb R^d$. We can now write 

\begin{align*}
    p(|z|^2) &= \sum_{i,j \in [d]^n} \bar z_{i_1} \cdots \bar z_{i_k} z_{i_{k+1}} \cdots z_{i_n} t_{i_1 \cdots i_n} z_{j_1} \cdots z_{j_k} \bar z_{j_{k+1}} \cdots \bar z_{j_n} t_{j_1 \cdots j_n} \\
    &= \sum_{\substack{\alpha', \beta' \in \mathbb N_d^{(k)}\\\alpha'', \beta'' \in \mathbb N_d^{(n-k)}}} \binom{k}{\alpha'}\binom{n-k}{\alpha''}\binom{k}{\beta'}\binom{n-k}{\beta''}\bar z^{\alpha'+\beta''} z^{\beta' + \alpha''} t_{\idx(\alpha'), \idx(\alpha'')} t_{\idx(\beta'), \idx(\beta'')}\\
    &= \sum_{\substack{\alpha', \beta' \in \mathbb N_d^{(k)}\\\alpha'', \beta'' \in \mathbb N_d^{(n-k)} \\ \alpha'+\beta'' = \beta'+\alpha''}} \binom{k}{\alpha'}\binom{n-k}{\alpha''}\binom{k}{\beta'}\binom{n-k}{\beta''}\bar z^{\alpha'+\beta''} z^{\beta' + \alpha''} t_{\idx(\alpha'), \idx(\alpha'')} t_{\idx(\beta'), \idx(\beta'')}
\end{align*}
where the last equality follows from \cref{eq:p-z2-alpha-alpha} which allows us to discard terms with $\alpha'+\beta'' \neq \beta'+\alpha''$.

From the last formula for $p(|z|^2)$ we would like to extract a decomposition as in the statement, as a product of a monomial in $|z|^2$ and a sum-of-squares polynomial $q$ evaluated at $|z|^2$. To this end, define the following $d$-tuples of non-negative integers:
\begin{align}
    \label{eq:def-mu'} \mu' &:= \min(\alpha',\beta')\\
    \nonumber \mu'' &:= \min(\alpha'',\beta'')\\
    \nonumber \mu &:= \min(\alpha',\beta',\alpha'',\beta'') = \min(\mu',\mu'')\\
    \nonumber \delta' &:= \mu' - \mu\\
    \nonumber \delta'' &:= \mu'' - \mu\\
    \nonumber \hat \alpha' &:= \alpha' - \delta'\\
    \nonumber \hat \alpha'' &:= \alpha'' - \delta''\\
    \nonumber \hat \beta' &:= \beta' - \delta'\\
    \label{eq:def-hat-beta''} \hat \beta'' &:= \beta'' - \delta'',
\end{align}
where the $\min$ operation is taken entry-wise. We have, for example, 
$$\hat \alpha'_i = \alpha'_i - (\min(\alpha'_i, \beta'_i) - \min(\alpha'_i, \beta'_i,\alpha''_i,\beta''_i)) \geq \alpha'_i - \min(\alpha'_i, \beta'_i) \geq 0.$$
Importantly, note that for every $i \in [d]$, we either have $\mu'_i = \mu_i$ or $\mu''_i = \mu_i$; in other words, the supports of the vectors $\delta'$ and $\delta''$ are disjoint. We shall write this as $\delta' \delta''=0$ using the entry-wise product of vectors. 

We shall now prove the following very important fact: 
\begin{equation}\label{eq:hats-equal}
    \hat \alpha ' = \hat \alpha'' \quad \text{ and } \quad \hat \beta' = \hat \beta''.
\end{equation}
For the $\alpha$'s, fix some coordinate $i \in [d]$ and denote by the Latin letter the $i$-th coordinate of the respective Greek letter (e.g.~$\hat b'' = \hat \beta''_i$). From the constraint $\alpha' + \beta'' = \beta' + \alpha''$, we obtain $a'+b'' = b'+a''$. This means that there exist rational numbers $c,x,y$ such that 
$$a'=c+x \quad b'' = c-x \quad b' = c+y \quad a'' = c-y.$$
We have then $m' = c + \min(x,y)$ and $m'' = c - \max(x,y)$; recall that $m = \min(m',m'')$. We now have 
$$\hat a' = c+x - (c+\min(x,y)) + m \quad \text{ and } \quad \hat a'' = c - y - (c-\max(x,y)) + m$$
and thus 
$$\hat a' - \hat a'' = x+y - \min(x,y) - \max(x,y) = 0.$$
A similar argument shows $\hat b' = \hat b''$, proving the claim. From now on, we shall write 
\begin{align*}
    \hat \alpha &:= \hat \alpha' = \hat \alpha''\\
    \hat \beta &:= \hat \beta' = \hat \beta''.
\end{align*}

From the claim in \cref{eq:hats-equal}, we obtain the following constraints on the non-negative vectors $\delta', \delta''$:
\begin{align*}
    |\delta'| &\leq k\\
    |\delta''| &\leq n-k\\
    |\delta''| - |\delta'| &= n-2 k.
\end{align*}

\medskip

Going back to the formula for $p(|z|^2)$, we can now write 
$$p(|z|^2) = \sum_{\substack{\alpha', \beta' \in \mathbb N_d^{(k)}\\\alpha'', \beta'' \in \mathbb N_d^{(n-k)} \\ \alpha'+\beta'' = \beta'+\alpha''}} \binom{k}{\hat \alpha + \delta'}\binom{n-k}{\hat \alpha + \delta''}\binom{k}{\hat \beta + \delta'}\binom{n-k}{\hat \beta + \delta''} (|z|^2)^{\hat \alpha + \hat \beta + \delta' + \delta''} t_{\idx(\hat \alpha + \delta'), \idx(\hat \alpha + \delta'')} t_{\idx(\hat \beta + \delta'), \idx(\hat \beta + \delta'')},$$
where the 4-tuple $(\delta', \delta'', \hat \alpha, \hat \beta)$ is a function of the 4-tuple $(\alpha', \beta', \alpha'', \beta'')$. 

We prove now a second crucial fact: the following two sets $S$ and $T$ are in bijection: 
$$S:=\{ (\alpha', \beta', \alpha'', \beta'') \in \mathbb N^d \mid |\alpha'| = |\beta'| = k, \, |\alpha''| = |\beta''| = n-k, \, \alpha'+\beta'' = \beta'+\alpha''\}$$
\begin{align*}
    T:=\{ (\delta', \delta'', \hat \alpha, \hat \beta) \in \mathbb N^d \mid &|\delta'| \leq k, \, |\delta''| \leq n-k, \, |\delta''| - |\delta'| = n-2k,\,  \delta' \delta'' = 0,\\ 
    &|\hat \alpha| = |\hat \beta| = k-|\delta'| = n-k-|\delta''|\}.
\end{align*}
Note that some conditions in the definition of $T$ are superfluous, but we prefer to write them for clarity. To prove the claim, we shall construct two maps $\phi:T \to S$, $\psi:S \to T$ and prove that $\psi \circ \phi = \mathrm{id}_T$ and $\phi \circ \psi = \mathrm{id}_S$. We introduce 
\begin{align*}
    \phi : T &\to S\\
    (\delta', \delta'', \hat \alpha, \hat \beta) &\mapsto (\underbrace{\hat \alpha + \delta'}_{\alpha'}, \underbrace{\hat \beta + \delta'}_{\beta'}, \underbrace{\hat \alpha + \delta''}_{\alpha''}, \underbrace{\hat \beta + \delta''}_{\beta''})
\end{align*}
and $\psi:S \to T$ which defines $(\delta', \delta'', \hat \alpha, \hat \beta)$ as functions of $(\alpha', \beta', \alpha'', \beta'')$ following Eqs.~\eqref{eq:def-mu'}---\eqref{eq:def-hat-beta''} and \cref{eq:hats-equal}. 

First, we need to show that the maps $\phi$, $\psi$ are well-defined. In the case of $\phi$, the only non-trivial property that one needs to check is
$$\alpha'+\beta'' = \hat \alpha + \hat \beta + \delta' + \delta'' = \beta'+\alpha''.$$
In the case of $\psi$, all the properties (including $\delta'\delta''=0$) have been checked previously, in the main proof. 

We now show $\psi \circ \phi = \mathrm{id}_T$. Let $(\delta', \delta'', \hat \alpha, \hat \beta) \in T$ and denote $(\alpha', \beta', \alpha'', \beta'') := \phi(\delta', \delta'', \hat \alpha, \hat \beta)$, $(\tilde \delta', \tilde \delta'', \hat{\tilde \alpha}, \hat{\tilde \beta}) := \psi(\alpha', \beta', \alpha'', \beta'')$. We use  Eqs.~\eqref{eq:def-mu'}---\eqref{eq:def-hat-beta''} to compute
$$\mu':=\min(\alpha',\beta') = \min(\hat \alpha + \delta', \hat \beta + \delta') = \min(\hat \alpha, \hat \beta) + \delta'.$$
Similarly, $\mu''=\min(\hat \alpha, \hat \beta) + \delta''$ and thus, using now $\delta'\delta''=0$, $\mu = \min(\hat \alpha, \hat \beta)$. We conclude that
$$\tilde \delta' := \mu'-\mu = \delta' \quad \text{ and } \quad \tilde \delta'' := \mu''-\mu = \delta''.$$
Similarly, 
$$\hat{\tilde \alpha} = \alpha' - \tilde \delta' = \hat \alpha \quad \text{ and } \quad \hat{\tilde \beta} = \beta' - \tilde \delta' = \hat \beta,$$
proving $\psi \circ \phi = \mathrm{id}_T$. 

Let us now show $\phi \circ \psi = \mathrm{id}_S$. Consider $(\alpha', \beta', \alpha'', \beta'') \in S$ and set $(\delta', \delta'', \hat \alpha, \hat \beta):=\psi(\alpha', \beta', \alpha'', \beta'')$ and then $(\tilde \alpha', \tilde \beta', \tilde \alpha'', \tilde \beta'') := \phi(\delta', \delta'', \hat \alpha, \hat \beta)$. We have 
$$\tilde \alpha' = \hat \alpha + \delta' = \alpha'$$
and similarly for the other 3 elements; here, we implicitly use the result in \cref{eq:hats-equal}. This proves the claim that the sets $S$ and $T$ are in bijection.

\medskip

Having proved the second claim, we can now index the sum formula for $p(|z|^2)$ by the 4-tuple $(\delta', \delta'', \hat \alpha, \hat \beta)$: 
\begin{align*}
    p(|z|^2) &= \sum_{(\delta', \delta'', \hat \alpha, \hat \beta) \in T} \binom{k}{\hat \alpha + \delta'}\binom{n-k}{\hat \alpha + \delta''}\binom{k}{\hat \beta + \delta'}\binom{n-k}{\hat \beta + \delta''} (|z|^2)^{\hat \alpha + \hat \beta + \delta' + \delta''} t_{\idx(\hat \alpha + \delta'), \idx(\hat \alpha + \delta'')} t_{\idx(\hat \beta + \delta'), \idx(\hat \beta + \delta'')}\\
    &=\!\!\!\!\!\!\!\!\!\!\sum_{\substack{
        \delta',\delta'' \in \mathbb N^d\\
        |\delta'| \leq k, \, |\delta''| \leq n-k\\
        |\delta''| - |\delta'| = n-2k\\
        \delta' \delta'' = 0
    }} \!\!\!\!\!\!\!\!\!\! (|z|^2)^{\delta'+\delta''} \!\!\!\!\!\!\!\!\! \sum_{\substack{\hat \alpha, \hat \beta \in \mathbb N^d\\ |\hat \alpha| = |\hat \beta| = k-|\delta'| = n-k-|\delta''|}}\!\!\!\!\!\!\!\!\!\!\!\!\!\!\!\!\!\!\!\! \binom{k}{\hat \alpha + \delta'}\binom{n-k}{\hat \alpha + \delta''}\binom{k}{\hat \beta + \delta'}\binom{n-k}{\hat \beta + \delta''} (|z|^2)^{\hat \alpha + \hat \beta } t_{\idx(\hat \alpha + \delta'), \idx(\hat \alpha + \delta'')} t_{\idx(\hat \beta + \delta'), \idx(\hat \beta + \delta'')}\\
    &=\!\!\!\!\!\sum_{\substack{
        \delta',\delta'' \in \mathbb N^d\\
        |\delta'| \leq k, \, |\delta''| \leq n-k\\
        |\delta''| - |\delta'| = n-2k\\
        \delta' \delta'' = 0
    }} \!\!\!\!\! (|z|^2)^{\delta'+\delta''}  \left[ \sum_{\substack{\hat \alpha \in \mathbb N^d\\ |\hat \alpha| = k-|\delta'| = n-k-|\delta''|}}\!\!\!\!\!\! \binom{k}{\hat \alpha + \delta'}\binom{n-k}{\hat \alpha + \delta''} (|z|^2)^{\hat \alpha} t_{\idx(\hat \alpha + \delta'), \idx(\hat \alpha + \delta'')} \right]^2\\
    &=\sum_{l=0}^k \sum_{\substack{
        \delta',\delta'' \in \mathbb N^d\\
        |\delta'| = k-l\\|\delta''| = n-k-l\\
        \delta' \delta'' = 0
    }}  (|z|^2)^{\delta'+\delta''}  \left[ \sum_{\substack{\hat \alpha \in \mathbb N^d\\ |\hat \alpha| = l}} \binom{k}{\hat \alpha + \delta'}\binom{n-k}{\hat \alpha + \delta''} (|z|^2)^{\hat \alpha} t_{\idx(\hat \alpha + \delta'), \idx(\hat \alpha + \delta'')} \right]^2.
\end{align*}
Since two polynomials that are identical on non-negative integers are the same, we have proven the result: 
$$p(x) = \sum_{l=0}^k \sum_{\substack{
        \delta',\delta'' \in \mathbb N^d\\
        |\delta'| = k-l\\|\delta''| = n-k-l\\
        \delta' \delta'' = 0
    }}  x^{\overbrace{\delta'+\delta''}^{\text{total degree $n-2l$}}} q_{\delta',\delta''}(x)^2$$
where 
$$q_{\delta',\delta''}(x) = \sum_{\substack{\hat \alpha \in \mathbb N^d\\ |\hat \alpha| = k-|\delta'| = n-k-|\delta''|}} \binom{k}{\hat \alpha + \delta'}\binom{n-k}{\hat \alpha + \delta''} t_{\idx(\hat \alpha + \delta'), \idx(\hat \alpha + \delta'')} x^{\hat \alpha} $$

\bigskip

\noindent\underline     {Proof of the Item 2.}

Assume that $p_{W[O]}(x) = x^{\alpha} \psi^2$ where $\psi \in \homo{d}{k}.$ We expand the polynomial $\psi$ in the basis of monomials $x^{\beta}, \beta \in \mathbb{N}^{(k)}_d.$

$$\psi(x) = \sum_{\beta \in \mathbb{N}^{(k)}_d} c_{\beta} x^{\beta} \implies \psi(x)^2 = \sum_{\beta, \beta' \in \mathbb{N}^{(k)}_d} c_{\beta} c_{\beta'} x^{\beta + \beta'}.$$ We now evaluate the polynomial for $z \odot \overbar{z}$ where $z \in \mathbb{C}^d.$ 

$$p_{W[O]}(z \odot \overbar{z}) = |z|^{2 \alpha} \sum_{\beta, \beta' \in \mathbb{N}^{(k)}_d} c_{\beta} c_{\beta'} |z|^{2\beta + 2\beta'}.$$ Let $\delta, \delta' \in \mathbb{N}^d$ such that $\delta + \delta' = \alpha.$ We then can rewrite the previous expression as 

$$p_{W[O]}(z \odot \overbar{z}) = |\sum_{\beta \in \mathbb{N}^{(k)}_d} c_{\beta} \overbar{z}^{\delta + \beta} z^{\delta' + \beta}|^2.$$ Consider the following vector.

\[
\ket{\phi}
:=
\sum_{\beta \in \mathbb{N}^{(k)}_d}
c_{\beta}
\left(
\ket{1}^{\otimes {\beta_1 + \delta_1}}
\cdots
\ket{d}^{\otimes{\beta_d + \delta_d}}
\otimes
\ket{1}^{\otimes (\beta_1 + \delta'_d)}
\cdots
\ket{d}^{\otimes (\beta_d + \delta'_d)}
\right)
\in (\mathbb{C}^d)^{\otimes n},
\] This allows us to rewrite the term inside the modulus as an inner product with the product vector in $z$ and $\bar{z}$. $$\sum_{\beta \in \mathbb{N}^{(k)}_d} c_{\beta} \overbar{z}^{\delta + \beta} z^{\delta' + \beta} = \braket{z^{\otimes k + |\delta|} \otimes \overbar{z}^{\otimes n-k-|\delta|} \vert \phi}.$$ Therefore, finally, we have the following series of equalities. 
$$p_{W[O]}(z \odot \overbar{z}) = |\sum_{\beta \in \mathbb{N}^{(k)}_d} c_{\beta} \overbar{z}^{\delta + \beta} z^{\delta' + \beta}| = |\braket{z^{\otimes k + |\delta|} \otimes \overbar{z}^{\otimes n-k-|\delta|} \vert \phi}|^2 = \braket{z^{\otimes n} \vert \ketbra{\phi}{\phi}^{\Gamma_{[k + |\delta|]}} \vert z^{\otimes n}}.$$ Since $|\alpha| = n-2k$, for any $l$ satisfying $0 \leq l \leq n-2k,$ there exists a $\delta \leq \alpha$ such that $|\delta| = l$. Defining $l':=k+l$, this also implies that for all $k\leq l' \leq n-k$, there exists a $\ket{\phi} \in (\mathbb{C}^d)^{\otimes n}$ such that 
$$\braket{z^{\otimes n} \vert \operatorname{diag}[W[O]] \vert z^{\otimes n}} = p_{W[O]}(z \odot \overbar{z}) = \braket{z^{\otimes n} \vert  \ketbra{\phi}{\phi}^{\Gamma_{[l']}} \vert z^{\otimes n}}.$$
This proves the claim.

\end{proof}

\section{Proof of qubit PPT criterion}

\label{sec:appendix-proof-qubits} 
We will study the following cone of SOS tensors for $d=2.$

$$\mathsf{SOS}^{(n)}_2:= \{T \mid p_{T}(x) \in \operatorname{cone}\{x^\alpha \psi \mid |\alpha| \in \{n, n-2, n-4 \ldots n-2 \lfloor n/2 \rfloor \}, \psi \in \Sigma^{(n-|\alpha|)}_2\}\}.$$ We will show that some of the terms are superfluous. In particular, we claim that $\alpha$ can be taken to be $\operatorname{max}(\alpha) \leq 1.$ So, for the case of $n$ even, we have, 
$$\mathsf{SOS}^{(n)}_2= \{T \mid p_{T}(x) \in \operatorname{cone}\{x^\alpha \psi \mid \alpha \in \{(1,1), (0,0)\}, \psi \in \Sigma^{(n-|\alpha|)}_2\}\}, $$ and in the case of $n$ odd, we have, 
$$\mathsf{SOS}^{(n)}_2= \{T \mid p_{T}(x) \in \operatorname{cone}\{x^\alpha \psi \mid \alpha \in \{(0,1), (1,0)\}, \psi \in \Sigma^{(n-|\alpha|)}_2\}\}.$$ 
\bigskip

We show the proof for $n$ even. Let $x^{\alpha} \psi$ be an arbitrary term such that $\Sigma^{(n-|\alpha|)}_2.$ Also $|\alpha|$ is even. Let $p : \mathbb{N} \rightarrow \{0,1\}$ denote the parity function. Then, $$x^{\alpha} \psi = x_1^{\alpha_1}  x_2^{\alpha_2} \psi = x_1^{p(\alpha_1)}x_2^{p(\alpha_2)} \underbrace{x_1^{\alpha_1- p(\alpha_1)}  x_2^{\alpha_2 - p(\alpha_2)} \psi}_{\in \Sigma^{n-p(\alpha_1)-p(\alpha_2)}_2}$$

Since $p(\alpha_1)+p(\alpha_2) = p(\alpha_1 +\alpha_2) = 0 \mod 2$, either both are $0$ or both are $1$. Hence, each term is equivalent to $x^{\alpha} \psi$ where $\Sigma^{n-|\alpha|}_2$ where $\alpha = (0,0)$ or $\alpha = (1,1).$ The proof for the odd case can be done similarly. 

By following the essential argument as \cref{thm:dual-SOS}, the argument above implies that the dual of $\mathsf{SOS}^{(n)}_2$ is equivalent to checking positivity of only \emph{two} moment matrices, one for each $\alpha$. 

\bigskip
For the even case, $n=2m$, the explicit moment matrices in terms of the tensors can be obtained as 

\[
M^{(0,0)}(T) = \begin{pmatrix}
T_{\idx(2m,0)} & T_{\idx(2m-1,1)} & T_{\idx(2m-2,2)} & \cdots & T_{\idx(m,m)} \\
T_{\idx(2m-1,1)} & T_{\idx(2m-2,2)} & T_{\idx(2m-3,3)} & \cdots & T_{\idx(m-1,m+1)} \\
T_{\idx(2m-2,2)} & T_{\idx(2m-3,3)} & T_{\idx(2m-4,4)} & \cdots & T_{\idx(m-2,m+2)} \\
\vdots & \vdots & \vdots & \ddots & \vdots \\
T_{\idx(m,m)} & T_{\idx(m-1,m+1)} & T_{\idx(m-2,m+2)} & \cdots & T_{\idx(0,2m)}
\end{pmatrix}.
\]

\bigskip

\[
M^{(1,1)}(T) = \begin{pmatrix}
T_{\idx(2m-1,1)} & T_{\idx(2m-2,2)} & T_{\idx(2m-3,3)} & \cdots & T_{\idx(m,m)} \\
T_{\idx(2m-2,2)} & T_{\idx(2m-3,3)} & T_{\idx(2m-4,4)} & \cdots & T_{\idx(m-1,m+1)} \\
T_{\idx(2m-3,3)} & T_{\idx(2m-4,4)} & T_{\idx(2m-5,5)} & \cdots & T_{\idx(m-2,m+2)} \\
\vdots & \vdots & \vdots & \ddots & \vdots \\
T_{\idx(m,m)} & T_{\idx(m-1,m+1)} & T_{\idx(m-2,m+2)} & \cdots & T_{\idx(1,2m-1)}
\end{pmatrix}.
\]

\bigskip
These matrices of size of order $n$ were also obtained in \cite[Theorem 1]{yu2016separability}. These matrices give necessary and sufficient conditions for the mixture of Dicke state $X$ to be PPT by \cref{thm:ppt-mom}, and hence, in effect, also separable by \cref{thm:no-ppt-in-qubits}. 

\begin{theorem}
    Let $X \in \operatorname{Herm}[\mathsf{DS}^{(2m)}_2]$. Then, $$X \in \mathsf{Sep}^{(2m)}_2 \iff M^{(1,1)}(Q[X]) \succeq 0 \text{ and } M^{(0,0)}(Q[X]) \succeq 0.$$
\end{theorem} Conditions for $n$ odd can be derived analogously to the case above.

\bigskip

\end{document}